\documentclass[runningheads]{llncs}
\usepackage[T1]{fontenc}
\usepackage{graphicx}
\usepackage{amsmath}
\usepackage{amssymb}
\usepackage{hyperref}
\usepackage{cleveref}
\usepackage{color}

\usepackage{xspace}
\usepackage[linesnumbered,ruled,vlined]{algorithm2e}

\SetCommentSty{mycommstyle}
\SetKw{True}{true}
\SetKw{False}{false}

\spnewtheorem{thm}[theorem]{Theorem}{\bfseries}{\itshape}
\spnewtheorem{cor}[theorem]{Corollary}{\bfseries}{\itshape}
\spnewtheorem{lem}[theorem]{Lemma}{\bfseries}{\itshape}
\spnewtheorem{prop}[theorem]{Proposition}{\bfseries}{\itshape}
\spnewtheorem{obs}[theorem]{Observation}{\bfseries}{\itshape}
\spnewtheorem{dfn}[theorem]{Definition}{\bfseries}{\rmfamily}
\spnewtheorem{rem}[theorem]{Remark}{\bfseries}{\rmfamily}
\spnewtheorem{exm}[theorem]{Example}{\bfseries}{\rmfamily}
\spnewtheorem{cnj}[theorem]{Conjecture}{\bfseries}{\rmfamily}

\crefname{dfn}{Definition}{Definitions}
\crefname{thm}{Theorem}{Theorems}
\crefname{cor}{Corollary}{Corollaries}
\crefname{lem}{Lemma}{Lemmata}
\crefname{prop}{Proposition}{Propositions}
\crefname{rem}{Remark}{Remarks}
\crefname{algorithm}{Algorithm}{Algorithms}
\crefname{algocf}{Algorithm}{Algorithms}
\crefname{obs}{Observation}{Observations}

\crefformat{equation}{(#2#1#3)}



\usepackage{mathtools}

\newcommand{\Oh}{\mathcal{O}} 

\newcommand{\SETH}{SETH\xspace}
\newcommand{\CNFSETH}{CNF-SETH\xspace}

\newcommand{\FF}{\mathbb{F}}
\newcommand{\ZZ}{\mathbb{Z}}
\newcommand{\NN}{\mathbb{N}}
\newcommand{\TT}{\mathcal{T}} 
\newcommand{\bag}{\mathbb{B}}

\newcommand{\eps}{\varepsilon}

\newcommand{\DS}{\textsc{Dominating Set}\xspace}

\newcommand{\OCT}{\textsc{Odd Cycle Transversal}\xspace}

\newcommand{\SAT}{\textsc{Satisfiability}\xspace}

\DeclareMathOperator{\lcw}{lin-cw} 
\DeclareMathOperator{\cw}{cw} 
\DeclareMathOperator{\tw}{tw} 
\DeclareMathOperator{\pw}{pw} 
\DeclareMathOperator{\poly}{poly} 

\DeclareMathOperator{\tctw}{tc-tw} 
\DeclareMathOperator{\tcpw}{tc-pw} 
\DeclareMathOperator{\mw}{mw} 
\DeclareMathOperator{\modtw}{mod-tw} 
\DeclareMathOperator{\modpw}{mod-pw} 
\DeclareMathOperator{\ms}{ms} 

\newcommand{\sep}{:}

\newcommand{\cexpr}{\mu}
\newcommand{\tree}{T}

\newcommand{\intro}[1]{{#1}}
\newcommand{\union}{\oplus}
\newcommand{\relab}[2]{\rho_{#1 \rightarrow #2}}
\newcommand{\join}[2]{\eta_{#1,#2}}

\newcommand{\lfct}{\mathtt{lab}}

\newcommand{\formula}{\sigma}
\newcommand{\tassign}{\tau}
\newcommand{\grpsize}{p}
\newcommand{\vgrpsize}{\beta}
\newcommand{\ngrps}{t}
\newcommand{\nvars}{n}
\newcommand{\nclss}{m}
\newcommand{\clss}{q}

\newcommand{\embedding}{\kappa}

\newcommand{\statemap}{\mathbf{state}}
\newcommand{\states}{\mathbf{states}}
\newcommand{\atoms}{\mathbf{atoms}}

\newcommand{\packing}{\mathcal{P}}
\newcommand{\cost}{\mathrm{cost}}

\newcommand{\rvertex}{\hat{r}}
\newcommand{\state}{\mathbf{s}}
\newcommand{\bolda}{\mathbf{a}}

\newcommand{\zero}{\mathbf{0}}
\newcommand{\one}{\mathbf{1}}
\newcommand{\two}{\mathbf{2}}

\newcommand{\all}{\mathbf{A}}

\newcommand{\conn}{\mathtt{conn}}
\newcommand{\sol}{\mathtt{sol}}

\newcommand{\powerset}[1]{\mathcal{P}(#1)}

\newcommand{\algo}{\mathcal{A}}

\newcommand{\tcpartition}{\Pi_{tc}}

\newcommand{\ST}{\textsc{Steiner Tree}\xspace}
\newcommand{\CVC}{\textsc{Connected Vertex Cover}\xspace}

\newcommand{\CDS}{\textsc{Connected Dominating Set}\xspace}
\newcommand{\FVS}{\textsc{Feedback Vertex Set}\xspace}
\newcommand{\IF}{\textsc{Induced Forest}\xspace}
\newcommand{\VC}{\textsc{Vertex Cover}\xspace}
\newcommand{\IS}{\textsc{Independent Set}\xspace}

\newcommand{\cfct}{\mathbf{c}}
\newcommand{\ctarget}{{\overline{c}}}
\newcommand{\budget}{\overline{b}}
\newcommand{\modtarget}{{\overline{m}}}

\newcommand{\terminals}{K}
\newcommand{\terminal}{v}

\newcommand{\wfct}{\mathbf{w}}
\newcommand{\wtarget}{{\overline{w}}}

\newcommand{\family}{\mathcal{F}}
\newcommand{\PP}{\mathbb{P}}

\newcommand{\rsols}{\mathcal{R}}
\newcommand{\sols}{\mathcal{S}}
\newcommand{\csols}{\mathcal{Q}}
\newcommand{\cuts}{\mathcal{C}}
\newcommand{\homcuts}[1]{{\mathcal{C}^{hom}_{#1}}}
\newcommand{\cc}{\mathtt{cc}}

\newcommand{\dpsols}{\mathcal{A}}
\newcommand{\dppoly}{A}
\DeclareMathOperator{\cons}{cons}
\DeclareMathOperator{\newedge}{edge}

\newcommand{\restrict}[1]{{\big|_{#1}}}

\newcommand{\module}{M}
\newcommand{\modpartition}{\Pi_{mod}}
\newcommand{\modprojection}{\pi}
\newcommand{\modfamily}{\mathcal{M}}
\newcommand{\smodfamily}{{\mathcal{M}_s}}
\newcommand{\modtree}{{\mathcal{M}_{\mathrm{tree}}}}
\newcommand{\modint}{{\mathcal{M}_{\mathrm{tree}}^*}}
\newcommand{\primefamily}{{\mathcal{H}_p}}

\newcommand{\children}{\mathtt{children}}

\newcommand{\nsib}{\mathcal{N}_{\mathtt{sib}}} 
\newcommand{\nall}{\mathcal{N}_{\mathtt{all}}}
\newcommand{\qvertex}{{v^q_\module}}
\newcommand{\pmodule}{{\module^{\uparrow}}}
\newcommand{\qgraph}[1]{{G^q_{#1}}}
\newcommand{\pquotient}{{\qgraph{\pmodule}}}
\newcommand{\pproj}{{\modprojection_\pmodule}}
\newcommand{\pprojinv}{{\modprojection^{-1}_\pmodule}}

\newcommand{\DP}{A}

\usepackage{subcaption}
\usepackage{tikz}
\usepackage{tikzit}

\tikzstyle{filled}=[fill=black, draw=black, shape=circle, inner sep=1pt, minimum size=5]
\tikzstyle{big empty}=[fill=white, draw=black, shape=rectangle, inner sep=2pt, minimum size=20pt]
\tikzstyle{small empty}=[fill=white, draw=black, shape=circle, inner sep=0pt, minimum size=20pt, line width=0.75pt]
\tikzstyle{rect}=[fill=black, draw=black, shape=rectangle, inner sep=1pt, minimum size=5pt]
\tikzstyle{big empty red}=[fill=white, draw={rgb,255: red,171; green,0; blue,60}, shape=rectangle, inner sep=2pt, minimum size=20pt]
\tikzstyle{mid empty}=[fill=white, draw=black, shape=rectangle, inner sep=2pt, minimum size=20pt]
\tikzstyle{mid empty red}=[fill=white, draw={rgb,255: red,171; green,0; blue,60}, shape=rectangle, inner sep=2pt, minimum size=20pt]
\tikzstyle{small empty red}=[fill={rgb,255: red,227; green,159; blue,178}, draw={rgb,255: red,171; green,0; blue,60}, shape=circle, inner sep=0pt, minimum size=20pt, line width=0.75pt]
\tikzstyle{red rect}=[fill={rgb,255: red,171; green,0; blue,60}, draw={rgb,255: red,171; green,0; blue,60}, shape=rectangle, minimum size=5pt]
\tikzstyle{big circle}=[fill=white, draw=black, shape=circle, line width=1pt, minimum size=25pt]
\tikzstyle{filled red}=[fill={rgb,255: red,171; green,0; blue,60}, draw={rgb,255: red,171; green,0; blue,60}, shape=circle, inner sep=1pt, minimum size=5pt]
\tikzstyle{state 0}=[fill=white, draw=black, shape=circle, line width=0.75pt, minimum size=10pt]
\tikzstyle{state L}=[fill={rgb,255: red,25; green,61; blue,182}, draw=black, shape=circle, line width=0.75pt, minimum size=10pt]
\tikzstyle{state R}=[fill={rgb,255: red,191; green,0; blue,64}, draw=black, shape=circle, line width=0.75pt, minimum size=10pt]
\tikzstyle{state unknown}=[fill=black, draw=black, shape=circle, line width=0.75pt, minimum size=10pt]
\tikzstyle{tiny empty}=[fill=white, draw=black, shape=circle, inner sep=0pt, minimum size=12pt, line width=0.75pt]
\tikzstyle{tiny filled}=[fill=black, draw=black, shape=circle, inner sep=0pt, minimum size=12pt, line width=0.75pt]
\tikzstyle{wide rectangle}=[fill=white, draw=black, shape=rectangle, inner sep=2 pt, minimum height=20pt, minimum width=30pt, line width=0.75pt]
\tikzstyle{empty 30pt}=[fill=white, draw=black, shape=circle, inner sep=0pt, minimum size=30pt, line width=0.75pt]
\tikzstyle{empty 20pt}=[fill=white, draw=black, shape=circle, inner sep=0pt, minimum size=20pt, line width=0.75pt]
\tikzstyle{filled 10pt}=[fill=black, draw=black, shape=circle, minimum size=10pt]
\tikzstyle{tiny rectangle}=[fill=white, draw=black, shape=rectangle, minimum width=5pt, minimum height=5pt, inner sep=0pt]
\tikzstyle{dash circle}=[fill=white, draw=black, shape=circle, dashed, inner sep=0.5pt]
\tikzstyle{tall}=[fill=white, draw=black, shape=rectangle, minimum width=1cm, minimum height=1.25cm, inner sep=0pt, line width=0.5pt]
\tikzstyle{filled R}=[fill={rgb,255: red,191; green,0; blue,64}, draw=black, shape=circle, minimum size=5pt, inner sep=1pt]
\tikzstyle{filled B}=[fill={rgb,255: red,0; green,0; blue,208}, draw=black, shape=circle, minimum size=5pt, inner sep=1pt]
\tikzstyle{filled G}=[fill={rgb,255: red,0; green,186; blue,0}, draw=black, shape=circle, minimum size=5pt, inner sep=1pt]
\tikzstyle{filled W}=[fill=white, draw=black, shape=circle, minimum size=5pt, inner sep=1pt]
\tikzstyle{big none}=[fill=none, draw=none, minimum size=5pt, inner sep=1pt]
\tikzstyle{state G}=[fill={rgb,255: red,0; green,186; blue,0}, draw=black, shape=circle, line  width=0.75pt, minimum size=10pt]
\tikzstyle{nano}=[fill=black, draw=black, shape=circle, minimum size=3pt, inner sep=0pt]
\tikzstyle{big empty}=[fill=white, draw=black, shape=circle, inner sep=0pt, minimum size=30pt]

\tikzstyle{thick}=[-, line width=1pt]
\tikzstyle{thick arrow}=[->, line width=1pt]
\tikzstyle{packing}=[-, line width=1pt, draw={rgb,255: red,171; green,0; blue,60}]
\tikzstyle{dotted}=[-, line width=1pt, dashed]
\tikzstyle{dashed cyan thick}=[-, dashed, line width=1pt, draw={rgb,255: red,2; green,200; blue,200}]
\tikzstyle{white}=[-, draw=white]
\tikzstyle{del}=[-, double, line width=1pt]
\tikzstyle{thin}=[-, line width=0.5pt]
\tikzstyle{red dashed}=[-, line width=1pt, draw={rgb,255: red,191; green,0; blue,64}, dashed]
\tikzstyle{dark green}=[-, line width=1pt, draw={rgb,255: red,44; green,127; blue,35}]
\tikzstyle{thin green}=[-, line width=0.5pt, draw={rgb,255: red,27; green,125; blue,52}]
\tikzstyle{thin red}=[-, line width=0.5pt, draw={rgb,255: red,181; green,0; blue,3}]
\tikzstyle{thin dashed}=[-, line width=0.5pt, dashed]
\tikzstyle{green fill}=[-, line width=0.5pt, draw={rgb,255: red,27; green,125; blue,52}, fill={rgb,255: red,162; green,255; blue,155}]
\tikzstyle{blue fill}=[-, line width=0.5pt, draw={rgb,255: red,18; green,0; blue,157}, fill={rgb,255: red,116; green,123; blue,255}]
\tikzstyle{thin arrow}=[->]
\tikzstyle{double arrow}=[{-latex}]

\begin{document}
\title{Tight Algorithms for Connectivity Problems Parameterized by Modular-Treewidth\thanks{The first author was partially supported by DFG Emmy Noether-grant (KR 4286/1).}}

\titlerunning{Connectivity Problems Parameterized by Modular-Treewidth}
%
\author{Falko Hegerfeld\inst{1}\orcidID{0000-0003-2125-5048} \and
Stefan Kratsch\inst{1}\orcidID{0000-0002-0193-7239}}
\authorrunning{F. Hegerfeld \and S. Kratsch}
%
\institute{Institut für Informatik, Humboldt-Universität zu Berlin, Germany
\email{\{hegerfeld,kratsch\}@informatik.hu-berlin.de}}
\maketitle              

\begin{abstract}
We study connectivity problems from a fine-grained parameterized perspective. Cygan et al.~(TALG 2022) first obtained algorithms with single-exponential running time $\alpha^{\tw} n^{\Oh(1)}$ for connectivity problems parameterized by treewidth ($\tw$) by introducing the cut-and-count-technique, which reduces the connectivity problems to locally checkable counting problems. In addition, the obtained bases $\alpha$ were proven to be optimal assuming the Strong Exponential-Time Hypothesis (SETH).

As only sparse graphs may admit small treewidth, these results are not applicable to graphs with dense structure. A well-known tool to capture dense structure is the \emph{modular decomposition}, which recursively partitions the graph into \emph{modules} whose members have the same neighborhood outside of the module. Contracting the modules, we obtain a \emph{quotient graph} describing the adjacencies between modules. Measuring the treewidth of the quotient graph yields the parameter \emph{modular-treewidth}, a natural intermediate step between treewidth and clique-width.
While less general than clique-width, modular-treewidth has the advantage that it can be computed as easily as treewidth. 

We obtain the first tight running times for connectivity problems parameterized by modular-treewidth. For some problems the obtained bounds are the same as relative to treewidth, showing that we can deal with a greater generality in input structure at no cost in complexity. We obtain the following randomized algorithms for graphs of modular-treewidth $k$, given an appropriate decomposition:
\begin{itemize}
	\item \ST can be solved in time $3^k n^{\Oh(1)}$,
	\item \CDS can be solved in time $4^k n^{\Oh(1)}$,
	\item \CVC can be solved in time $5^k n^{\Oh(1)}$,
	\item \FVS can be solved in time $5^k n^{\Oh(1)}$.
\end{itemize}
The first two algorithms are tight due to known results and the last two algorithms are complemented by new tight lower bounds under SETH.

\keywords{connectivity  \and modular-treewidth \and tight algorithms}
\end{abstract}

\section{Introduction}

Connectivity constraints are a very natural form of global constraints in the realm of graph problems. We study connectivity problems from a fine-grained parameterized perspective. The starting point is an influential paper of Cygan et al.~\cite{CyganNPPRW22} introducing the cut-and-count-technique which yields randomized algorithms with running time $\Oh^*(\alpha^{\tw})$\footnote{The $\Oh^*$-notation hides polynomial factors in the input size.}, for some constant \emph{base} $\alpha > 1$, for connectivity problems parameterized by \emph{treewidth} ($\tw$). The obtained bases $\alpha$ were proven to be optimal assuming the Strong Exponential-Time Hypothesis\footnote{The hypothesis that for every $\delta < 1$, there is some $\clss$ such that $\clss$-\SAT cannot be solved in time $\Oh(2^{\delta n})$, where $n$ is the number of variables.} (SETH) \cite{CyganNPPRW11arxiv}.

Since dense graphs cannot have small treewidth, the results for treewidth do not help for graphs with dense structure. A well-known tool to capture dense structure is the \emph{modular decomposition} of a graph, which recursively partitions the graph into \emph{modules} whose members have the same neighborhood outside of the module. Contracting these modules, we obtain a \emph{quotient graph} describing the adjacencies between the modules. Having isolated the dense part to the modules, measuring the complexity of the quotient graph by standard graph parameters such as treewidth yields e.g.\ the parameter \emph{modular-treewidth} ($\modtw$), a natural intermediate step between treewidth and clique-width. While modular-treewidth is not as general as clique-width, the algorithms for computing treewidth transfer to modular-treewidth, yielding e.g.\ reasonable constant-factor approximations for modular-treewidth in single-exponential time, whereas for clique-width we are currently only able to obtain approximations with exponential error.

We obtain the first tight running times for connectivity problems parameterized by modular-treewidth. To do so, we lift the algorithms using the cut-and-count-technique from treewidth to modular-treewidth. A crucial observation is that all vertices inside a module will be connected by choosing a single vertex from a neighboring module. In some cases, this observation is strong enough to lift the treewidth-based algorithms to modular-treewidth for free, i.e., the base $\alpha$ of the running time does not increase, showing that we can deal with a greater generality in input structure at no cost in complexity for these problems. 

\begin{thm}[informal]\label{thm:intro_reductions}
There are one-sided error Monte-Carlo algorithms that, given a decomposition witnessing modular-treewidth $k$, can solve
\begin{itemize}
	\item \ST in time $\Oh^*(3^k)$,
	\item \CDS in time $\Oh^*(4^k)$.
\end{itemize}
\end{thm}
These bases are optimal under SETH, by known results of Cygan et al.~\cite{CyganNPPRW11arxiv}.

However, in other cases the interplay of the connectivity constraint and the remaining problem constraints does increase the complexity for modular-treewidth compared to treewidth. In these cases, we provide new algorithms adapting the cut-and-count-technique to this more intricate setting. 

\begin{thm}[informal]\label{thm:intro_algos}
There are one-sided error Monte-Carlo algorithms that, given a decomposition witnessing modular-treewidth $k$, can solve 
\begin{itemize}
 \item \CVC in time $\Oh^*(5^k)$,
 \item \FVS in time $\Oh^*(5^k)$.
\end{itemize}
\end{thm}

Both problems can be solved in time $\Oh^*(3^k)$ parameterized by treewidth~\cite{CyganNPPRW22}. In contrast, \VC (without the connectivity constraint) has complexity $\Oh^*(2^k)$ with respect to treewidth~\cite{LokshtanovMS18} and modular-treewidth simultaneously.

For these latter two problems, we provide new lower bounds to show that the bases are optimal under SETH. However, we do not need the full power of the modular decomposition to prove the lower bounds. The modular decomposition allows for \emph{recursive} partitioning, when instead allowing for only a single level of partitioning and limited complexity inside the modules, we obtain parameters called \emph{twinclass-pathwidth} ($\tcpw$) and \emph{twinclass-treewidth}. 

\begin{thm}\label{thm:intro_lower_bounds}
Unless SETH fails, the following statements hold for any $\eps > 0$:
\begin{itemize}
 \item \CVC cannot be solved in time $\Oh^*((5-\eps)^{\tcpw})$.
 \item \FVS cannot be solved in time $\Oh^*((5-\eps)^{\tcpw})$.
\end{itemize}
\end{thm}

The obtained results on connectivity problems parameterized by modular-treewidth are situated in the larger context of a research program aimed at determining the optimal running times for connectivity problems relative to width-parameters of differing generality, thus quantifying the price of generality in this setting. The known results are summarized in \cref{table:conn_time_overview}. Beyond the results for treewidth by Cygan et al.~\cite{CyganNPPRW11arxiv,CyganNPPRW22}, Bojikian et al.~\cite{BojikianCHK23} obtain tight results for the more restrictive \emph{cutwidth} by either providing faster algorithms resulting from combining cut-and-count with the rank-based approach or by showing that the same lower bounds already hold for cutwidth. Hegerfeld and Kratsch~\cite{HegerfeldK23cw} consider \emph{clique-width} and obtain tight results for \CVC and \CDS. Their algorithms combine cut-and-count with several nontrivial techniques to speed up dynamic programming on clique-expressions, where the interaction between cut-and-count and clique-width can yield more involved states compared to modular-treewidth, as clique-width is more general. These algorithms are complemented by new lower bound constructions following similar high-level principles as for modular-treewidth, but allow for more flexibility in the gadget design due to the mentioned generality. However, the techniques of Hegerfeld and Kratsch~\cite{HegerfeldK23cw} for clique-width yield tight results for fewer problems compared to the present work; in particular, the optimal bases for \ST and \FVS parameterized by clique-width are currently not known.

	\newcommand{\sexp}[1]{$\Oh^*(#1^k)$}%
	\begin{table}%
		\centering
		\begin{tabular}{l|cccc}%
			Parameters & cutwidth & treewidth & modular-tw & clique-width\\%
			\hline%
			\CVC & \sexp{2} & \sexp{3} & \sexp{5} & \sexp{6} \\%
			\CDS & \sexp{3} & \sexp{4}	& \sexp{4} & \sexp{5} \\%
			\ST  & \sexp{3} & \sexp{3} & \sexp{3} & ? \\%
			\FVS & \sexp{2} & \sexp{3} & \sexp{5} & ? \\%
			\hline%
			References & \cite{BojikianCHK23} & \cite{CyganNPPRW11arxiv,CyganNPPRW22} & here & \cite{HegerfeldK23cw}%
		\end{tabular}%
		\caption{Optimal running times of connectivity problems with respect to various width-parameters listed in increasing generality. The results in the penultimate column are obtained in this paper. The ``?'' denotes cases, where an algorithm with single-exponential running time is known by Bergougnoux and Kant\'e~\cite{BergougnouxK19a}, but a gap between the lower bound and algorithm remains.}\label{table:conn_time_overview}%
		\vspace*{-1cm}
	\end{table}%

\subsubsection*{Related work.}
We survey some more of the literature on parameterized algorithms for connectivity problems relative to dense width-parameters. Bergougnoux~\cite{Bergougnoux19} has applied cut-and-count to several width-parameters based on structured neighborhoods such as clique-width, rank-width, or mim-width. Building upon the rank-based approach of Bodlaender et al.~\cite{BodlaenderCKN15}, Bergougnoux and Kant\'e~\cite{BergougnouxK19a} obtain single-exponential running times $\Oh^*(\alpha^{\cw})$ for a large class of connectivity problems parameterized by clique-width ($\cw$). The same authors~\cite{BergougnouxK21} also generalize this approach to other dense width-parameters via structured neighborhoods. All these works deal with general \textsc{Connected $(\sigma, \rho)$-Dominating Set} problems capturing a wide range of problems; this generality of problems (and parameters) comes at the cost of yielding running times that are far from optimal for specific problem-parameter-combinations, e.g., the first article~\cite{Bergougnoux19} is the most optimized for clique-width and obtains the running time $\Oh^*((2^{4 + \omega})^{\cw}) \geq \Oh^*(64^{\cw})$, where $\omega$ is the matrix multiplication exponent~\cite{AlmanW21}, for \CDS. Bergougnoux et al.~\cite{BergougnouxDJ23} obtain XP algorithms parameterized by mim-width for problems expressible in a logic that can also capture connectivity constraints. Beyond dense width-parameters, cut-and-count has also been applied to the parameters branchwidth~\cite{PinoBR16} and treedepth~\cite{HegerfeldK20,NederlofPSW20}.

Our version of modular-treewidth was first used by Bodlaender and Jansen for \textsc{Maximum Cut}~\cite{BodlaenderJ00}. Several papers~\cite{Lampis20,Mengel16,PaulusmaSS16} also use the name modular-treewidth, but use it to refer to what we call \emph{twinclass-treewidth}. In particular, Lampis~\cite{Lampis20} obtains tight results under SETH for \textsc{$q$-Coloring} with respect to twinclass-treewidth and clique-width. Hegerfeld and Kratsch~\cite{HegerfeldK22} obtain tight results for \OCT parameterized by twinclass-pathwidth and clique-width and for \DS parameterized by twinclass-cutwidth. Kratsch and Nelles~\cite{KratschN22} combine modular decompositions with tree-depth in various ways and obtain parameterized algorithms for various efficiently solvable problems.

\subsubsection*{Organization.}
In \cref{sec:modtw_prelims} we discuss the general preliminaries and \cref{sec:modtw_cutandcount} the cut-and-count-technique. We prove \cref{thm:intro_reductions} in \cref{sec:modtw_reduction}. \Cref{sec:modtw_cvc_algo} contains the \CVC algorithm of \cref{thm:intro_algos} and \cref{sec:modtw_fvs_algo} contains the \FVS algorithm. \Cref{sec:modtw_cvc_lb} contains the \CVC lower bound of \cref{thm:intro_reductions} and \cref{sec:modtw_fvs_lb} the \FVS lower bound. \Cref{sec:modtw_vc_algo} contains an algorithm for \VC used as a subroutine. The problem definitions can be found in \cref{sec:problems}.

\section{Preliminaries}
\label{sec:modtw_prelims}

For two integers $a,b$ we write $a \equiv_c b$ to indicate equality modulo $c \in \NN$. We use Iverson's bracket notation: for a boolean predicate $p$, we have that $[p]$ is $1$ if $p$ is true and $0$ otherwise. For a function $f$ we denote by $f[v \mapsto \alpha]$ the function $(f \setminus \{(v, f(v))\}) \cup \{(v, \alpha)\}$, viewing $f$ as a set. By $\FF_2$ we denote the field of two elements. For $n_1, n_2 \in \ZZ$, we write $[n_1, n_2] = \{x \in \ZZ \sep n_1 \leq x \leq n_2\}$ and $[n_2] = [1, n_2]$. For a function $f\colon V \rightarrow \ZZ$ and a subset $W \subseteq V$, we write $f(W) = \sum_{v \in W} f(v)$. Note that for functions $g\colon A \rightarrow B$, where $B \not\subseteq \ZZ$, and a subset $A' \subseteq B$, we still denote the \emph{image of $A'$ under $g$} by $g(A') = \{g(v) \sep v \in A'\}$. If $f \colon A \rightarrow B$ is a function and $A' \subseteq A$, then $f\big|_{A'}$ denotes the \emph{restriction} of $f$ to $A'$ and for a subset $B' \subseteq B$, we denote the \emph{preimage of $B'$ under $f$} by $f^{-1}(B') = \{a \in A \sep f(a) \in B'\}$. The \emph{power set} of a set $A$ is denoted by $\powerset{A}$.

\subsubsection*{Graph Notation.}
We use common graph-theoretic notation and the essentials of parameterized complexity. Let $G = (V, E)$ be an undirected graph. For $X \subseteq V$, we denote by $G[X]$ the subgraph of $G$ induced by $X$. The \emph{open neighborhood} of $v \in V$ is given by $N_G(v) = \{u \in V \sep \{u,v\} \in E\}$, whereas the \emph{closed neighborhood} is given by $N_G[v] = N_G(v) \cup \{v\}$. For $X \subseteq V$, we define $N_G[X] = \bigcup_{v \in X} N_G[v]$ and $N_G(X) = N_G[X] \setminus X$. The degree of $v \in V$ is denoted $\deg_G(v) = |N_G(v)|$. For two disjoint vertex subsets $A, B \subseteq V$, we define $E_G(A,B) = \{\{a,b\} \in E(G) \sep a \in A, b \in B\}$ and adding a \emph{join} between $A$ and $B$ means adding an edge between every vertex in $A$ and every vertex in $B$. We denote the \emph{number of connected components} of $G$ by $\cc(G)$. A \emph{cut} of $G$ is a partition $V = V_L \cup V_R$, $V_L \cap V_R = \emptyset$, of its vertices into two parts.

\subsubsection*{Tree Decompositions.} A \emph{path/tree decomposition} of a graph $G=(V,E)$ is a pair $(\TT, (\bag_t)_{t \in V(\TT)})$, where $\TT$ is a path/tree and every \emph{bag} $\bag_t \subseteq V$, $t \in V(\TT)$, is a set of vertices such that the following properties are satisfied:
  \begin{itemize}
   \item every vertex $v \in V$ is contained in some bag $v \in \bag_t$, 
   \item every edge $\{v,w\} \in E$ is contained in some bag $\{u,v\} \subseteq \bag_t$,
   \item for every vertex $v$, the set $\{t \in V(\TT) \sep v \in \bag_t\}$ is connected in $\TT$.
  \end{itemize}
  The \emph{width} of a path/tree decomposition $(\TT, (\bag_t)_{t \in
    V(\TT)})$ is $\max_{t \in V(\TT)} |\bag_t| - 1$. The \emph{pathwidth/treewidth} of a graph $G$, denoted $\pw(G)$ or $\tw(G)$ respectively, is the minimum width of a path/tree decomposition of $G$. For dynamic programming algorithms on tree decompositions, it is convenient to use \emph{very nice tree decompositions}~\cite{CyganNPPRW22}, further refining the \emph{nice tree decompositions} of Kloks~\cite{Kloks94}.

\begin{dfn}
  A tree decomposition $(\TT, (\bag_t)_{t \in V(\TT)})$ is \emph{very nice} if it is rooted at the \emph{root node} $\rvertex \in V(\TT)$ with $\bag_{\rvertex} = \emptyset$ and every bag $\bag_t$ has one of the following types:
  \begin{itemize}
    \item \textbf{Leaf bag:} $t$ has no children and $\bag_t = \emptyset$.
    \item \textbf{Introduce vertex $v$ bag:} $t$ has exactly one child $t'$ and $\bag_t = \bag_{t'} \cup \{v\}$ with $v \notin \bag_{t'}$.
    \item \textbf{Forget vertex $v$ bag:} $t$ has one child $t'$ and $\bag_t = \bag_{t'} \setminus \{v\}$ with $v \in \bag_{t'}$.
    \item \textbf{Introduce edge $\{v,w\}$ bag:} $t$ is labeled with an edge $\{v,w\} \in E$ and $t$ has one child $t'$ which satisfies $\{v,w\} \subseteq \bag_t = \bag_{t'}$.
    \item \textbf{Join bag:} $t$ has exactly two children $t_1$ and $t_2$ with $\bag_t = \bag_{t_1} = \bag_{t_2}$.
  \end{itemize}
  Furthermore, we require that every edge in $E$ is introduced exactly once.
\end{dfn}

\begin{lem}[\cite{CyganNPPRW22}]\label{thm:very_nice_tree_decomposition}
  Any tree decomposition of $G$ can be converted into a very nice tree decomposition of $G$ with the same width in polynomial time.
\end{lem}

\subsubsection*{Quotients and Twins.} 
Let $\Pi$ be a partition of $V(G)$. The \emph{quotient graph} $G / \Pi$ is given by $V(G / \Pi) = \Pi$ and $E(G / \Pi) = \{\{B_1, B_2\} \subseteq \Pi \sep B_1 \neq B_2, \exists u \in B_1, v \in B_2 \colon \{u,v\} \in E(G)\}$. We say that two vertices $u, v$ are \emph{twins} if $N(u) \setminus \{v\} = N(v) \setminus \{u\}$. The equivalence classes of this relation are called \emph{twinclasses} and we let $\tcpartition(G)$ denote the partition of $V(G)$ into twinclasses. If $N(u) = N(v)$, then $u$ and $v$ are \emph{false twins} and if $N[u] = N[v]$, then $u$ and $v$ are \emph{true twins}. Every twinclass of size at least 2 consists of only false twins or only true twins. A false twinclass induces an independent set and a true twinclass induces a clique. 

\subsubsection*{Lifting to Twinclasses.} 
The \emph{twinclass-treewidth} and \emph{twinclass-pathwidth} of $G$ are defined by $\tctw(G) = \tw(G/\tcpartition(G))$ and $\tcpw(G) = \pw(G/\tcpartition(G))$, respectively. The parameters twinclass-treewidth and twinclass-pathwidth have been considered before under the name modular treewidth and modular pathwidth~\cite{Lampis20,Mengel16,PaulusmaSS16}. We use the prefix \emph{twinclass} instead of \emph{modular} to distinguish from the quotient graph arising from a \emph{modular partition} of $G$. 

\subsubsection*{Modular Decomposition.} 
A vertex set $\module \subseteq V(G)$ is a \emph{module} of $G$ if $N(v) \setminus \module = N(w) \setminus \module$ for every pair $v, w \in \module$ of vertices in $\module$. Equivalently, for every $u \in V(G) \setminus \module$ it holds that $\module \subseteq N(u)$ or $\module \cap N(u) = \emptyset$. In particular, every twinclass is a module. We let $\modfamily(G)$ denote the set of all modules of $G$. The modules $\emptyset$, $V(G)$, and all singletons are called \emph{trivial}. A graph that only admits trivial modules is called \emph{prime}. If $\module \neq V(G)$, then we say that $\module$ is \emph{proper}. For two disjoint modules $\module_1, \module_2 \in \modfamily(G)$, either $\{\{v,w\} \sep v \in \module_1, w \in \module_2\} \subseteq E(G)$ or $\{\{v,w\} \sep v \in \module_1, w \in \module_2\} \cap E(G) = \emptyset$; in the first case, $\module_1$ and $\module_2$ are \emph{adjacent} and in the second case, they are \emph{nonadjacent}.

A module $\module$ is \emph{strong} if for every module $\module' \in \modfamily(G)$ we have that $\module \cap \module' = \emptyset$, $\module \subseteq \module'$, or $\module' \subseteq \module$, so strong modules intersect other modules only in a trivial way. Let $\smodfamily(G)$ denote the set of all strong modules of $G$. The defining property of strong modules implies that $\smodfamily(G)$ is a \emph{laminar set family}. Hence, if we consider $\modtree(G) = \smodfamily(G) \setminus \{\emptyset\}$ with the inclusion-relation, the associated Hasse diagram, i.e., there is an edge from $\module_1 \in \modtree(G)$ to $\module_2 \in \modtree(G)$ if $\module_1 \subsetneq \module_2$ and there is no $\module_3 \in \modtree(G)$ with $\module_1 \subsetneq \module_3 \subsetneq \module_2$, is a rooted tree, called the \emph{modular decomposition (tree)} of $G$. We freely switch between viewing $\modtree(G)$ as a set family or as the modular decomposition tree of $G$. In the latter case, we usually speak of \emph{nodes} of the modular decomposition tree.

Every graph $G$ with at least two vertices can be uniquely partitioned into a set of inclusion-maximal non-trivial strong modules $\modpartition(G) = \{\module_1, \ldots, \module_\ell\}$, with $\ell \geq 2$, called \emph{canonical modular partition}. For $\module \in \modtree(G)$ with $|\module| \geq 2$, let $\children(\module) = \modpartition(G[\module])$ as the sets in $\modpartition(G[\module])$ are precisely the children of $\module$ in the modular decomposition tree; if $|\module| = 1$, then $\children(\module) = \emptyset$. We write $\modint(G) = \modtree(G) \setminus \{\{v\} \sep v \in V\}$. Forming the \emph{quotient graph $\qgraph{\module} = G[\module] / \modpartition(G[\module])$ at $\module \in \modint(G)$}, there are three cases: 
\begin{thm}[\cite{Gallai67}]\label{thm:gallai_modular}
	For $\module \in \modint(G)$, exactly one of the following holds:
	\begin{itemize}
		\item \textbf{Parallel node}: $G[\module]$ is not connected and $\qgraph{\module}$ is an independent set,
		\item \textbf{Series node}: the complement $\overline{G[\module]}$ is not connected and $\qgraph{\module}$ is a clique,
		\item \textbf{Prime node}: $\modpartition(G[\module])$ consists of the inclusion-maximal proper modules of $G[\module]$ and $\qgraph{\module}$ is prime. 
	\end{itemize}
\end{thm}
We collect the graphs that appear as prime quotient graphs in the modular decomposition of $G$ in the family $\primefamily(G) = \{\qgraph{\module} \sep \module \in \modint(G), \text{$\qgraph{\module}$ is prime}\}$. The modular decomposition tree can be computed in time $\Oh(n+m)$, see e.g.\ Tedder et al.~\cite{TedderCHP08} or the survey by Habib and Paul~\cite{HabibP10}.

Let $\module \in \modtree(G) \setminus \{V\}$ and $\pmodule \in \modtree(G)$ be its \emph{parent module}. We have that $\module \in \modpartition(G[\pmodule])$, hence $\module$ appears as a vertex of the quotient graph $\pquotient$; we will also denote this vertex by $\qvertex$. Note that $\pquotient$ is the only quotient graph in the modular decomposition of $G$ where $\module$ appears as a vertex. So, we implicitly know that $\qvertex \in V(\pquotient)$ without having to specify $\pmodule$. To each quotient graph $\qgraph{\pmodule} = G[\pmodule] / \modpartition(G[\pmodule])$, $\pmodule \in \modint(G)$, appearing in the modular decomposition, we also associate a \emph{canonical projection} $\pproj \colon \pmodule \rightarrow V(\qgraph{\pmodule})$ with $\pproj(v) = \qvertex$ whenever $v \in \module \in \modpartition(G[\pmodule])$.

\subsubsection*{Lifting to Modules.}
Many graph problems can be solved by working only on $\primefamily(G)$. Hence, we consider the values of standard graph parameters on $\primefamily(G)$. We define the \emph{modular-width} of $G$ by $\mw(G) = \max(2, \max_{H \in \primefamily(G)} |V(H)|)$, the \emph{modular-pathwidth} by $\modpw(G) = \max(2, \max_{H \in \primefamily(G)} \pw(H))$, and the \emph{modular-treewidth} by $\modtw(G) = \max(2, \max_{H \in \primefamily(G)} \tw(H))$. By combining an algorithm to compute the modular decomposition tree with an algorithm to compute treewidth, we obtain the following. 

\begin{thm}\label{thm:compute_modtw_transfer}
  If $\algo_{\tw}$ is an algorithm that given an $n$-vertex graph $G$ and an integer $k$, in time $\Oh(f(k)n^c)$, $c \geq 1$, either outputs a tree decomposition of width at most $g(k)$ or determines that $\tw(G) > k$, then there is an algorithm $\algo_{\modtw}$ that given an $n$-vertex $m$-edge graph $G$ and an integer $k$, in time $\Oh(f(k)n^c + m)$ either outputs a tree decomposition of width at most $g(k)$ for every prime quotient graph $\qgraph{\module} \in \primefamily(G)$ or determines that $\modtw(G) > k$.
\end{thm}

\begin{proof}
  The algorithm $\algo_{\modtw}$ works as follows. We first compute the modular decomposition tree of $G$ in time $\Oh(n + m)$ with, e.g., the algorithm of Tedder et al.~\cite{TedderCHP08} and obtain the family of prime quotient graphs $\primefamily(G)$. Since the modular decomposition tree has $n$ leaves and every internal node has at least two children, we obtain that $|\modtree(G)| \leq 2n$. This also implies that $\sum_{H \in \primefamily(G)} |V(H)| \leq 2n$, since the vertices of the quotient graph $G^q_\module$ at $\module \in \modint(G)$ are precisely the children of $\module$ in the modular decomposition tree. We run $\algo_{\tw}$ on every $H \in \primefamily(G)$ and bound the running time, neglecting the constant term, of this step as follows:
  \begin{equation*}
    \sum_{H \in \primefamily(G)} f(k) |V(H)|^c \leq f(k) n^{c-1} \sum_{\mathclap{H \in \primefamily(G)}} |V(H)| \leq 2 f(k) |V(H)|^{c}
  \end{equation*} 
  The algorithm is clearly correct, so this concludes the proof. \qed
\end{proof}

\begin{cor}
  There is an algorithm, that given an $n$-vertex graph $G$ and an integer $k$, in time $2^{\Oh(k)}n + m$ either outputs a tree decomposition of width at most $2k+1$ for every prime quotient graph $\qgraph{\module} \in \primefamily(G)$ or determines that $\modtw(G) > k$.
\end{cor}

\begin{proof}
  We apply \cref{thm:compute_modtw_transfer} with the algorithm of Korhonen~\cite{Korhonen21} that satisfies $f(k) = 2^{\Oh(k)}$ and $g(k) = 2k+1$. \qed
\end{proof}

\subsubsection*{Associated Subgraphs for Modular-Treewidth.} Given a very nice tree decomposition $(\TT^q_{\pmodule}, (\bag^q_t)_{t \in V(\TT^q_{\pmodule})})$ of the quotient graph $\qgraph{\pmodule}$, we associate to every node $t \in V(\TT^q_{\pmodule})$ a subgraph $G^q_t =(V^q_t, E^q_t)$ of $\qgraph{\pmodule}$ as follows: 
\begin{itemize}
 \item $V^q_t$ contains all $\qvertex \in V(\qgraph{\pmodule})$ such that there is a descendant $t'$ of $t$ in $\TT^q_{\pmodule}$ with $\qvertex \in \bag^q_{t'}$,
 \item $E^q_t$ contains all $\{v^q_{\module_1}, v^q_{\module_2}\} \in E(\qgraph{\pmodule})$ that were introduced by a descendant of $t$ in $\TT^q_{\pmodule}$.
\end{itemize}
Based on the vertex subsets of the quotient graph $\qgraph{\pmodule}$, we define vertex subsets of the original graph $G[\pmodule]$ as follows: $\bag_t = \pprojinv(\bag^q_t) = \bigcup_{v^q_\module \in \bag^q_t} \module$ and $V_t = \pprojinv(V^q_t) = \bigcup_{v^q_\module \in V^q_t} \module$. We also transfer the edge set as follows
\begin{equation*}
	E_t = \bigcup_{\qvertex \in V^q_t} E(G[\module]) \cup \smashoperator[r]{\bigcup_{\{v^q_{\module_1}, v^q_{\module_2}\} \in E^q_t}} \{\{u_1, u_2\} \sep u_1 \in \module_1 \wedge u_2 \in \module_2\},
\end{equation*} 
allowing us to define the graph $G_t = (V_t, E_t)$ associated to any node $t \in V(\TT^q_\pmodule)$.

\subsubsection*{Clique-Expressions and Clique-Width.}
A \emph{labeled graph} is a graph $G = (V,E)$ together with a \emph{label function} $\lfct \colon V \rightarrow \NN = \{1, 2, 3, \ldots\}$; we usually omit mentioning $\lfct$ explicitly. A labeled graph is \emph{$k$-labeled} if $\lfct(v) \leq k$ for all $v \in V$. We consider the following four operations on labeled graphs: 
\begin{itemize}
  \item the \emph{introduce}-operation $\intro{\ell}(v)$ which constructs a single-vertex graph whose unique vertex $v$ has label $\ell$,
  \item the \emph{union}-operation $G_1 \union G_2$ which constructs the disjoint union of two labeled graphs $G_1$ and $G_2$,
  \item the \emph{relabel}-operation $\relab{i}{j}(G)$ changes the label of all vertices in $G$ with label $i$ to label $j$, 
  \item the \emph{join}-operation $\join{i}{j}(G)$, $i \neq j$, which adds an edge between every vertex in $G$ with label $i$ and every vertex in $G$ with label $j$. 
\end{itemize}
A valid expression that only consists of introduce-, union-, relabel-, and join-operations is called a \emph{clique-expression}. The graph constructed by a clique-expression $\cexpr$ is denoted $G_\cexpr$ and the label function is denoted $\lfct_\cexpr \colon V(G_\cexpr) \rightarrow \NN$. We associate to a clique-expression $\cexpr$ the syntax tree $\tree_\cexpr$ in the natural way and to each node $t \in V(\tree_\cexpr)$ the corresponding operation. For any node $t \in V(\tree_\cexpr)$ the subtree rooted at $t$ induces a \emph{subexpression} $\cexpr_t$. When a clique-expression $\cexpr$ is fixed, we define $G_t = G_{\cexpr_t}$ and $\lfct_t = \lfct_{\cexpr_t}$ for any $v \in V(\tree_\cexpr)$. We say that a clique-expression $\cexpr$ is a \emph{$k$-clique-expression} or just \emph{$k$-expression} if $(G_t, \lfct_t)$ is $k$-labeled for all $t \in V(\tree_\cexpr)$. The \emph{clique-width} of a graph $G$, denoted by $\cw(G)$, is the minimum $k$ such that there exists a $k$-expression $\cexpr$ with $G = G_\cexpr$. A clique-expression $\cexpr$ is \emph{linear} if in every union-operation the second graph consists only of a single vertex. Accordingly, we also define the \emph{linear-clique-width} of a graph $G$, denoted $\lcw(G)$, by only considering linear clique-expressions.

\subsubsection*{Strong Exponential-Time Hypothesis.}
The \emph{Strong Exponential-Time Hypothesis} (\SETH) \cite{CalabroIP09,ImpagliazzoPZ01} concerns the complexity of $\clss$-\SAT, i.e., \SAT where every clause contains at most $\clss$ literals. We define $c_\clss = \inf \{\delta \sep \clss\text{-\SAT can be solved in time } \Oh(2^{\delta \nvars}) \}$ for all $\clss \geq 3$. The \emph{Exponential-Time} \emph{Hypothesis} (ETH) of Impagliazzo and Paturi~\cite{ImpagliazzoP01} posits that $c_3 > 0$, whereas the Strong Exponential-Time Hypothesis states that $\lim_{\clss \rightarrow \infty} c_\clss = 1$. 
Or equivalently, for every $\delta < 1$, there is some $\clss$ such that $\clss$-\SAT cannot be solved in time $\Oh(2^{\delta \nvars})$.
For one of our lower bounds, the following weaker variant of \SETH, also called \CNFSETH, is sufficient.

\begin{cnj}[\CNFSETH]
  \label{conj:cnfseth}
  For every $\eps > 0$, there is no algorithm solving \SAT with $n$ variables and $m$ clauses in time $\Oh(\poly(m)(2-\eps)^n)$.
\end{cnj}

\subsection{Parameter Relationships}

\begin{lem}\label{thm:cw_modpw}
  For any graph $G$, we have $\cw(G) \leq \modpw(G) + 2$. An appropriate clique-expression can be computed in polynomial time given optimal path decompositions of the graphs in $\primefamily(G)$.
\end{lem}

\begin{proof}
  We construct a clique-expression $\cexpr$ for $G$ using at most $\modpw(G) + 2$ labels by working bottom-up along the modular decomposition tree. More precisely, we inductively construct $(\modpw(G)+2)$-expressions $\cexpr_\module$ for every $G[\module]$, $\module \in \modtree(G)$. 
  
  As the base case, we consider the leaves of the modular decomposition tree which correspond to singleton modules $\{v\}$, $v \in V$, and therefore each $\cexpr_{\{v\}}$ simply consists of a single introduce-operation. For any internal node $\module$ of the modular decomposition tree with $\modpartition(G[\module]) = \{\module_1, \ldots, \module_\ell\}$, we inductively assume that the clique-expressions $\cexpr_i := \cexpr_{\module_i}$ for $G[\module_i]$, $i \in [\ell]$, have already been constructed. Furthermore, we assume without loss of generality that every $\cexpr_i$ relabels all vertices to label $1$ at the end. We now distinguish between the node type of $\module$ in the modular decomposition tree. If $\module$ is a parallel node, then we obtain $\cexpr_\module$ by successively taking the union of all $\cexpr_i$, $i \in [\ell]$. 
  
  If $\module$ is a series node, then we set $\cexpr'_{1} := \cexpr_1$ and $\cexpr'_{i+1} := \relab{2}{1}(\join{1}{2}(\cexpr'_{i} \union \relab{1}{2}(\cexpr_{i+1})))$ for all $i \in [\ell-1]$ and $\cexpr_\module = \cexpr'_\ell$. So, we add one child module after the other and add all edges to the previous child modules using two labels.
  
  If $\module$ is a prime node, then we consider an optimal path decomposition $(\TT^q, (\bag^q_t)_{t \in V(\TT^q)})$ of the quotient graph $\qgraph{\module} = G[\module] / \modpartition(G[\module])$. By \cref{thm:very_nice_tree_decomposition}, we can assume that it is a very nice path decomposition. We inductively construct clique-expressions $\cexpr'_t$ for every $t \in V(\TT^q)$ such that every module in the current bag has a private label and all forgotten modules get label $\ell_{max} := \modpw(G) + 2$. Since every bag contains at most $\modpw(G) + 1$ modules, all smaller labels may be used as private labels. If $\rvertex$ denotes the root node of $\TT^q$, then we set $\cexpr_\module = \cexpr'_{\rvertex}$. The base case is given by the leaf node with $\bag^q_t = \emptyset$, where $\cexpr'_t$ is simply the empty expression.
  
  For an introduce vertex node $t$ introducing vertex $v^q_{\module_i}$, with child $s$, let $\ell_i$ denote the smallest empty label at the end of $\cexpr'_{s}$ and set $\cexpr'_t = \cexpr'_{s} \union \relab{1}{\ell_i}(\cexpr_i)$.
  
  For an introduce edge node $t$ introducing edge $\{v^q_{\module_i}, v^q_{\module_j}\}$, with child $s$, let $\ell_i$ and $\ell_j$ denote the labels of $\module_i$ and $\module_j$ respectively in $\cexpr'_{s}$ and set $\cexpr'_t = \join{\ell_i}{\ell_j}(\cexpr'_{s})$.
  
  For a forget vertex node $t$, which forgets vertex $v^q_{\module_i}$, with child $s$, we let $\ell_i$ denote the label of $\module_i$ in $\cexpr'_{s}$ and set $\cexpr'_t = \relab{\ell_i}{\ell_{max}}(\cexpr'_{s})$. \qed
\end{proof}

Note that \cref{thm:cw_modpw} can only hold for modular-pathwidth and not modular-treewidth, as already for treewidth, Corneil and Rotics~\cite{CorneilR05} show that for every $k$ there exists a graph $G_k$ with treewidth $k$ and clique-width exponential in $k$.

\begin{lem}
  For any graph $G$, we have $\modpw(G) \leq \max(2, \tcpw(G))$ and $\modtw(G) \leq \max(2, \tctw(G))$.
\end{lem}

\begin{proof}
  Since parallel and series nodes do not affect $\modpw(G)$ or $\modtw(G)$, it is sufficient to consider the prime nodes. Let $G[\module]$, $\module \in \modint(G)$, be some internal prime node in the modular decomposition tree of $G$. We want to show that $\pw(\qgraph{\module}) = \pw(G[\module] / \modpartition(G[\module])) \leq \pw(G / \tcpartition(G)) = \tcpw(G)$ and similarly for the treewidth. We claim that $\qgraph{\module}$ is a subgraph of $G / \tcpartition(G)$ which implies the desired inequalities.
  
  Since $\module$ is a module, we see that the twinclasses of $G[\module]$ have the form $C \cap \module$, where $C$ is a twinclass of $G$. Therefore, the graph $G[\module] / \tcpartition(G[\module])$ is an induced subgraph of $G / \tcpartition(G)$. Furthermore, every proper twinclass of $G[\module]$ is also a proper module of $G[\module]$. By \cref{thm:gallai_modular}, $\modpartition(G[\module])$ must consist of all inclusion-maximal proper modules of $G[\module]$. Thus, $\tcpartition(G[\module])$ is a finer partition than $\modpartition(G[\module])$ and $\qgraph{\module} = G[\module] / \modpartition(G[\module])$ is an induced subgraph of $G[\module] / \tcpartition(G[\module])$ which shows our claim. \qed
\end{proof}

\begin{thm}[\cite{HegerfeldK22}]
  \label{thm:hierarchy_cliquewidth}
  For a graph $G$, we have $\cw(G) \leq \lcw(G) \leq \tcpw(G) + 4 \leq \pw(G) + 4$.
\end{thm}

\subsubsection*{Mixed-search.} To prove that the graphs in our lower bound constructions have small pathwidth, it is easier to use a \emph{search game} characterization instead of directly constructing a path decomposition. The search game corresponding to pathwidth is the \emph{mixed-search-game}. In such a game, the graph $G$ represents a system of tunnels where all edges are contaminated by a gas. The objective is to clear all edges of this gas. An edge can be cleared by either placing searchers at both of its endpoints or by moving a searcher along the edge. If there is a path from an uncleared edge to a cleared edge without any searchers on the vertices or edges of the path, then the cleared edge is recontaminated. A \emph{search strategy} is a sequence of operations of the following types: a searcher can be placed on or removed from a vertex, and a searcher on a vertex can be moved along an incident edge and placed on the other endpoint. We say that a search strategy is \emph{winning} if after its termination all edges are cleared. The \emph{mixed-search-number} of a graph $G$, denoted $\ms(G)$, is the minimum number of searchers required for a winning strategy of the mixed-search-game on $G$.

\begin{lem}[\cite{TakahashiUK95}]
  \label{thm:mixed_search}
  We have that $\pw(G) \leq \ms(G) \leq \pw(G) + 1$.
\end{lem}

\section{Cut and Count for Modular-Treewidth}
\label{sec:modtw_cutandcount}

\subsection{General Approach}

In this section, we give an overview of the cut-and-count-technique and adapt it to parameterization by modular-treewidth. If we solve a problem on a graph $G = (V,E)$ involving connectivity constraints, we can make the following general definitions. We let $\sols \subseteq \powerset{U}$ denote the set of \emph{solutions}, living over some \emph{universe} $U$, and we have to determine whether $\sols$ is empty or not. The cut-and-count-technique does so in two parts:
\begin{itemize}
	\item \textbf{Cut part:} By \emph{relaxing} the connectivity constraints, we obtain a set $\sols \subseteq \rsols \subseteq \powerset{U}$ of possibly connected solutions. The set $\csols$ will contain pairs $(X, C)$ consisting of a candidate solution $X \in \rsols$ and a consistent cut $C$ of $X$, which is defined in \cref{dfn:cons_cut}. 
	\item \textbf{Count part:} We compute $|\csols|$ modulo some power of $2$ such that all non-connected solutions $X \in \rsols \setminus \sols$ cancel, because they are consistent with too many cuts. Hence, only connected candidates $X \in \sols$ remain.
\end{itemize}
The main definition and property for the cut-and-count-technique are as follows.

\begin{dfn}[\cite{CyganNPPRW22}]
  \label{dfn:cons_cut}
  A cut $(V_L, V_R)$ of an undirected graph $G = (V, E)$ is \emph{consistent} if $u \in V_L$ and $v \in V_R$ implies $\{u,v\} \notin E$. A \emph{consistently cut subgraph} of $G$ is a pair $(X, (X_L, X_R))$ such that $X \subseteq V$ and $(X_L, X_R)$ is a consistent cut of $G[X]$. We denote the set of consistently cut subgraphs of $G$ by $\cuts(G)$.
\end{dfn}

\begin{lem}[\cite{CyganNPPRW22}]
  \label{thm:cons_cut}
  Let $X$ be a subset of vertices. The number of consistently cut subgraphs $(X, (X_L, X_R))$ is equal to $2^{\cc(G[X])}$.
\end{lem}

\begin{proof}
  By the definition of a consistently cut subgraph $(X, (X_L, X_R))$ we have for every connected component $C$ of $G[X]$ that either $C \subseteq X_L$ or $C \subseteq X_R$. Hence, there are two choices for every connected component and we obtain $2^{\cc(G[X])}$ different consistently cut subgraphs $(X, (X_L, X_R))$. \qed
\end{proof}

The cut-and-count-approach can fail if $|\sols|$ is divisible by the considered power of $2$, as then even the connected solutions would cancel each other out. The isolation lemma, \cref{thm:isolation}, allows us to avoid this problem at the cost of randomization: We sample a weight function $\wfct \colon U \rightarrow [N]$ and instead count pairs with a fixed weight, then the isolation lemma tells us that it is likely that there exists a weight with a unique solution, which therefore cannot cancel.

\begin{dfn}
  A function $\wfct \colon U \rightarrow \ZZ$ \emph{isolates} a set family $\family \subseteq \powerset{U}$ if there is a unique $S' \in \family$ with $\wfct(S') = \min_{S \in \family} \wfct(S)$, where for subsets $X$ of $U$ we define $\wfct(X) = \sum_{u \in X} \wfct(u)$.
\end{dfn}

\begin{lem}[Isolation Lemma, \cite{MulmuleyVV87}]
  \label{thm:isolation}
  Let $\emptyset \neq \family \subseteq \powerset{U}$ be a set family over a universe $U$. Let $N \in \NN$ and for each $u \in U$ choose a weight $\wfct(u) \in [N]$ uniformly and independently at random. Then
  $\PP[\wfct \text{ isolates } \family] \geq 1 - |U|/N$.
\end{lem}

\cref{thm:cons_cut} distinguishes disconnected candidates from connected candidates via the number of consistent cuts for the respective candidate. We determine this number not for a single relaxed solution, but for all of them with a fixed weight. 

To apply the cut-and-count-technique for modular-treewidth, we first study how connectivity interacts with the modular structure. Typically, we consider vertex sets $X$ contained in some module $\pmodule \in \modint(G)$ that intersect at least two child modules of $\pmodule$, i.e., $|\modprojection_\pmodule(X)| \geq 2$. When $|\modprojection_\pmodule(X)| = 1$, we can recurse in the modular decomposition tree until at least two child modules are intersected or we arrive at an easily solvable special case. The following exchange argument shows that the connectivity of $G[X]$ is not affected by the precise intersection $X \cap \module$, $\module \in \children(\pmodule)$, but only whether $X \cap \module$ is empty or not.

\begin{lem}
  \label{thm:module_exchange_connected}
 Let $\pmodule \in \modint(G)$ and $X \subseteq \pmodule$ be a subset with $|\modprojection_\pmodule(X)| \geq 2$ and such that $G[X]$ is connected. For any module $\module \in \children(\pmodule)$ with $X \cap \module \neq \emptyset$ and $\emptyset \neq Y \subseteq \module$, the graph $G[(X \setminus \module) \cup Y]$ is connected.
\end{lem}

\begin{proof}
  Since $G[X]$ is connected and intersects at least two modules, there has to be a module $\module' \in \children(\pmodule)$ adjacent to $\module$ such that $X \cap \module' \neq \emptyset$. The edges between $Y$ and $X \cap \module'$ induce a biclique and hence all incident vertices must be connected to each other. Fix a vertex $u \in X \cap \module$ and consider any $w \in X \setminus \module$, then $G[X]$ contains an $u,w$-path $P$ such that the vertex $v$ after $u$ on $P$ is in $X \setminus \module$. For any $y \in Y$, we obtain an $y,w$-path $P_y$ in $G[(X \setminus \module) \cup Y]$ by replacing $u$ with $y$ in $P$. Finally, consider two vertices $u,w \in X \setminus \module$, then there is an $u,w$-path $P$ in $G[X]$. If $P$ does not intersect $\module$, then $P$ is also a path in $G[(X \setminus \module) \cup Y]$. Otherwise, we can assume that $P$ contains exactly one vertex $v$ of $\module$ and simply replace $v$ with some $y \in Y$ to obtain a $u,w$-path $P'$ in $G[(X \setminus \module) \cup Y]$. Hence, $G[(X \setminus \module) \cup Y]$ is connected as claimed. \qed
\end{proof}

Building upon \cref{thm:module_exchange_connected} allows us to reduce checking the connectivity of $G[X]$ to the quotient graph at $\pmodule$, as $\pquotient$ is isomorphic to the induced subgraph of $G$ obtained by picking one vertex from each child module of $\pmodule$.

\begin{lem}
  \label{thm:quotient_connected}
  Let $\pmodule \in \modint(G)$ and $X \subseteq \pmodule$ with $|\pproj(X)| \geq 2$, i.e., $X$ intersects at least two modules in $\children(\pmodule)$. It holds that $G[X]$ is connected if and only if $\pquotient[\pproj(X)]$ is connected.
\end{lem}

\begin{proof}
  For every module $\module \in \children(\pmodule)$ with $X \cap \module \neq \emptyset$, pick a vertex $v_\module \in X \cap \module$ and define $X' = \{v_\module \sep X \cap \module \neq \emptyset, \module \in \children(\pmodule)\} \subseteq X$. Note that $G[X']$ is isomorphic to $\pquotient[\pproj(X)]$. Hence, we are done if we can show that $G[X]$ is connected if and only if $G[X']$ is connected. If $G[X]$ is connected, then so is $G[X']$ by repeatedly applying \cref{thm:module_exchange_connected}. 
  
  For the converse, suppose that $G[X']$ is connected. We argue that every $v \in X \setminus X'$ is adjacent to some $w \in X'$ and then it follows that $G[X]$ is connected as well. There is some $\module \in \children(\pmodule)$ with $v \in \module$ and $v_\module \in X'$ by definition of $X'$. Since $|X'| \geq 2$ and $G[X']$ is connected, there is a neighbor $w \in X'$ of $v_\module$ in $G[X']$ and $w = v_{\module'}$ for some $\module' \in \children(\pmodule) \setminus \{\module\}$. The vertex $w$ has to be a neighbor of $v$ because $\module$ is a module and $w \notin \module$. \qed
\end{proof}

\cref{thm:quotient_connected} tells us that we do not need to consider \emph{heterogeneous} cuts, i.e., $(X, (X_L, X_R)) \in \cuts(G)$ with $X_L \cap \module \neq \emptyset$ and $X_R \cap \module \neq \emptyset$ for some module $\module \in \modpartition(G)$, because checking connectivity can be reduced to a set that contains at most one vertex per module. 

\begin{dfn}
  Let $\pmodule \in \modint(G)$. We say that a cut $(X_L, X_R)$, with $X_L \cup X_R \subseteq \pmodule$, is \emph{$\pmodule$-homogeneous} if $X_L \cap \module = \emptyset$ or $X_R \cap \module = \emptyset$ for every $\module \in \children(\pmodule)$. We may just say that $(X_L, X_R)$ is \emph{homogeneous} when $\pmodule$ is clear from the context. We define for every subgraph $G'$ of $G$ the set $\homcuts{\pmodule}(G') = \{(X, (X_L, X_R)) \in \cuts(G') \sep (X_L, X_R) \text{ is $\pmodule$-homogeneous}\}.$
\end{dfn}

Combining \cref{thm:cons_cut} with \cref{thm:quotient_connected}, the connectivity of $G[X]$ can be determined by counting $\pmodule$-homogeneous consistent cuts of $G[X]$ modulo 4.

\begin{lem}
  \label{thm:hom_cut}
  Let $\pmodule \in \modint(G)$ and $X \subseteq \pmodule$ with $|\pproj(X)| \geq 2$. It holds that $|\{(X_L, X_R) \sep (X, (X_L, X_R)) \in \homcuts{\pmodule}(G)\}| = 2^{\cc(\pquotient[\pproj(X)])}$ and $G[X]$ is connected if and only if $|\{(X_L, X_R) \sep (X, (X_L, X_R)) \in \homcuts{\pmodule}(G)\}| \neq 0 \mod 4$. 
\end{lem}

\begin{proof}
 Fix $\pmodule \in \modint(G)$ and $X \subseteq \pmodule$ with $|\pproj(X)| \geq 2$. For any set $S \subseteq \pmodule$, we write $S^q = \pproj(S)$ in this proof. We will argue that the map $(X_L, X_R) \mapsto (X^q_L, X^q_R)$ is a bijection between $\{(X_L, X_R) \sep (X, (X_L, X_R)) \in \homcuts{\pmodule}(G)\}$ and $\{(Y_L, Y_R) \sep (X^q, (Y_L, Y_R)) \in \cuts(\pquotient)\}$. First of all, notice that $(X^q_L, X^q_R)$ is a cut of $\pquotient[X^q]$ because $(X_L, X_R)$ is homogeneous. Furthermore, $(X^q_L, X^q_R)$ is a consistent cut, since any edge $\{v^q_{\module_1}, v^q_{\module_2}\}$ crossing $(X^q_L, X^q_R)$ would give rise to an edge $\{u_1, u_2\}$, $u_i \in \module_i$, $i \in [2]$, crossing $(X_L, X_R)$ which contradicts the assumption that $(X_L, X_R)$ is a consistent cut.
 
 For injectivity, consider $(X, (X_L, X_R))$, $(X, (Z_L, Z_R)) \in \homcuts{\pmodule}(G)$ such that $(X^q_L, X^q_R) = (Z^q_L, Z^q_R)$. Since they are homogeneous cuts, we can compute 
 $$X_L = \bigcup_{\qvertex \in X^q_L} X \cap \module = \bigcup_{\qvertex \in Z^q_L} X \cap \module = Z_L$$
 and similarly for $X_R = Z_R$. For surjectivity, note that every $(Y_L, Y_R)$ with $(X^q, (Y_L, Y_R)) \in \cuts(\pquotient)$ is hit by the following homogeneous cut $(X, (\bigcup_{v^q_\module \in Y_L} X \cap \module, \bigcup_{v^q_\module \in Y_R} X \cap \module))$. 
 
 Finally, we can apply \cref{thm:cons_cut} to $X^q \subseteq V(\pquotient)$ to obtain, via the bijection, that $|\{(X, (X_L, X_R)) \in \homcuts{\pmodule}(G)\}| = 2^{\cc(\pquotient[X^q])}$. Hence, $\pquotient[X^q]$ is connected if and only if $|\{(X, (X_L, X_R)) \in \homcuts{\pmodule}(G)\}| \neq 0 \mod 4$. The statement then follows by \cref{thm:quotient_connected}. \qed
\end{proof}

\section{Reductions}\label{sec:modtw_reduction}
\subsection{Steiner Tree}
\label{sec:modtw_st_reduction}

In the \textsc{(Node)} \ST problem, we are given a graph $G = (V,E)$, a set of \emph{terminals} $\terminals \subseteq V$, a cost function $\cfct \colon V \rightarrow \NN \setminus \{0\}$, and an integer $\budget$ and we have to decide whether there exists a subset of vertices $X \subseteq V$ such that $\terminals \subseteq X$, $G[X]$ is connected, and $\cfct(X) \leq \budget$.

We assume that $G$ is a connected graph, otherwise the answer is trivially no if the terminals are distributed across several connected components, or we can just look at the connected component containing all terminals. We also assume that $G[\terminals]$ is not connected, as otherwise $X = \terminals$ is trivially an optimal solution. Furthermore, we assume that the costs $\cfct(v)$, $v \in V$, are at most polynomial in $|V|$.

For \ST, it is sufficient to consider the topmost quotient graph $G^q := G^q_V = G / \modpartition(G)$, unless there is a single module $\module \in \modpartition(G) = \children(V)$ containing all terminals. In this edge case, we find a solution of size $|\terminals| + 1$, by taking a vertex in a module adjacent to $\module$, or we consider the graph $G[\module]$, allowing us to recurse into the module $\module$.

We first consider the case that all terminals are contained in a single module $\module \in \modpartition(G)$. The next lemma shows that we can either find a solution of size $|\terminals| + 1$, which can be computed in polynomial time, or it suffices to consider the graph $G[\module]$. 

\begin{lem}
  \label{thm:st_all_terminals_in_module}
  If there is a module $\module \in \modpartition(G)$ of $G$ such that $\terminals \subseteq \module$, then there is an optimum Steiner tree $X$ satisfying $X \subseteq \module$, or there is an optimum Steiner tree $X$ satisfying $|X| = |\terminals| + 1$.
\end{lem}

\begin{proof}
  Consider a Steiner tree $X$ such that $X \not\subseteq \module$, then $X$ has to contain at least one vertex $v$ inside a module $\module' \in \modpartition(G)$ adjacent to $\module$. We claim that $X' = \terminals \cup \{v\}$ is a Steiner tree with $\cfct(X') \leq \cfct(X)$. Clearly, $X' \subseteq X$, and since the costs are positive we have that $\cfct(X') \leq \cfct(X)$. Since $\terminals \subseteq \module$, the vertex $v$ is adjacent to all terminals $\terminals$ and $G[X']$ is connected, hence $X'$ is a Steiner tree.
  
  If there is no optimum Steiner tree $X$ satisfying $X \subseteq \module$, then by applying the previous argument to an optimum Steiner tree, we obtain an optimum Steiner tree $X$ satisfying $|X| = |\terminals| + 1$. \qed
\end{proof}

After recursing until no module $\module \in \modpartition(G)$ contains all terminals (and updating $G$ accordingly), we can apply the following reduction to solve the problem if the quotient graph is prime. Let $(G, \terminals, \cfct, \budget)$ be a \ST instance such that $|\modprojection_V(\terminals)| \geq 2$ and $G^q = G / \modpartition(G)$ is prime. We consider the \ST instance $(G^q, \terminals^q, \cfct^q, \budget^q)$ where $\terminals^q = \modprojection_{V}(\terminals)$, $\cfct^q(v^q_\module) = \cfct(\terminals \cap \module) = \sum_{v \in \terminals \cap \module} \cfct(v)$ if $\terminals \cap \module \neq \emptyset$ and $\cfct^q(v^q_\module) = \min_{v \in \module} \cfct(v)$ otherwise, and $\budget^q = \budget$.

\begin{lem}\label{thm:st_mod_reduction}
  Suppose that $(G, \terminals, \cfct, \budget)$ is a \ST instance such that no module $\module \in \modpartition(G)$ contains all terminals $\terminals$ and $G^q$ is prime. 
  
  Then, the answer to the \ST instance $(G, \terminals, \cfct, \budget)$ is positive if and only if the answer to the \ST instance $(G^q, \terminals^q, \cfct^q, \budget^q)$ is positive.
\end{lem}

\begin{proof}
  If $X$ is an optimum Steiner tree of $(G, \terminals, \cfct, \budget)$, then we claim that $X^q = \modprojection_V(X)$ is a Steiner tree of $(G^q, \terminals^q, \cfct^q, \budget^q)$ with $\cfct^q(X^q) \leq \cfct(X)$. We have that $\terminals^q = \modprojection_V(\terminals)$, so $\terminals \subseteq X$ implies that $\terminals^q \subseteq X^q$. By \cref{thm:quotient_connected}, we see that $G^q[X^q]$ is connected as well. By definition of $X^q$ and $\cfct^q$, we have for all $v^q_\module \in X^q$ that $\cfct^q(v^q_\module) \leq \cfct(X \cap \module)$ and hence $\cfct^q(X^q) \leq \cfct(X) \leq \budget = \budget^q$.
 
 If $X^q$ is an optimum Steiner tree of $(G^q, \terminals^q, \cfct^q, \budget^q)$, then we claim that $X = \terminals \cup \{v_\module \sep v^q_\module \in X^q, \terminals \cap \module = \emptyset\}$, where $v_\module = \arg \min_{v \in \module} \cfct(v)$, is a Steiner tree of $(G, \terminals, \cfct, \budget)$ with $\cfct(X) \leq \cfct^q(X^q)$. We have that $\terminals \subseteq X$ by definition of $X$ and for the costs we compute that $\cfct(X) = \cfct(\terminals) + \cfct(X \setminus \terminals) = \cfct^q(\terminals^q) + \cfct^q(X^q \setminus \terminals^q) = \cfct^q(X^q) \leq \budget^q = \budget$. Note that $X^q$ satisfies $X^q = \modprojection_V(X)$ by definition of $X$. Therefore, \cref{thm:quotient_connected} implies that $G[X]$ is connected and $X$ is a Steiner tree of $G$. \qed
\end{proof}

\begin{prop}[\cite{CyganNPPRW22}]
  \label{thm:st_tw_algo}
  There exists a Monte-Carlo algorithm that given a tree decomposition of width at most $k$ for $G$ solves \ST in time $\Oh^*(3^k)$. The algorithm cannot give false positives and may give false negatives with probability at most $1/2$.
\end{prop}

\begin{proof}
  The algorithm presented by Cygan et al.~\cite{CyganNPPRW11} can be easily augmented to handle positive vertex costs in this running time under the assumption that the costs $\cfct(v)$, $v \in V$, are at most polynomial in $|V|$. \qed
\end{proof}

By recursing, applying \cref{thm:st_tw_algo} to solve the reduced instance from \cref{thm:st_mod_reduction}, and handling parallel and series nodes, we obtain the following.

\begin{thm}
  \label{thm:st_modtw_algo}
  There exists a Monte-Carlo algorithm that given a tree decomposition of width at most $k$ for every prime node in the modular decomposition of $G$ solves \ST in time $\Oh^*(3^k)$. The algorithm cannot give false positives and may give false negatives with probability at most $1/2$.
\end{thm}

\begin{proof} 
  If no module $\module \in \modpartition(G)$ contains all terminals $\terminals$, then we want to invoke \cref{thm:st_mod_reduction}. If $G^q$ is a parallel node, then the answer is trivially no. If $G^q$ is a series node, then $G[\terminals]$ is already connected, but we have assumed that this is not the case. Hence, by \cref{thm:gallai_modular} $G^q$ must be a prime node and we can indeed invoke \cref{thm:st_mod_reduction}, so it suffices to solve the \ST instance $(G^q, \terminals^q, \cfct^q, \budget^q)$. By definition of modular-treewidth, we have $\tw(G^q) \leq \modtw(G) \leq k$ and we are given a corresponding tree decomposition of $G^q$. Hence, we can simply run the algorithm of \cref{thm:st_tw_algo} and return its result.
  
  If some module $\module \in \modpartition(G)$ contains all terminals $\terminal$, then due to \cref{thm:st_all_terminals_in_module} we first compute in polynomial time an optimum Steiner tree $X_1$ of $G$ subject to $|X_1| = |\terminals| + 1$ by brute force. If $\cfct(X_1) \leq \budget$, then we answer yes. Otherwise, we repeatedly recurse into the module $\module$ until we reach a node $G^q_* = G_* / \modpartition(G_*)$ in the modular decomposition of $G$ such that no $\module_* \in \modpartition(G_*)$ contains all terminals $\terminals$. We can then solve the \ST instance $(G_*, \terminals, \cfct\big|_{V(G_*)}, \budget)$ like in the first paragraph and return its answer. Note that this recursion can never lead to a $G_*$ with $|V(G_*)| = 1$ as that would imply $|\terminals| = 1$, which contradicts the assumption that $G[\terminals]$ is not connected. 
  
  As we call \cref{thm:st_tw_algo} at most once, we obtain the same error bound. \qed
\end{proof}

Cygan et al.~\cite{CyganNPPRW11arxiv} have shown that \ST cannot be solved in time $\Oh^*((3 - \eps)^{\pw(G)})$ for some $\eps > 0$, unless \SETH fails. Since $\modtw(G) \leq \tw(G) \leq \pw(G)$, this shows that the running time of \cref{thm:st_modtw_algo} is tight.

\subsection{Connected Dominating Set}
\label{sec:modtw_cds_reduction}

In the \CDS problem, we are given a graph $G = (V,E)$, a cost function $\cfct \colon V \rightarrow \NN \setminus \{0\}$, and an integer $\budget$ and we have to decide whether there exists a subset of vertices $X \subseteq V$ such that $N_G[X] = V$ and $G[X]$ is connected. We assume that $G$ is connected, otherwise the answer is trivially no, and that the costs $\cfct(v)$, $v \in V$, are at most polynomial in $|V|$.

\CDS can be solved by essentially considering only the first quotient graph. First, we will have to handle some edge cases though. If the first quotient graph $G^q = \qgraph{V} = G / \modpartition(G)$ contains a \emph{universal} vertex $\qvertex \in V(G^q)$, i.e., $N_{G^q}[\qvertex] = V(G^q)$, then there could be a connected dominating set $X$ of $G$ that is fully contained in $\module$. We search for such a connected dominating set by recursively solving \CDS on $G[\module]$. At some point, we arrive at a graph, where the first quotient graph does not contain a universal vertex, or at the one-vertex graph. In the latter case, the answer is trivial. Otherwise, the structure of connected dominating sets allows us to solve the problem on the quotient graph $G^q$.

\begin{lem}
  \label{thm:universal_implies_series}
 If $|V| \geq 2$, then $G^q$ contains a universal vertex if and only if $G^q$ is a clique.
\end{lem}

\begin{proof}
 The reverse direction is simple: every vertex of a clique is a universal vertex.
 
 For the forward direction, first notice that $G^q$ cannot be a parallel node if $G^q$ contains a universal vertex. Suppose that $G^q$ contains a universal vertex $v^q_{\module_0}$. Consider the set $\module = V(G) \setminus \module_0$ and notice that $\module$ has to be a module of $G$, because $v^q_{\module_0}$ is a universal vertex in $G^q$. If $G^q$ were a prime node, then all modules in $\modpartition(G)$ are maximal proper modules by \cref{thm:gallai_modular}, but $V(G) = \module_0 \cup \module$ implies that $|\modpartition(G)| \leq 2$ which contradicts that $G^q$ is prime. Therefore, the only remaining possibility is that $G^q$ is a series node, i.e., $G^q$ is a clique. \qed
\end{proof}

\begin{lem}
 \label{thm:cds_mod_structure}
 If $G^q$ is a prime node, then no connected dominating set $X$ of $G$ is contained in a single module $\module \in \modpartition(G)$. Furthermore, for any optimum connected dominating set $X$ of $G$ and module $\module \in \modpartition(G)$ it holds that either $X \cap \module = \emptyset$ or $X \cap \module = \{v_\module\}$, where $v_\module$ is some vertex of minimum cost in $\module$.
\end{lem}

\begin{proof}
  By \cref{thm:universal_implies_series}, $G^q$ cannot contain a universal vertex. Suppose that $X \subseteq \module$ for some $\module \in \modpartition(G)$. Since $v^q_\module \in V(G^q)$ is not a universal vertex, there exists a module $\module' \in \modpartition(G) \setminus \{\module\}$ that is not adjacent to $\module$, hence $X$ cannot dominate the vertices in $\module'$ and thus cannot be a connected dominating set.
  
  For the statement about optimum connected dominating sets, suppose that $X$ is a connected dominating set of $G$ and $\cfct(X \cap \module) > \cfct(v_\module) > 0$, where $v_\module$ is some vertex of minimum cost in $\module$, for some $\module \in \modpartition(G)$. The set $X' = (X \setminus \module) \cup \{v_\module\}$ satisfies $\cfct(X') < \cfct(X)$ and \cref{thm:module_exchange_connected} shows that $G[X']$ is connected. Since $X$ is a connected dominating set intersecting at least two modules, there has to be a module $\module' \in \modpartition(G)$ that is adjacent to $\module$ and satisfies $X \cap \module' \neq \emptyset$. Since $\module \neq \module'$, there is some $v \in X' \cap \module' \neq \emptyset$ which dominates all vertices in $\module$. Hence, $X'$ is a dominating set as well.
  
  Repeatedly applying this argument shows the statement about optimum connected dominating sets.  \qed
\end{proof}

\begin{prop}[\cite{CyganNPPRW22}]
  \label{thm:cds_tw_algo}
  There exists an algorithm that given a tree decomposition of width at most $k$ for $G$ and a weight function $\wfct$ isolating the optimum connected dominating sets solves \CDS in time $\Oh^*(4^k)$. If $\wfct$ is not isolating, then the algorithm may return false negatives.
\end{prop}

\begin{proof}
  The algorithm presented by Cygan et al.~\cite{CyganNPPRW22} can be easily augmented to handle positive vertex costs in this running time under the assumption that the costs $\cfct(v)$, $v \in V$, are at most polynomial in $|V|$. Notice that the only source of randomness in the algorithm of Cygan et al.\ is the sampling of a weight function. If we are already given an isolating weight function, the algorithm will always succeed. \qed
\end{proof}

As for \ST, the strategy is again to essentially just call the known algorithm for \CDS parameterized by treewidth on the quotient graphs. However, a single call will not be sufficient in the case of \CDS; to still obtain the same success probability, we will analyze the behavior of isolating weight functions under the following reduction.

Let $(G, \cfct, \budget)$ be a \CDS instance such that $G^q$ is a prime node and let $\wfct \colon V \rightarrow \NN$ be a weight function. In each $\module \in \modpartition(G)$ pick a vertex $v^{\cfct, \wfct}_\module$ that lexicographically minimizes $(\cfct(v), \wfct(v))$ among all vertices $v \in \module$. We construct the \CDS instance $(G^q, \cfct^q, \budget)$ with $\cfct^q(v^q_\module) = \cfct(v^{\cfct, \wfct}_\module)$ for all $v^q_\module \in V(G^q)$ and define the weight function $\wfct^q(v^q_\module) = \wfct(v^{\cfct, \wfct}_\module)$ for all $v^q_\module \in V(G^q)$.

\begin{lem}
  \label{thm:cds_modtw_reduction}
  Let $(G, \cfct, \budget)$ be a \CDS instance such that $G^q$ is a prime node, let $\wfct \colon V \rightarrow \NN$ be a weight function, and let $(G^q, \cfct^q, \budget)$ and $\wfct^q$ be defined as above. The following statements hold:
 \begin{enumerate}
  \item If $X$ is an optimum connected dominating set of $(G, \cfct)$, then $X^q = \modprojection_V(X)$ is a connected dominating set of $G^q$ with $\cfct^q(X^q) = \cfct(X)$.
  \item If $X^q$ is an optimum connected dominating set of $(G^q, \cfct^q)$, then $X = \{v^{\cfct, \wfct}_\module \sep v^q_\module \in X^q\}$ is a connected dominating set of $G$ with $\cfct(X) = \cfct^q(X^q)$.
  \item If $\wfct$ isolates the optimum connected dominating sets of $(G, \cfct)$, then $\wfct^q$ isolates the optimum connected dominating sets of $(G^q, \cfct^q)$.
 \end{enumerate}
\end{lem}

\begin{proof}
  First, notice that the subgraph $G' = (V', E')$ of $G$ induced by $\{v^{\cfct,\wfct}_\module \sep \module \in \modpartition(G)\}$ is isomorphic to $G^q$.
  \begin{enumerate}
    \item Let $X$ be an optimum connected dominating set of $(G, \cfct)$ and set $X^q = \modprojection_V(X)$. We compute 
    \begin{equation*}
      \cfct^q(X^q) = \sum_{v^q_\module \in X^q} \cfct^q(v^q_\module) = \sum_{\substack{\module \in \modpartition(G):\\X \cap \module \neq \emptyset}} \cfct(v^{\cfct,\wfct}_\module) = \sum_{\substack{\module \in \modpartition(G):\\X \cap \module \neq \emptyset}} \cfct(X \cap \module) = \cfct(X),
    \end{equation*}
    where the penultimate equality follows from \cref{thm:cds_mod_structure} and the choice of $v^{\cfct, \wfct}_\module$. Furthermore, we can assume $X \cap \module = \{v^{\cfct, \wfct}_\module\}$ whenever $X \cap \module \neq \emptyset$ by \cref{thm:cds_mod_structure}. Then, the isomorphism between $G^q$ and $G'$ also maps $X^q$ to $X$ and hence $X^q$ has to be a connected dominating set of $G^q$.
    \item Suppose that $X^q$ is an optimum connected dominating set of $(G^q, \cfct^q)$. Defining $X$ as above, we see that $X^q$ satisfies $X^q = \modprojection_V(X)$. By \cref{thm:universal_implies_series}, $G^q$ contains no universal vertex, hence $|X^q| \geq 2$ and $X$ must intersect at least two modules. Therefore, we can apply \cref{thm:quotient_connected} to see that $G[X]$ is connected. The isomorphism between $G^q$ and $G'$ shows that $X$ must dominate all vertices in $V'$.
    
    For any vertex $v \in V \setminus (X \cup V')$ and its module $v \in \module \in \modpartition(G)$, we claim that there exists a module $\module' \in \modpartition(G)$ such that $v^{\cfct, \wfct}_{\module'} \in X$ dominates $v$. If $X \cap \module = \emptyset$, then there exists an adjacent module $\module'$ with $X \cap \module' \neq \emptyset$, because the vertex $v^{\cfct, \wfct}_\module \in V'$ must be dominated by $X$. If $X \cap \module \neq \emptyset$, a module $\module'$ with the same properties exists, because $X$ intersects at least two modules and $G[X]$ is connected. In either case, $v^{\cfct, \wfct}_{\module'}$ must dominate the vertex $v$ by the module property, hence $X$ is a connected dominating set of $G$. It remains to compute 
    \begin{equation*}
      \cfct(X) = \sum_{v^{\cfct,\wfct}_\module \in X} \cfct(v^{\cfct,\wfct}_\module) = \sum_{v^q_\module \in X^q} \cfct^q(v^q_\module) = \cfct^q(X^q).
    \end{equation*}
    \item The first two statements show that connected dominating sets in $(G, \cfct)$ and $(G^q, \cfct^q)$ have the same optimum cost. Suppose that $\wfct$ is a weight function that isolates the optimum connected dominating sets of $(G, \cfct)$ and let $X$ be the optimum connected dominating set that is isolated by $\wfct$. Therefore, $X$ lexicographically minimizes $(\cfct(X), \wfct(X))$ among all connected dominating sets of $G$. By \cref{thm:cds_mod_structure}, we know that $X \cap \module = \{v'_\module\}$ whenever $X \cap \module \neq \emptyset$, where $v'_\module$ is a vertex of minimum cost in $\module$. 
    
    We claim that $v'_\module = v^{\cfct, \wfct}_\module$ for all modules $\module \in \modpartition(G)$ with $X \cap \module \neq \emptyset$. By definition of $v^{\cfct, \wfct}_\module$, we must have $\wfct(v'_\module) \geq \wfct(v^{\cfct, \wfct}_\module)$. If $\wfct(v'_\module) > \wfct(v^{\cfct, \wfct}_\module)$, then we could reduce the weight of $X$ by exchanging $v'_\module$ with $v^{\cfct, \wfct}_\module$, contradicting the minimality of $(\cfct(X), \wfct(X))$. If $\wfct(v'_\module) = \wfct(v^{\cfct, \wfct}_\module)$ and $v'_\module \neq v^{\cfct, \wfct}_\module$, then $X$ cannot be the isolated connected dominating set, because by exchanging $v'_\module$ and $v^{\cfct, \wfct}_\module$ we would obtain another connected dominating set of the same cost and weight. This proves the claim.
    
    Using the claim, we compute 
    \begin{equation*}
      \wfct^q(X^q) = \sum_{v^q_\module \in X^q} \wfct^q(v^q_\module) = \sum_{\substack{\module \in \modpartition(G):\\X \cap \module \neq \emptyset}} \wfct(v^{\cfct, \wfct}_\module) = \wfct(X).
    \end{equation*}
    Finally, consider any other optimum connected dominating set $Y^q \neq X^q$ of $G^q$. Setting $Y = \{v^{\cfct, \wfct}_\module \sep v^q_\module \in Y^q\} \neq X$, we obtain $Y^q = \modprojection_V(Y)$ and $\cfct(Y) = \cfct^q(Y^q) = \cfct^q(X^q) = \cfct(X)$, hence $\wfct^q(Y^q) = \wfct(Y) > \wfct(X) = \wfct^q(X^q)$, where the inequality follows because $\wfct$ isolates the optimum connected dominating sets of $(G, \cfct)$. This shows that $\wfct^q$ isolates the optimum connected dominating sets of $(G^q, \cfct^q)$. \qed
  \end{enumerate}
\end{proof}

\begin{thm}
  \label{thm:cds_modtw_algo}
  There exists a Monte-Carlo algorithm that given a tree decomposition of width at most $k$ for every prime node in the modular decomposition of $G$ solves \CDS in time $\Oh^*(4^k)$. The algorithm cannot give false positives and may give false negatives with probability at most $1/2$. 
\end{thm}

\begin{proof} 
  We begin by sampling a weight function $\wfct \colon V \rightarrow [2|V|]$. By \cref{thm:isolation}, $\wfct$ isolates the optimum connected dominating sets of $(G, \cfct)$ with probability at least $1/2$. The algorithm proceeds top-down through the modular decomposition tree of $G$, but we only recurse further if the current node is a series node. Each recursive call is determined by some $\pmodule \in \modtree(G)$ and we have to determine in this call if a connected dominating set $X$ of $G[\pmodule]$ with $\cfct(X) \leq \budget$ exists, i.e., solve the \CDS instance $(G[\pmodule], \cfct\restrict{\pmodule}, \budget)$. The weight function $\wfct$ is passed along by considering its restriction, i.e., $\wfct\restrict{\pmodule}$. 
  
  Let $\algo_{\tw}$ denote the algorithm from \cref{thm:cds_tw_algo}. Our algorithm may perform several calls to $\algo_{\tw}$, where each call may return false negatives when the considered weight function is not isolating. We return to the error analysis after finishing the description of the modular-treewidth algorithm. 
  
  We begin by explaining the three base cases. If $|\pmodule| = 1$, then we let $\pmodule = \{v_\pmodule\}$ and check whether $\cfct(v_\pmodule) \leq \budget$ and return yes or no accordingly. Otherwise, we have $|\pmodule| \geq 2$ and can consider $\pquotient$. If $\pquotient$ is a parallel node, then the answer is trivially no. If $\pquotient$ is a prime node, then we can invoke \cref{thm:cds_modtw_reduction} to reduce the \CDS instance $(G[\pmodule], \cfct\restrict{\pmodule}, \budget)$ to a \CDS instance on the quotient graph $\pquotient$. We are given a tree decomposition of $\pquotient$ of width at most $k$ by assumption. We run $\algo_{\tw}$ on the quotient instance together with the weight function from \cref{thm:cds_modtw_reduction} and return its result.
   
  Finally, suppose that $\pquotient$ is a series node. In this case, any set $X$ of size $2$ that intersects two different modules $\module \in \children(\pmodule) = \modpartition(G[\pmodule])$ is a connected dominating set of $G[\pmodule]$. We compute all those sets by brute force in polynomial time and return yes if any of them satisfies $\cfct(X) \leq \budget$. Otherwise, we need to recurse into the modules $\module \in \children(\pmodule)$, because any connected dominating set of $G[\module]$ will also be a connected dominating set of $G[\pmodule]$. We return true if at least one of these recursive calls returns true. This concludes the description of the algorithm and we proceed with the error analysis now. 
  
  The only source of errors is that we may call $\algo_{\tw}$ with a non-isolating weight function, but this can only yield false negatives and hence the modular-treewidth algorithm cannot give false positives either. Even if the sampled weight function is isolating, this may not be the case for the restrictions $\wfct\restrict{\pmodule}$, $\pmodule \in \modtree(G)$. Nonetheless, we show that if $\wfct$ is isolating, then the modular-treewidth algorithm does not return an erroneous result. To do so, we show that if $\wfct\restrict\pmodule$ is isolating at a series node, then the weight function in the branch containing the isolated optimum connected dominating set must be isolating as well.
  
  To be precise, suppose that $\pquotient$ is a series node and that $\wfct\restrict{\pmodule}$ isolates $X^*$ among the optimum connected dominating sets of $(G[\pmodule], \cfct\restrict{\pmodule})$. We claim that $\wfct\restrict{\module}$, $\module \in \children(\pmodule)$, isolates $X^*$ among the optimum connected dominating sets of $(G[\module], \cfct\restrict{\module})$ if $X^* \subseteq \module$. This follows by a simple exchange argument: if $\wfct\restrict{\module}$ is not isolating, i.e., there is some optimum connected dominating set $X \neq X^*$ of $(G[\module], \cfct\restrict{\module})$ with $\wfct(X) = \wfct(X^*)$, then $X$ is also an optimum connected dominating set of $(G[\pmodule], \cfct\restrict{\pmodule})$, contradicting that $\wfct\restrict{\pmodule}$ is isolating $X^*$. If $X^*$ intersects multiple modules $\module \in \children(\pmodule)$, then $X^*$ is found deterministically among the sets of size $2$. 
  
  As $\wfct$ is isolating with probability at least $1/2$ this concludes the error analysis. Furthermore, for every module $\module \in \modtree(G)$, we need at most time $\Oh^*(4^k)$. Therefore, the theorem statement follows. \qed
\end{proof}

Cygan et al.~\cite{CyganNPPRW11arxiv} have shown that \CDS cannot be solved in time $\Oh^*((4 - \eps)^{\pw(G)})$ for some $\eps > 0$, unless \SETH fails. Since $\modtw(G) \leq \tw(G) \leq \pw(G)$, this shows that the running time of \cref{thm:cds_modtw_algo} is tight.

\section{Connected Vertex Cover Algorithm}
\label{sec:modtw_cvc_algo}

In the \CVC problem, we are given a graph $G = (V,E)$, a cost function $\cfct\colon V \rightarrow \NN \setminus \{0\}$, and an integer $\budget$ and we have to decide whether there exists a subset of vertices $X \subseteq V$ with $\cfct(X) \leq \budget$ such that $G - X$ contains no edges and $G[X]$ is connected. We will assume that the values of the cost function $\cfct$ are polynomially bounded in the size of the graph $G$. We also assume that $G$ is connected and contains at least two vertices, hence $|\modpartition(G)| \geq 2$ and $G^q := G^q_V = G / \modpartition(G)$ cannot be edgeless.

To solve \CVC, we begin by computing some optimum (possibly non-connected) vertex cover $Y_\module$ with respect to $\cfct{\big|_\module}$ for every module $\module \in \modpartition(G)$ that $G[\module]$ contains at least one edge. If $G[\module]$ contains no edges, then we set $Y_\module = \{v^*_\module\}$, where $v^*_\module \in \module$ is a vertex minimizing the cost inside $\module$, i.e., $v^*_\module := \arg \min_{v \in \module} \cfct(v)$. The vertex covers can be computed in time $\Oh^*(2^{\modtw(G)})$ by using the algorithm from \cref{thm:is_mtw_algo}.

\begin{dfn}
 Let $X \subseteq V$ be a vertex subset. We say that $X$ is \emph{nice} if for every module $\module \in \modpartition(G)$ it holds that $X \cap \module \in \{\emptyset, Y_\module, \module\}$. 
\end{dfn}

We will show that it is sufficient to only consider nice vertex covers via some exchange arguments. This allows us to only consider a constant number of states per module in the dynamic programming algorithm.

\begin{lem}
  \label{thm:cvc_mod_structure}
  If there exists a connected vertex cover $X$ of $G$ that intersects at least two modules in $\modpartition(G)$, then there exists a connected vertex cover $X'$ of $G$ that is nice and intersects at least two modules in $\modpartition(G)$ with $\cfct(X') \leq \cfct(X)$.
\end{lem}

\begin{proof}
  Let $X$ be the given connected vertex cover. Via exchange arguments, we will see that we can find a nice connected vertex cover with the same cost. Suppose that there is a module $\module \in \modpartition(G)$ such that $G[\module]$ contains no edges and $1 \leq |X \cap \module| < |\module|$. We claim that $X' = (X \setminus \module) \cup \{v^*_\module\}$ is a connected vertex cover with $\cfct(X') \leq \cfct(X)$. For any module $\module' \in \modpartition(G)$ adjacent to $\module$, we must have that $X' \cap \module' = X \cap \module' = \module'$, else there would be an edge between $\module$ and $\module'$ that is not covered by $X$. In particular, all edges incident to $\module$ are already covered by $X \setminus \module = X' \setminus \module$. By \cref{thm:module_exchange_connected}, $X'$ is connected and we have that $\cfct(X') \leq \cfct(X)$ due to the choice of $v^*_\module$. 
  
  If $\module \in \modpartition(G)$ is a module such that $G[\module]$ contains at least one edge, then we consider two cases. If $\cfct(X \cap \module) < \cfct(Y_\module)$, then $X \cap \module$ cannot be a vertex cover of $G[\module]$ and hence $X$ would not be a vertex cover of $G$. If $\cfct(Y_\module) \leq \cfct(X \cap \module) < \cfct(\module)$, then we claim that $X' = (X \setminus \module) \cup Y_\module$ is a connected vertex cover with $\cfct(X') \leq \cfct(X)$. By assumption, we have $\cfct(X') \leq \cfct(X)$. We must have that $X \cap \module \neq \module$, therefore, as before, $X$ and $X'$ must fully contain all modules adjacent to $\module$ to cover all edges leaving $\module$. Since $G[\module]$ contains at least one edge, we have that $Y_\module \neq \emptyset$ and $G[X']$ must be connected by \cref{thm:module_exchange_connected}. 
  
  By repeatedly applying these arguments to $X$, we obtain the claim. \qed
\end{proof}

The next lemma enables us to handle connected vertex covers that are contained in a single module with polynomial-time preprocessing. 

\begin{lem}
 \label{thm:cvc_single_module}
 A vertex set $X \subseteq V$ is a connected vertex cover of $G$ with $X \subseteq \module$ for some module $\module \in \modpartition(G)$ if and only if $X = \module$, all edges of $G$ are incident to $\module$, and $G[\module]$ is connected.
\end{lem}

\begin{proof}
  The reverse direction is trivial. We will show the forward direction.	Since $G$ is connected and $|\modpartition(G)| \geq 2$, there exists a module $\module' \in \modpartition(G)$ adjacent to $\module$. If $X \neq \module$, then there exists an edge between $\module$ and $\module'$ that is not covered by $X$. If there is an edge in $G$ not incident to $\module$, then clearly $X$ cannot cover all edges. Clearly, $G[X] = G[\module]$ must be connected. \qed
\end{proof}

Before going into the main algorithm, we handle the edge case of series nodes. The following lemma shows that there are only a polynomial number of interesting cases for series nodes, hence we can check them by brute force in polynomial time.

\begin{lem}
 \label{thm:cvc_series}
 If $G^q$ is a clique of size at least two, then for any vertex cover $X$ there is some $\module' \in \modpartition(G)$ such that for all other modules $\module' \neq \module \in \modpartition(G)$, we have $X \cap \module = \module$.
\end{lem}

\begin{proof}
 Suppose there are two modules $\module_1 \neq \module_2 \in \modpartition(G)$ such that $X \cap \module_1 \neq \module_1$ and $X \cap \module_2 \neq \module_2$. These modules are adjacent, because $G^q$ is a clique< and thus $X$ cannot be a vertex cover, since there exists an uncovered edge between $\module_1 \setminus X$ and $\module_2 \setminus X$. \qed
\end{proof}
 
\subsection{Dynamic Programming for Prime Nodes}
 
 It remains to handle the case that $G$ is a prime node. Due to \cref{thm:cvc_single_module}, we only need to look for connected vertex covers that intersect at least two modules in $\modpartition(G)$ now. Hence, we can make use of \cref{thm:cvc_mod_structure} and \cref{thm:hom_cut}. We are given a tree decomposition $(\TT^q, (\bag^q_t)_{t \in V(\TT^q)})$ of the quotient graph $G^q := G^q_V = G / \modpartition(G)$ of width $k$ and by \cref{thm:very_nice_tree_decomposition}, we can assume that it is a very nice tree decomposition.
 
 To solve \CVC on $G$, we perform dynamic programming along the tree decomposition $\TT^q$ using the cut-and-count-technique. \cref{thm:hom_cut} allows us to work directly on the quotient graph. We begin by presenting the cut-and-count-formulation of the problem. For any subgraph $G'$ of $G$, we define the \emph{relaxed solutions} $\rsols(G') = \{X \subseteq V(G') \sep X \text{ is a nice vertex cover of $G'$}\}$ and \emph{the cut solutions} $\csols(G') = \{(X, (X_L, X_R)) \in \homcuts{V}(G') \sep X \in \rsols(G')\}$.

 For the isolation lemma, cf.\ \cref{thm:isolation}, we sample a weight function $\wfct\colon V \rightarrow [2n]$ uniformly at random. We will need to track the cost $\cfct(X)$, the weight $\wfct(X)$, and the number of intersected modules $|\modprojection_V(X)|$ of each partial solution $(X, (X_L, X_R))$. Accordingly, we define $\rsols^{\ctarget, \wtarget, \modtarget}(G') = \{X \in \rsols(G') \sep \cfct(X) = \ctarget, \wfct(X) = \wtarget, |\modprojection_V(X)| = \modtarget\}$ and $\csols^{\ctarget, \wtarget, \modtarget}(G') = \{(X, (X_L, X_R)) \in \csols(G') \sep X \in \rsols^{\ctarget, \wtarget, \modtarget}(G')\}$ for all subgraphs $G'$ of $G$, $\ctarget \in [0, \cfct(V)], \wtarget \in [0, \wfct(V)], \modtarget \in [0, |\modpartition(G)|]$.
 
 As discussed, to every node $t \in V(\TT^q)$ we associate a subgraph $G^q_t = (V^q_t, E^q_t)$ of $G^q$ in the standard way, which in turn gives rise to a subgraph $G_t = (V_t, E_t)$ of $G$. The subgraphs $G_t$ grow module by module and are considered by the dynamic program, hence we define $\rsols^{\ctarget, \wtarget, \modtarget}_t = \rsols^{\ctarget, \wtarget, \modtarget}(G_t)$ and $\csols^{\ctarget, \wtarget, \modtarget}_t = \csols^{\ctarget, \wtarget, \modtarget}(G_t)$ for all $\ctarget$, $\wtarget$, and $\modtarget$. We will compute the sizes of the sets $\csols^{\ctarget, \wtarget, \modtarget}_t$ by dynamic programming over the tree decomposition $\TT^q$, but to do so we need to parameterize the partial solutions by their state on the current bag.
 
 Disregarding the side of the cut, \cref{thm:cvc_mod_structure} tells us that each module $\module \in \modpartition(G)$ has one of three possible states for some $X \in \rsols^{\ctarget, \wtarget, \modtarget}_t$, namely $X \cap \module \in \{\emptyset, Y_\module, \module\}$. Since we are considering homogeneous cuts there are two possibilities if $X \cap \module \neq \emptyset$; $X \cap \module$ is contained in the left side of the cut or in the right side. Thus, there are five total choices. We define $\states = \{\zero, \one_L, \one_R, \all_L, \all_R\}$ with $\one$ denoting that the partial solution contains at least one vertex, but not all, from the module and with $\all$ denoting that the partial solution contains all vertices of the module; the subscript denotes the side of the cut.
 
 A function of the form $f\colon \bag^q_t \rightarrow \states$ is called \emph{$t$-signature}. For every node $t \in V(\TT^q)$, cost $\ctarget \in [0, \cfct(V)]$, weight $\wtarget \in [0, \wfct(V)]$, number of modules $\modtarget \in [0, |\modpartition(G)|]$, and $t$-signature $f$, the family $\dpsols^{\ctarget, \wtarget, \modtarget}_t(f)$ consists of all $(X, (X_L, X_R)) \in \csols^{\ctarget, \wtarget, \modtarget}_t$ that satisfy for all $v^q_\module \in \bag^q_t$:
 \begin{align*}
	 f(v^q_\module) = \zero_{\phantom{L}} & \leftrightarrow _{\phantom{L}}X \cap \module = \emptyset, \\
	 f(v^q_\module) = \one_L & \leftrightarrow X_L \cap \module = Y_\module \neq \module, &
	 f(v^q_\module) = \one_R & \leftrightarrow X_R \cap \module = Y_\module \neq \module, \\
	 f(v^q_\module) = \all_L & \leftrightarrow X_L \cap \module = \module, &
	 f(v^q_\module) = \all_R & \leftrightarrow X_R \cap \module = \module.
 \end{align*}
 Recall that by considering homogeneous cuts, we have that $X_L \cap \module = \emptyset$ or $X_R \cap \module = \emptyset$ for every module $\module \in \modpartition(G)$. We use the condition $Y_\module \neq \module$ for the states $\one_L$ and $\one_R$ to ensure a well-defined state for modules of size 1. Note that the sets $\dpsols^{\ctarget, \wtarget, \modtarget}_t(f)$, ranging over $f$, partition $\csols^{\ctarget, \wtarget, \modtarget}_t$ due to considering nice vertex covers and homogeneous cuts.

 Our goal is to compute the size of $\dpsols^{\ctarget, \wtarget, \modtarget}_{\rvertex}(\emptyset) = \csols^{\ctarget, \wtarget, \modtarget}_{\rvertex} = \csols^{\ctarget, \wtarget, \modtarget}(G)$, where $\rvertex$ is the root vertex of the tree decomposition $\TT^q$, modulo 4 for all $\ctarget$, $\wtarget$, $\modtarget$. By \cref{thm:hom_cut}, there is a connected vertex cover $X$ of $G$ with $\cfct(X) = \ctarget$ and $\wfct(X) = \wtarget$ if the result is nonzero. 

 We present the recurrences for the various bag types to compute $\DP^{\ctarget, \wtarget, \modtarget}_t(f) = |\dpsols^{\ctarget, \wtarget, \modtarget}_t(f)|$; if not stated otherwise, then $t \in V(\TT^q)$, $\ctarget \in [0, \cfct(V)]$, $\wtarget \in [0, \wfct(V)]$, $\modtarget \in [0, |\modpartition(G)|]$, and $f$ is a $t$-signature. We set $\DP^{\ctarget, \wtarget, \modtarget}_t(f) = 0$ whenever at least one of $\ctarget$, $\wtarget$, or $\modtarget$ is negative.
 
\subsubsection*{Leaf bag.} We have that $\bag^q_t = \bag_t = \emptyset$ and $t$ has no children. The only possible $t$-signature is $\emptyset$ and the only possible partial solution is $(\emptyset, (\emptyset, \emptyset))$. Hence, we only need to check the tracker values:
\begin{equation*}
	\DP_t^{\ctarget, \wtarget, \modtarget}(\emptyset) = [\ctarget = 0] [\wtarget = 0] [\modtarget = 0].
\end{equation*} 

\subsubsection*{Introduce vertex bag.} We have $\bag^q_t = \bag^q_s \cup \{v^q_\module\}$, where $s \in V(\TT^q)$ is the only child of $t$ and $v^q_\module \notin \bag^q_s$. Hence, $\bag_t = \bag_s \cup \module$. We have to consider all possible interactions of a partial solution with $\module$, since we are considering nice vertex covers these interactions are quite restricted. To formulate the recurrence, we let, as an exceptional case, $f$ be an $s$-signature here and not a $t$-signature. Since no edges of the quotient graph $G^q$ incident to $v^q_\module$ are introduced yet, we only have to check some edge cases and update the trackers when introducing $v^q_\module$:
\begin{equation*}
\begin{array}{lcll}
	\DP^{\ctarget, \wtarget, \modtarget}_t(f[v^q_\module \mapsto \zero]) & = & [G[\module] \text{ is edgeless}] & \DP^{\ctarget, \wtarget, \modtarget}_s(f), \\
	\DP^{\ctarget, \wtarget, \modtarget}_t(f[v^q_\module \mapsto \one_L]) & = & [|\module| > 1] & \DP^{\ctarget - \cfct(Y_\module), \wtarget - \wfct(Y_\module), \modtarget - 1}_s(f), \\
	\DP^{\ctarget, \wtarget, \modtarget}_t(f[v^q_\module \mapsto \one_R]) & = & [|\module| > 1] & \DP^{\ctarget - \cfct(Y_\module), \wtarget - \wfct(Y_\module), \modtarget - 1}_s(f), \\
	\DP^{\ctarget, \wtarget, \modtarget}_t(f[v^q_\module \mapsto \all_L]) & = & & \DP^{\ctarget - \cfct(\module), \wtarget - \wfct(\module), \modtarget - 1}_s(f), \\
	\DP^{\ctarget, \wtarget, \modtarget}_t(f[v^q_\module \mapsto \all_R]) & = & & \DP^{\ctarget - \cfct(\module), \wtarget - \wfct(\module), \modtarget - 1}_s(f). \\
\end{array}
\end{equation*}

\subsubsection*{Introduce edge bag.} Let $\{v^q_{\module_1}, v^q_{\module_2}\}$ denote the introduced edge. We have that $\{v^q_{\module_1}, v^q_{\module_2}\} \subseteq \bag^q_t = \bag^q_s$. The edge $\{v^q_{\module_1}, v^q_{\module_2}\}$ corresponds to adding a join between the modules $\module_1$ and $\module_2$. We need to filter all solutions whose states at $\module_1$ and $\module_2$ are not consistent with $\module_1$ and $\module_2$ being adjacent. There are essentially two possible reasons: either not all edges between $\module_1$ and $\module_2$ are covered, or the introduced edges go across the homogeneous cut. We implement this via the helper function $\cons\colon \states \times \states \rightarrow \{0,1\}$ which is defined by $\cons(\state_1, \state_2) = [\{\state_1, \state_2\} \cap \{\all_L, \all_R\} \neq \emptyset][\state_1 \in \{\one_L, \all_L\} \rightarrow \state_2 \notin \{\one_R, \all_R\}][\state_1 \in \{\one_R, \all_R\} \rightarrow \state_2 \notin \{\one_L, \all_L\}]$ or, equivalently, the following table:
\begin{equation*}
	\begin{array}{l|ccccc}
	\cons  & \zero & \one_L & \one_R & \all_L & \all_R \\
	\hline
	\zero  & 0 & 0 & 0 & 1 & 1 \\
	\one_L & 0 & 0 & 0 & 1 & 0 \\
	\one_R & 0 & 0 & 0 & 0 & 1 \\
	\all_L & 1 & 1 & 0 & 1 & 0 \\
	\all_R & 1 & 0 & 1 & 0 & 1 
	\end{array}
\end{equation*}
The recurrence is then simply given by 
\begin{equation*}
	\DP^{\ctarget, \wtarget, \modtarget}_t(f) = \cons(f(v^q_{\module_1}), f(v^q_{\module_2})) \DP^{\ctarget, \wtarget, \modtarget}_s(f).
\end{equation*}

\subsubsection*{Forget vertex bag.} We have that $\bag^q_t = \bag^q_s \setminus \{v^q_\module\}$, where $v^q_\module \in \bag^q_s$ and $s \in V(\TT^q)$ is the only child of $t$. Here, we only need to forget the state at $v^q_\module$ and accumulate the contributions from the different states $\qvertex$ could assume, as the states are disjoint no overcounting happens:
\begin{equation*}
	\DP^{\ctarget, \wtarget, \modtarget}_t(f) = \sum_{\state \in \states} \DP^{\ctarget, \wtarget, \modtarget}_s(f[v \mapsto \state]).
\end{equation*}

\subsubsection*{Join bag.} We have $\bag^q_t = \bag^q_{s_1} = \bag^q_{s_2}$, where $s_1, s_2 \in V(\TT^q)$ are the children of $t$. Two partial solutions, one at $s_1$, and the other at $s_2$, can be combined when the states agree on all $v^q_\module \in \bag^q_t$. Since we update the trackers already at introduce vertex bags, we need to take care that the values of the modules in the bag are not counted twice. For this sake, define $S^f = \bigcup_{v^q_\module \in f^{-1}(\{\one_L, \one_R\})} Y_\module \cup \bigcup_{v^q_\module \in f^{-1}(\{\all_L, \all_R\})} \module$ for all $t$-signatures $f$. This definition satisfies $X \cap \bag_t = S^f$ for all $(X, (X_L, X_R)) \in \dpsols^{\ctarget, \wtarget, \modtarget}(f)$. Then, the recurrence is given by
\begin{equation*}
	\DP_t^{\ctarget, \wtarget, \modtarget}(f) = \sum_{\substack{\ctarget_1 + \ctarget_2 = \ctarget + \cfct(S^f) \\ \wtarget_1 + \wtarget_2 = \wtarget + \wfct(S^f)}} \sum_{\modtarget_1 + \modtarget_2 = \modtarget + (|\bag^q_t| - f^{-1}(\zero))} \DP_{s_1}^{\ctarget_1, \wtarget_1, \modtarget_1}(f) \DP_{s_2}^{\ctarget_2, \wtarget_2, \modtarget_2}(f).
\end{equation*}

\begin{lem}\label{thm:cvc_mtw_prime}
	If $G^q$ is prime, then there exists a Monte-Carlo algorithm that, given a tree decomposition for $G^q$ of width at most $k$ and the sets $Y_\module$ for all $\module \in \modpartition(G)$, determines whether there is a connected vertex cover $X$ of $G$ with $\cfct(X) \leq \budget$ intersecting at least two modules of $\modpartition(G)$ in time $\Oh^*(5^k)$. The algorithm cannot give false positives and may give false negatives with probability at most $1/2$.
\end{lem}

\begin{proof}
	The algorithm samples a weight function $\wfct\colon V \rightarrow [2n]$ uniformly at random. Using the recurrences, we compute the values $\DP^{\ctarget, \wtarget, \modtarget}_{\rvertex}(\emptyset)$ modulo 4 for all $\ctarget \in [0, \cfct(V)]$, $\wtarget \in [0, \wfct(V)]$, $\modtarget \in [2, |\modpartition(G)|]$. Setting $\sols^{\ctarget, \wtarget, \modtarget} = \{X \in \rsols^{\ctarget, \wtarget, \modtarget}(G) \sep G[X] \text{ is connected}\}$, we have that $|\csols^{\ctarget, \wtarget, \modtarget}(G)| = |\csols^{\ctarget, \wtarget, \modtarget}_{\rvertex}| = \DP^{\ctarget, \wtarget, \modtarget}_{\rvertex}(\emptyset) = \sum_{X \in \rsols^{\ctarget, \wtarget, \modtarget}(G)} 2^{\cc(G[X])} \equiv_4 2|\sols^{\ctarget, \wtarget, \modtarget}|$ by \cref{thm:hom_cut}. By \cref{thm:isolation}, $\wfct$ isolates the set of optimum nice connected vertex covers intersecting at least two modules of $\modpartition(G)$ with probability at least $1/2$. If $\ctarget$ denotes the optimum value, then there exist choices of $\wtarget$ and $\modtarget$ such that $|\sols^{\ctarget, \wtarget, \modtarget}| = 1$ and hence $\DP^{\ctarget, \wtarget, \modtarget}_{\rvertex}(\emptyset) \not\equiv_4 0$. The algorithm searches for the smallest such $\ctarget$ and returns true if $\ctarget \leq \budget$. Note that if a connected vertex cover $X$ intersecting at least two modules with $\cfct(X) \leq \budget$ exists, then so does a nice one by \cref{thm:cvc_mod_structure}. If $\ctarget > \budget$, the algorithm returns false.
 
 It remains to prove the correctness of the provided recurrences and the running time of the algorithm. We first consider the running time. Since a very nice tree decomposition has polynomially many nodes and since the cost function $\cfct$ is assumed to be polynomially bounded, there are $\Oh^*(5^k)$ table entries to compute. Furthermore, it is easy to see that every recurrence can be computed in polynomial time, hence the running time of the algorithm follows. We proceed by proving the correctness of the recurrences.

 If $t$ is a leaf node, then we have that $V_t = \emptyset$ and hence $\csols_t^{\ctarget, \wtarget, \modtarget}$ can contain at most $(\emptyset, (\emptyset, \emptyset))$, and we have that $\cfct(\emptyset) = \wfct(\emptyset) = |\modprojection_V(\emptyset)| = 0$, which is checked by the recurrence.
 
 If $t$ is an introduce vertex node introducing $v^q_\module$, consider $(X,(X_L, X_R)) \in \dpsols_t^{\ctarget, \wtarget, \modtarget}(f[\qvertex \mapsto \state])$, where $f$ is some $s$-signature and $\state \in \states$. We have that $(X \setminus \module, (X_L \setminus \module, X_R \setminus \module)) \in \dpsols_s^{\ctarget', \wtarget', \modtarget'}(f)$ for $\ctarget' = \cfct(X \setminus \module)$, $\wtarget' = \wfct(X \setminus \module)$, $\modtarget' = |\modprojection_V(X \setminus \module)|$. Depending on $\state$, we argue that this sets up a bijection between $\dpsols_t^{\ctarget, \wtarget, \modtarget}(f[\qvertex \mapsto \state])$ and $\dpsols_s^{\ctarget', \wtarget', \modtarget'}(f)$. The injectivity of this map follows in general by observing that $\state$ completely determines the interaction of $(X, (X_L, X_R))$ with $\module$.
 \begin{itemize}
  \item $\state = \zero$: We have $X \cap \module = \emptyset$, which implies that $G[\module]$ does not contain an edge, as $X$ cannot be a vertex cover of $G_t$ otherwise. In this case, the mapping is essentially the identity mapping, hence the trackers do not change and it is clearly bijective. 
  \item $\state = \one_L$: We have $X \cap \module = X_L \cap \module = Y_\module \neq \module$ and $X_R \cap \module = \emptyset$. Due to $\emptyset \neq Y_\module \neq \module$, we have that $|\module| > 1$. As $X \cap \module = Y_\module$, we update the trackers according to $Y_\module$. Note that any $(X', (X'_L, X'_R)) \in \dpsols_s^{\ctarget', \wtarget', \modtarget'}(f)$ is hit by $(X' \cup Y_\module, (X'_L \cup Y_\module, X'_R)) \in \dpsols_t^{\ctarget, \wtarget, \modtarget}(f[\qvertex \mapsto \state])$, which relies on the fact that no edges incident to $\qvertex$ have been introduced yet, so that neither the vertex cover property nor consistent cut property can be violated when extending by $Y_\module$.
  \item $\state = \one_R$: analogous to the previous case.
  \item $\state = \all_L$: We have $X \cap \module = X_L \cap \module = \module$ and $X_R \cap \module = \emptyset$. Hence, we update the trackers according to $\module$. For surjectivity, we see that $(X', (X'_L, X'_R)) \in \dpsols_s^{\ctarget', \wtarget', \modtarget'}(f)$ is hit by $(X' \cup \module, (X'_L \cup \module, X'_R)) \in \dpsols_t^{\ctarget, \wtarget, \modtarget}(f[\qvertex \mapsto \state])$, which again relies on the fact that no edges incident to $\qvertex$ have been introduced yet.
  \item $\state = \all_R$: analogous to the previous case.
 \end{itemize}
 
 If $t$ is an introduce edge bag introducing edge $\{v^q_{\module_1}, v^q_{\module_2}\}$, then $\csols_t^{\ctarget, \wtarget, \modtarget} \subseteq \csols_s^{\ctarget, \wtarget, \modtarget}$ and we need to filter out all $(X,(X_L,X_R)) \in \csols_s^{\ctarget, \wtarget, \modtarget} \setminus \csols_t^{\ctarget, \wtarget, \modtarget}$. A partial solution $(X, (X_L, X_R)) \in \csols_s^{\ctarget, \wtarget, \modtarget}$ has to be filtered if and only if an edge between $\module_1$ and $\module_2$ is not covered or an edge between $X \cap \module_1$ and $X \cap \module_2$ connects both sides of the homogeneous cut. These criteria are implemented by the function $\cons$; the first case corresponds to $\cons(\state_1, \state_2) = \zero$ for all $\state_1, \state_2 \in \{\zero, \one_L, \one_R\}$ and the second case corresponds to $\cons(\state_1, \state_2) = \zero$ whenever $\state_1 \neq \zero \neq \state_2$ and the cut subscript of $\state_1$ and $\state_2$ disagrees.
 
 If $t$ is a forget vertex bag forgetting $v^q_\module$, then $\csols_t^{\ctarget, \wtarget, \modtarget} = \csols_s^{\ctarget, \wtarget, \modtarget}$ and every $(X,(X_L,X_R)) \in \csols_t^{\ctarget, \wtarget, \modtarget}$ is counted by some $\DP_s^{\ctarget, \wtarget, \modtarget}(f[v^q_\module \mapsto \mathbf{s}])$ with $\mathbf{s}$ being the appropriate state and the states are disjoint as already noted.
 
 If $t$ is a join bag, then $V_t = V_{s_1} \cup V_{s_2}$ and $\bag_t = \bag_{s_1} = \bag_{s_2} = V_{s_1} \cap V_{s_2}$. Since $G_{s_1}$ and $G_{s_2}$ are subgraphs of $G_t$, any $(X, (X_L, X_R)) \in \dpsols_t^{\ctarget, \wtarget, \modtarget}(f)$ splits into $(X^1, (X^1_L, X^1_R)) \in \dpsols_{s_1}^{\ctarget_1, \wtarget_1, \modtarget_1}(f)$ and $(X^2, (X^2_L, X^2_R)) \in \dpsols_{s_2}^{\ctarget_2, \wtarget_2, \modtarget_2}(f)$, where $X^i = X \cap V_{s_i}$, $X^i_L = X_L \cap V_{s_i}$, $X^i_R = X_R \cap V_{s_i}$ for $i \in [2]$. Since $S^f = X\cap \bag_t = X^1\cap \bag_t = X^2\cap \bag_t$, some overcounting occurs when adding up e.g.\ the costs $\ctarget_1$ and $\ctarget_2$.  This is accounted for by the equation $\ctarget_1 + \ctarget_2 = \ctarget + \cfct(S^f)$ and similarly for the weights and the number of modules hit by $X$. Vice versa, the union of the graphs $G_{s_1}$ and $G_{s_2}$ yields $G_t$, and any $(X^1, (X^1_L, X^1_R)) \in \dpsols_{s_1}^{\ctarget_1, \wtarget_1, \modtarget_1}(f)$ and $(X^2, (X^2_L, X^2_R)) \in \dpsols_{s_2}^{\ctarget_2, \wtarget_2, \modtarget_2}(f)$ must agree on $\bag_t$, since the behavior on $\bag_t$ is completely specified by $f$. Therefore, one can argue that $(X^1 \cup X^2, (X^1_L \cup X^2_L, X^1_R \cup X^2_R)) \in \dpsols_t^{\ctarget, \wtarget, \modtarget}(f)$. \qed
\end{proof}

Putting everything together, we obtain the following algorithm.

\begin{thm}
 \label{thm:cvc_mtw_algo}
 There exists a Monte-Carlo algorithm that given a tree decomposition of width at most $k$ for every prime quotient graph $H \in \primefamily(G)$, solves \CVC in time $\Oh^*(5^k)$. The algorithm cannot give false positives and may give false negatives with probability at most $1/2$.
\end{thm}

\begin{proof}
 If $|V(G)| = 1$, then $\emptyset$ is a connected vertex cover and we can always answer true. Otherwise, we first compute the sets $Y_\module$ for all $\module \in \modpartition(G)$ in time $\Oh^*(2^k)$ using \cref{thm:is_mtw_algo}. Using \cref{thm:cvc_single_module}, we first check in polynomial time if there is any connected vertex cover $X$ of $G$ contained in a single module with $\cfct(X) \leq \budget$. If yes, then we return true. Otherwise, we will proceed based on the node type of $V(G)$ in the modular decomposition of $G$.
  
 If $V(G)$ is a parallel node, i.e., $G^q$ is an independent set of size at least two, then $G$ cannot be connected, contradicting our assumption. If $V(G)$ is a series node, i.e., $G^q$ is a clique of size at least two, then we solve the problem in polynomial time using \cref{thm:cvc_mod_structure} and \cref{thm:cvc_series}, which tell us that there only $3 |\modpartition(G)|$ possible solutions to consider. 

 If $G^q$ is prime, then it remains to search for connected vertex covers intersecting at least two modules and hence we can invoke \cref{thm:cvc_mtw_prime}. This completes the proof. \qed
\end{proof}

Note that \cref{thm:cvc_mtw_algo} gets a tree decomposition for \emph{every} quotient graph as input, whereas \cref{thm:cvc_mtw_prime} only requires a tree decomposition for the topmost quotient graph. This is due to the fact that the algorithm in \cref{thm:is_mtw_algo} to compute the vertex cover $Y_\module$ of $G[\module]$ for every $\module \in \modtree(G)$ requires a decomposition for every quotient graph, but the vertex covers are enough information to enable us to solve \CVC by just considering the topmost quotient graph.

\section{Feedback Vertex Set Algorithm}
\label{sec:modtw_fvs_algo}

\newcommand{\optfamily}{\mathcal{F}_{opt}}

\newcommand{\vertex}{\mathbf{v}}
\newcommand{\indset}{\mathbf{is}}
\newcommand{\indfor}{\mathbf{if}}

\newcommand{\vtarget}{{\overline{v}}}
\newcommand{\etarget}{{\overline{e}}}
\newcommand{\markers}{D}
\newcommand{\mtarget}{{\overline{d}}}

\newcommand{\req}{\mathbf{req}}

\newcommand{\marking}{\varphi}

\newcommand{\twoind}{{\two_{\mathcal{I}}}}
\newcommand{\twoedge}{{\two_{\mathcal{E}}}}

The cut-and-count-technique applies more naturally to the dual problem \IF instead of \FVS, so we choose to study the dual problem. An instance of \IF consists of a graph $G = (V, E)$, and a budget $\budget \in \NN$, and the task is to decide whether there exists a vertex set $X \subseteq V$ with $|X| \geq \budget$ such that $G[X]$ is a forest. As our algorithm is quite technical, we only consider the case of unit costs here to reduce the amount of technical details. 

For \CVC, it was sufficient to essentially only look at the first quotient graph, because we did not have to compute \emph{connected} vertex covers for the subproblems, only usual vertex covers. However, for \IF this is not the case; here, we do need to compute an induced forest in each module $\module \in \modtree(G)$. This essentially means that we need a \emph{nested} dynamic programming algorithm; one \emph{outer dynamic program (outer DP)} along the modular decomposition tree and one \emph{inner dynamic program (inner DP)} along the tree decompositions of the quotient graphs solving the subproblems of the outer DP. 

The inner DP will again be using the cut-and-count-technique and can therefore produce erroneous results due to the randomization. We will carefully analyze where errors can occur and see that a single global sampling of an isolating weight function will be sufficient, even though some subproblems might be solved incorrectly. For this reason, the notation in this section will more closely track which node of the modular decomposition we are working on, as the setup in the \CVC algorithm would be too obfuscating here.

\subsubsection*{Notation.} $\pmodule \in \modtree(G)$ will denote the parent module and represents the current subproblem to be solved by the inner DP. The inner DP will work on the quotient graph $\pquotient = G[\pmodule] / \modpartition(G[\pmodule])$ whose vertices correspond to modules $\module \in \children(\pmodule) = \modpartition(G[\pmodule])$; associated to the quotient graph $\pquotient$ is the projection $\pproj\colon \pmodule \rightarrow V(\pquotient)$. By $\qvertex \in \pquotient$ we refer to the vertex in the quotient graph corresponding to $\module$. At times, it will be useful to not have to specify the parent module and then we say that two modules $\module_1, \module_2 \in \modtree(G)$ are \emph{siblings} if there is some $\pmodule$ such that $\module_1, \module_2 \in \children(\pmodule)$, i.e., they have the same parent. For a module $\module \in \modtree(G)$, we let $\nsib(\module)$ denote the family of sibling modules of $\module$ that are adjacent to $\module$ and we define $\nall(\module) = \{ \module' \in \modtree(G) \sep \module \cap \module' = \emptyset, E_G(\module, \module') \neq \emptyset\}$, i.e., the family of all strong modules that are adjacent to $\module$.

\subsection{Structure of Optimum Induced Forests}

We begin by studying the structure of optimum induced forests with respect to the modular decomposition. Let $\optfamily(G)$ be the family of maximum induced forests of $G$. We start by giving some definitions to capture the structure of induced forests with respect to the modular decomposition.

\begin{dfn}\label{dfn:fvs_modtw_marking}
  Let $X \subseteq V(G)$ be a vertex subset. We associate with $X$ a \emph{module-marking} $\marking_X \colon \modtree(G) \rightarrow \{\zero, \one, \twoind, \twoedge\}$ defined by
  \begin{equation*}
    \marking_X(\module) = 
    \begin{cases}
      \zero,  & \text{if $|X \cap \module| = 0$}, \\
      \one,   & \text{if $|X \cap \module| = 1$}, \\
      \twoind, & \text{if $|X \cap \module| \geq 2$ and $G[X \cap \module]$ contains no edge}, \\
      \twoedge, & \text{if $|X \cap \module| \geq 2$ and $G[X \cap \module]$ contains at least one edge}.
    \end{cases}
  \end{equation*}
\end{dfn}
We use module-markings to describe the states taken by an induced forest $X$ on the modules $M \in \modtree(G)$.
Ordering $\zero < \one < \twoind < \twoedge$, note that every module-marking $\marking_X$ is \emph{monotone} in the following sense: for all $\module_1, \module_2 \in \modtree(G)$ the inclusion $\module_1 \subseteq \module_2$ implies that $\marking_X(\module_1) \leq \marking_X(\module_2)$.

Any induced forest has to satisfy some local properties relative to the modules which are captured by the following definition.

\begin{dfn}\label{dfn:fvs_modtw_nice}
  Let $X \subseteq V(G)$ be a vertex subset. We say that $X$ is \emph{forest-nice} if for every $\module \in \modtree(G)$ the following properties hold:
  \begin{itemize}
    \item If $\marking_X(\module) = \twoind$, then $\marking_X(\nall(\module)) \subseteq \{\zero, \one\}$ and $|\nsib(\module) \cap \marking_X^{-1}(\one)| \leq 1$.
    \item If $\marking_X(\module) = \twoedge$, then $\marking_X(\nall(\module)) \subseteq \{\zero\}$.
  \end{itemize}
\end{dfn}
The ``degree-condition'' $|\nsib(\module) \cap \marking_X^{-1}(\one)| \leq 1$ deliberately only talks about the sibling modules, as we can have arbitrarily long chains of modules with $v \in \module_1 \subseteq \module_2 \subseteq \cdots \subseteq \module_\ell$, so no useful statement is possible if we would instead consider all modules.

\begin{lem}\label{thm:forest_nice}
 Every induced forest $X \subseteq V(G)$ of $G$ is forest-nice.
\end{lem}

\begin{proof}
 Consider any $\module \in \modtree(G)$ with $|X \cap \module| \geq 2$. If there were some module $\module' \in \nall(\module)$ with $|X \cap \module'| \geq 2$, then $G[X \cap (\module \cup \module')]$ contains a cycle of size 4 as all edges between $\module$ and $\module'$ exist in $G$, hence such $\module'$ cannot exist. If, additionally, $G[X \cap \module]$ contains an edge, then any $\module' \in \nall(\module)$ with $X \cap \module' \neq \emptyset$ would necessarily lead to a cycle of size 3 in $G[X \cap (\module \cup \module')]$, hence such $\module'$ cannot exist. Finally, suppose that $\marking_X(\module) = \twoind$ and two neighboring sibling modules $\module_1 \neq \module_2 \in \nsib(\module)$ with $\marking_X(\module_1) = \marking_X(\module_2) = \one$ exist. We must have $\module_1 \cap \module_2 = \emptyset$ and therefore a cycle of size 4 would exist in $G[X \cap (\module \cup \module_1 \cup \module_2)]$, which is again not possible. \qed
\end{proof}

The modular structure allows us to perform the following exchange arguments.

\begin{lem}\label{thm:forest_exchange}
 Let $X$ be an induced forest of $G$ and $\module \in \modtree(G)$.
 \begin{enumerate}
  \item If $\marking_X(\module) = \twoind$ and $Y$ is an independent set of $G[\module]$, then $(X \setminus \module) \cup Y$ is an induced forest of $G$.
  \item If $\marking_X(\module) = \twoedge$ and $Y$ is an induced forest of $G[\module]$, then $(X \setminus \module) \cup Y$ is an induced forest of $G$.
 \end{enumerate}
\end{lem}

\begin{proof}
 We set $X' = (X \setminus \module) \cup Y$ in both cases. Since $X' \setminus \module = X \setminus \module$, there cannot be any cycle in $G[X' \setminus \module]$. Also there cannot be any cycle in $G[X \cap \module] = G[Y]$ by assumption.
 \begin{enumerate}
   \item Suppose there is a cycle $C'$ in $G[X']$. By the previous arguments, we must have $C' \cap \module \neq \emptyset$ and $C' \setminus \module \neq \emptyset$. We will argue that such a cycle would give rise to a cycle $C$ in $G[X]$, contradicting the assumption that $X$ is an induced forest. Let $v_1, \ldots, v_\ell, v_1$ be the sequence of vertices visited by $C'$ and let $v_{i_1}, \ldots, v_{i_r}$ with $1 \leq i_1 < \cdots < i_r \leq \ell$ denote the vertices of $C'$ that are in $\module$. If some edge of $C'$, say $\{v_1,v_2\}$ without loss of generality, is contained in $G[X' \setminus \module]$, pick some $u \in X \cap \module$ and consider the cycle $C$ given by the vertex sequence $v_1, v_2, \ldots, v_{i_1 - 1}, u, v_{i_r + 1}, \ldots, v_\ell, v_1$; $C$ is a cycle of $G[X]$ as the edges $\{v_{i_1 - 1}, u\}$ and $\{u, v_{i_r + 1}\}$ exist in $G$, because $u, v_{i_1 - 1}, v_{i_r + 1} \in \module$. If no such edge exists in $C'$, then $C'$ is a cycle in the biclique with parts $X' \cap \module$ and $N_G(X' \cap \module)$, in particular $|C' \cap \module| \geq 2$ and $|C' \setminus \module| \geq 2$. Since $|X \cap \module| \geq 2$ by assumption and $|X \cap \module| = |X' \cap \module| \geq |C' \setminus \module| \geq 2$, it follows that $G[X]$ contains a biclique with parts of size at least two and hence $G[X]$ must contain a cycle.
   
   \item Since $X$ is forest-nice by \cref{thm:forest_nice}, $\marking_X(\module) = \twoedge$ implies that $\marking_{X'}(\nall(\module)) = \marking_X(\nall(\module)) \subseteq \{\zero\}$, and therefore $(X' \cap \module, X' \setminus \module)$ is a consistent cut of $G[X']$. Therefore any cycle $C$ in $G[X']$ must be fully contained in either $X' \cap \module$ or $X' \setminus \module$, but we ruled out each of these cases previously. Hence, $G[X']$ contains no cycle. \qed
 \end{enumerate}
\end{proof}

\cref{thm:forest_exchange} allows us to see that maximum induced forests must make locally optimal choices inside each module. We capture these local choices with the following two definitions.

\begin{dfn}\label{dfn:fvs_modtw_opt_sub}
  Let $X \subseteq V(G)$ be a vertex subset. We say that $X$ has \emph{optimal substructure} if for every $\module \in \modtree(G)$ the following properties hold:
  \begin{itemize}
    \item If $\marking_X(\module) = \twoind$, then $X \cap \module$ is a maximum independent set of $G[\module]$. 
    \item If $\marking_X(\module) = \twoedge$, then $X \cap \module$ is a maximum induced forest of $G[\module]$.
  \end{itemize}
\end{dfn}

\begin{dfn}\label{dfn:fvs_modtw_promotion}
  Let $X \subseteq V(G)$ be a vertex subset. We say that $X$ has the \emph{promotion property} if for every $\module \in \modtree(G)$ with $|X \cap \module| \geq 2$ and $\marking_X(\nall(\module)) = \{\zero\}$, we have that $X \cap \module$ is a maximum induced forest of $G[\module]$.
\end{dfn}
While we could have subsumed the promotion property as part of the definition of optimal substructure, we define it separately as it has more involved implications on the dynamic program and deserves separate care.

\begin{lem}\label{thm:fvs_opt_sub}
 Every maximum induced forest of $G$, i.e., $X \in \optfamily(G)$, has optimal substructure and the promotion property.
\end{lem}

\begin{proof}
  \cref{thm:forest_nice} already shows that $X$ is forest-nice. If $X$ would not have optimal substructure, then we can invoke \cref{thm:forest_exchange} to obtain a larger induced forest $X'$, hence $X$ would not be a maximum induced forest.
  
  We prove a strengthened exchange argument to show the promotion property. We claim that for any induced forest $X$ of $G$, module $\module \in \modtree(G)$ with $\marking_X(\module) \in \{\twoind, \twoedge\}$ and $\marking_X(\nall(\module)) \subseteq \{\zero\}$, and induced forest $Y$ of $G[\module]$, the set $X' = (X \setminus \module) \cup Y$ is again an induced forest of $G$. Suppose that $G[X']$ contains a cycle $C'$. By assumption on $X$, $C'$ cannot be contained in $G[X' \setminus \module] = G[X \setminus \module]$. By assumption on $Y$, $C'$ cannot be contained in $G[X' \cap \module] = G[Y]$. Therefore, $C'$ must intersect $X' \cap \module$ and $X' \setminus \module$ simultaneously. However, $\marking_{X'}(\nall(\module)) = \marking_X(\nall(\module)) \subseteq \{\zero\}$ implies that $(X' \cap \module, X' \setminus \module)$ is a consistent cut of $G[X']$ and hence such a cycle $C'$ cannot exist. Therefore $X'$ is also an induced forest. If an induced forest $X$ violates the promotion property, then we can invoke this exchange argument to see that $X$ cannot be a maximum induced forest. \qed
\end{proof}

Since any induced forest $X$ is forest-nice, the condition $\marking_X(\module) = \twoedge$ implies $\marking_X(\nall(\module)) \subseteq \{\zero\}$ and therefore the second condition of optimal substructure also follows from the promotion property. 

The requirement $|X \cap \module| \geq 2$ in the promotion property could also be removed. However, the dynamic programming on quotient graphs will only apply the underlying exchange argument when $|X \cap \module| \geq 2$ holds, therefore we already add this requirement here.

Note that a forest-nice vertex subset $X$ does not necessarily induce a forest as a cycle could be induced by the modules $\module \in \modpartition(G)$ with $\marking_X(\module) = \one$.

\subsection{Application of Isolation Lemma}

We will again use the cut-and-count-technique and the isolation lemma to solve \IF parameterized by modular-treewidth. However, since \IF is a maximization problem, we feel it is more natural to use a maximization version of the isolation lemma as we must closely investigate when isolation transfers to subproblems. Let us define the appropriate terminology.

\begin{dfn}
  A function $\wfct \colon U \rightarrow \ZZ$ \emph{max-isolates} a set family $\family \subseteq 2^U$ if there is a unique $S' \in \family$ with $\wfct(S') = \max_{S \in \family} \wfct(S)$, where for subsets $X$ of $U$ we define $\wfct(X) = \sum_{u \in X} \wfct(u)$.
\end{dfn}

\begin{lem}[Adapt proof of \cite{MulmuleyVV87} or \cite{TaShma15}]
  \label{thm:max_isolation}
  Let $\family \subseteq 2^U$ be a nonempty set family over a universe $U$. Let $N \in \NN$ and for each $u \in U$ choose a weight $\wfct(u) \in [N]$ uniformly and independently at random. Then
  $\PP[\wfct \text{ max-isolates } \family] \geq 1 - |U|/N$.
\end{lem}

Due to \cref{thm:forest_nice} and \cref{thm:fvs_opt_sub}, we want our algorithm to compute maximum independent sets and maximum induced forests of $G[\module]$ for every $\module \in \modtree(G)$. The computation of the maximum independent sets can be done deterministically quickly enough using \cref{thm:is_mtw_algo}. To compute the maximum induced forests however, we essentially want to recursively call our algorithm again, but the algorithm is randomized. Doing this naively and sampling a weight function for each call would exponentially decrease the success probability depending on the depth of the modular decomposition tree.

To circumvent this issue, we sample a global weight function only once and let the subproblems inherit this weight function, observing that for all ``important'' subproblems the inherited weight function is max-isolating if the global weight function is (for appropriate choices of set families).

We define $\optfamily(G, \state)$, where $\state \in \{\zero, \one, \twoind, \twoedge\}$, as the family of maximum sets $X$ subject to $G[X]$ being a forest and $\marking_X(V(G)) \leq \state$. Hence, we have that $\optfamily(G, \twoedge) = \optfamily(G)$ and $\optfamily(G, \twoind)$ is the family of maximum independent sets of $G$ and $\optfamily(G, \one)$ is the family of singleton sets.

\begin{lem}\label{thm:fvs_descendant_isolation}
  Let $N \in \NN$ and assume that $\wfct \colon V(G) \rightarrow [N]$ is a weight function that max-isolates $\optfamily(G)$. Let $X \in \optfamily(G)$ be the set that is max-isolated by $\wfct$. For every $\module \in \modtree(G)$, we have that $\wfct\big|_\module$ max-isolates $X \cap \module$ in $\optfamily(G[\module], \marking_X(\module))$.
\end{lem}

\begin{proof}
  $X$ has optimal substructure due to \cref{thm:fvs_opt_sub}, therefore we have $X \cap \module \in \optfamily(G[\module], \marking_X(\module))$ for all $\module \in \modtree(G)$. Suppose there is some $\module \in \modtree(G)$ such that $\wfct\big|_\module$ does not max-isolate $\optfamily(G[\module], \marking_X(\module))$, then there is some $X \cap \module \neq Y \in \optfamily(G[\module], \marking_X(\module))$ with $\wfct(Y) \geq \wfct(X \cap \module)$. By \cref{thm:forest_exchange}, $X' = (X \setminus \module) \cup Y$ must satisfy $X' \in \optfamily(G)$, $X' \neq X$, and $\wfct(X') \geq \wfct(X)$. However, then $\wfct$ cannot max-isolate $X$ in $\optfamily(G)$. \qed
\end{proof}

We remark that the previous lemma allows for the possibility that, e.g. $\wfct\big|_\module$ max-isolates $\optfamily(G[\module], \twoind)$, but $\wfct\big|_\module$ does not max-isolate $\optfamily(G[\module], \twoedge) = \optfamily(G[\module])$, which can lead to our algorithm not finding an optimum induced forest for this subinstance. 

\subsection{Detecting Acyclicness}

Let us describe how to check whether a forest-nice subset $X$ induces a forest. The property of being forest-nice essentially allows us to only consider the induced subset on a quotient graph which we then handle by lifting cut-and-count. The property of being forest-nice is a global property in the sense that it considers the whole modular decomposition tree. We first introduce a local version of forest-nice that only considers the children of a parent module $\pmodule \in \modint(G)$:

\newcommand{\squotient}{{\widetilde{G}^q}}

\begin{dfn}\label{dfn:fvs_modtw_nice_local}
  Let $\pmodule \in \modint(G)$, $\squotient$ be a subgraph of $\pquotient$, and $X \subseteq \pmodule$ with $X^q := \pproj(X) \subseteq V(\squotient)$, we say that $X$ is \emph{$\pmodule$-forest-nice with respect to $\squotient$}, if the following properties hold for all $\qvertex \in V(\squotient)$:
  \begin{itemize}
    \item If $\marking_X(\module) = \twoind$, then $\deg_{\squotient[X^q]}(\qvertex) \leq 1$ and $\marking_X(\module') \in \{\zero, \one\}$ for all $v^q_{\module'} \in N_{\squotient}(\qvertex)$.
    \item If $\marking_X(\module) = \twoedge$, then $\marking_X(\module') = \zero$ for all $v^q_{\module'} \in N_{\squotient}(\qvertex)$.
  \end{itemize}
  In the case $\squotient = \pquotient$, we simply say that $X$ is \emph{$\pmodule$-forest-nice}.
\end{dfn}

As the (very nice) tree decomposition of $\pquotient$ adds edges one-by-one, we need to account for changes in the neighborhoods of vertices in the local definition of forest-niceness via $\squotient$. Otherwise, \cref{dfn:fvs_modtw_nice_local} is essentially the same definition as \cref{dfn:fvs_modtw_nice}, but only considering the child modules of $\pmodule$.  In particular, if $X$ is forest-nice, then $X \cap \pmodule$ is $\pmodule$-forest-nice for all $\pmodule \in \modint(G)$.

The next lemma essentially shows that in a $\pmodule$-forest-nice set $X$ no cycles intersecting some module $\module \in \children(\pmodule)$ in more than one vertex exist, hence all possible cycles can already be seen in the quotient graph.

\begin{lem}\label{thm:forest_quotient}
	Let $\pmodule \in \modint(G)$ and $X \subseteq \pmodule$ be $\pmodule$-forest-nice and suppose that $G[X \cap \module]$ is a forest for all modules $\module \in \children(\pmodule)$ and define $X^q = \pproj(X)$. Then, $G[X]$ is a forest if and only if $\pquotient[X^q]$ is a forest.
\end{lem}

\begin{proof}
  The graph $\pquotient[X^q]$ can be considered a subgraph of $G[X]$, so if $\pquotient[X^q]$ is not a forest, then neither is $G[X]$.
  
  For the other direction, suppose that $G[X]$ contains a cycle $C$. It cannot be that $C \subseteq X \cap \module$ for some $\module \in \children(\pmodule)$, since $G[X \cap \module]$ contains no cycle by assumption. It also cannot be that $G[C \cap M]$ contains an edge for some $\module \in \children(\pmodule)$, since $\pmodule$-forest-nice would then imply that $C$ is contained in $\module$, which we just ruled out. If $|C \cap \module| \geq 2$ for some $\module \in \children(\pmodule)$, then $\pmodule$-forest-nice implies that at most one neighboring sibling module $\module'$ is intersected by $C$ and $|C \cap \module'| \geq 1$, but since $G[C \cap \module]$ cannot contain an edge, this means that the vertices in $C \cap \module$ must have degree one in $C$, so $C$ cannot be a cycle.
  Finally, we must have $|C \cap \module| \leq 1$ for all $\module \in \children(\pmodule)$, but any such cycle $C$ clearly gives rise to a cycle $C^q = \pproj(C)$ in $G^q[X^q]$, too. \qed
\end{proof}

\begin{lem}[Lemma~4.5 in \cite{CyganNPPRW22}]
  \label{thm:forest_check}
  Let $G$ be a graph with $n$ vertices and $m$ edges. Then, $G$ is a forest if and only if $\cc(G) \leq n - m$ if and only if $\cc(G) = n - m$.
\end{lem}
·
One could use the \emph{marker technique} already used by Cygan et al.~\cite{CyganNPPRW22} for the treewidth-parameterization together with \cref{thm:forest_check} to obtain a cut-and-count algorithm, but the marker technique results in several further technical details to take care of. The marker technique can be avoided by working modulo higher powers of two instead of only modulo two, which was also done by Nederlof et al.~\cite{NederlofPSW20} when applying cut-and-count to edge-based problems parameterized by treedepth. We also do so, to obtain a cleaner presentation of our algorithm. 

\begin{lem}
  \label{thm:forest_hom_cuts}
  Let $\pmodule \in \modint(G)$ and $X \subseteq \pmodule$ be $\pmodule$-forest-nice and suppose that $G[X \cap \module]$ is a forest for all modules $\module \in \children(\pmodule)$. Let $X^q = \pproj(X)$ and let $\vtarget = |X^q|$ and $\etarget = |E(\pquotient[X^q])|$. Then, $G[X]$ is a forest if and only if $|\{(X_L, X_R) \sep (X, (X_L, X_R)) \in \homcuts{\pmodule}(G)\}| \not\equiv_{2^{\vtarget - \etarget + 1}} 0$.
\end{lem}

\begin{proof}
  By \cref{thm:hom_cut}, we have that $|\{(X_L, X_R) \sep (X, (X_L, X_R)) \in \homcuts{\pmodule}(G)\}| = 2^{\cc(\pquotient[X^q])}$. By \cref{thm:forest_check}, we see that $\pquotient[X^q]$ is a forest if and only if $|\{(X_L, X_R) \sep (X, (X_L, X_R)) \in \homcuts{\pmodule}(G)\}| \neq 0 \mod 2^{\vtarget^q - \etarget^q + 1}$. The lemma then follows via \cref{thm:forest_quotient}. \qed
\end{proof}

\subsection{Outer DP: Candidate Forests}

\newcommand{\fixv}[1]{{Y^{\one}_{#1}}}
\newcommand{\fixis}[1]{{Y^{\twoind}_{#1}}}
\newcommand{\fixif}[1]{{Y^{\twoedge}_{#1}}}
\newcommand{\nv}[1]{{\#^{\vertex}_{#1}}}
\newcommand{\nis}[1]{{\#^{\indset}_{#1}}}
\newcommand{\nif}[1]{{\#^{\indfor}_{#1}}}
\newcommand{\cv}[1]{{c^{\one}_{#1}}} 
\newcommand{\cis}[1]{{c^{\twoind}_{#1}}}
\newcommand{\cif}[1]{{c^{\twoedge}_{#1}}}
\newcommand{\wv}[1]{{w^{\one}_{#1}}}
\newcommand{\wis}[1]{{w^{\twoind}_{#1}}}
\newcommand{\wif}[1]{{w^{\twoedge}_{#1}}}

\newcommand{\dpfam}{\mathcal{R}}
\newcommand{\dpfamcuts}{\mathcal{Q}}
\newcommand{\targets}{{\ctarget, \wtarget, \vtarget, \etarget}}

\newcommand{\cwpairs}[1]{P_{#1}}

Fix an \IF instance $(G = (V, E), \budget)$ and a weight function $\wfct \colon V \rightarrow [N]$ throughout this section. To solve \IF parameterized by modular-treewidth, we perform dynamic programming in two ways: we proceed bottom-up along the modular decomposition tree of $G$ and to compute the table entries for the node corresponding to module $\pmodule \in \modint(G)$, we use the tables of the children $\children(\pmodule) = \modpartition(G[\pmodule])$ and perform dynamic programming along the tree decomposition of the associated quotient graph $\pquotient = G[\pmodule] / \modpartition(G[\pmodule])$.

For every module $\module \in \modtree(G)$, we have the following data precomputed:
\begin{itemize}
  \item a singleton set $\fixv{\module}$ in $\module$ that maximizes $\wfct(\fixv{\module})$ and its weight $\wv{\module} = \wfct(\fixv{\module})$,
  \item a maximum independent set $\fixis{\module}$ of $G[\module]$ that maximizes $\wfct(\fixis{\module})$, the size $\cis{\module} = |\fixis{\module}|$ and the weight $\wis{\module} = \wfct(\fixis{\module})$ of such an independent set.
\end{itemize}
The vertex data can clearly be precomputed in polynomial time and the independent set data can be precomputed in time $\Oh^*(2^{\modtw(G)})$ by running the \IS algorithm from \cref{thm:is_mtw_algo}.

\subsubsection*{Candidate Forests.} We will recursively define for each module $\pmodule \in \modtree(G)$, the \emph{$\pmodule$-candidate forest} $\fixif{\pmodule}$ (which depends on the fixed weight function $\wfct$). Among all induced forests $X$ of $G[\pmodule]$ found by the algorithm, the forest $\fixif{\pmodule}$ lexicographically maximizes $(|X|, \wfct(X))$. Due to the randomization in the cut-and-count-technique however, it can happen that $\fixif{\pmodule}$ is not necessarily a maximum induced forest of $G[\pmodule]$. We will see that if we sampled an isolating weight function $\wfct$, then no errors will occur for the ``important'' subproblems, hence still allowing us to find a maximum induced forest of the whole graph.
The definition of $\fixif{\pmodule}$ is mutually recursive with the definition of the solution family that will be defined afterwards. 

\subsubsection*{Properties of Candidate Forests.} We highlight several properties of the candidate forests that are important for the algorithm.
\begin{itemize}
 \item The base case is given by $\fixif{\{v\}} = \{v\}$ for all $v \in V(G)$.
 \item $\fixif{\pmodule}$ is an induced forest of $G[\pmodule]$.
 \item If $G[\pmodule]$ contains no edge, then $\fixif{\pmodule} = \fixis{\pmodule}$.
 \item If $G[\pmodule]$ contains an edge, then $|\fixif{\pmodule}| > |\fixis{\pmodule}|$.
\end{itemize}

Given $\fixif{\module}$ for all $\module \in \children(\pmodule)$, we can describe how to compute $\fixif{\pmodule}$. This step depends on which kind of node $\pmodule$ corresponds to in the modular decomposition. We first handle the degenerate cases of a parallel or series node and then proceed with the much more challenging case of a prime node.

\subsubsection{Computing Candidate Forests in Parallel and Series Nodes.}\label{sec:fvs_parallel_series}

  If $\pmodule \in \modint(G)$ is a parallel node, i.e., $\pquotient$ is an independent set, then \cref{thm:forest_nice} and \cref{thm:fvs_opt_sub} tell us to simply take a maximum induced forest inside each child module $\module \in \children(\pmodule)$. Hence, we set $\fixif{\pmodule} = \bigcup_{\module \in \children(\pmodule)} \fixif{\module}$ and accordingly $\cif{\pmodule} = \sum_{\module \in \children(\pmodule)} \cif{\module}$ and $\wif{\pmodule} = \sum_{\module \in \children(\pmodule)} \wif{\module}$.

	If $\pmodule \in \modint(G)$ is a series node, then we first analyze the structure of maximum induced forests with respect to a series node.
\begin{lem}\label{thm:fvs_modtw_series}
  Let $\pmodule \in \modint(G)$ and $X$ be a maximum induced forest of $G[\pmodule]$. If $\pmodule$ is a series module, i.e., the quotient graph $\pquotient$ is a clique, then one of the following statements holds:
  \begin{itemize}
    \item $X \subseteq \module$ for some $\module \in \children(\pmodule)$ and $X$ is a maximum induced forest of $G[\module]$.
    \item $X \subseteq \module_1 \cup \module_2$ for some $\module_1 \neq \module_2 \in \children(\pmodule)$ and $X \cap \module_1$ is a maximum independent set of $G[\module_1]$ and $|X \cap \module_2| = 1$.
 \end{itemize}
\end{lem}

\begin{proof}
  Suppose that $X$ intersects three different modules in $\children(\pmodule)$, since they are all adjacent $X$ would induce a triangle. Hence, $X$ can intersect at most two different modules. By \cref{thm:forest_nice} and \cref{thm:fvs_opt_sub}, $X$ is forest-nice, has optimal substructure and satisfies the promotion property. If $X$ intersects only a single module $\module$, then the first statement follows due to the promotion property. If $X$ intersects two modules, then the second statement follows due to $X$ being forest-nice and optimal substructure. \qed
\end{proof}

 Given the maximum independent sets $\fixis{\module}$ for all $\module \in \children(\pmodule)$, we can in polynomial time compute an optimum induced forest $\tilde{Y}_\pmodule$ of $G[\pmodule]$ subject to the second condition in \cref{thm:fvs_modtw_series}. We compare the induced forests $\tilde{Y}_\pmodule$ and all $\fixif{\module}$ for all $\module \in \children(\pmodule)$ lexicographically by their cost and weight and, motivated by \cref{thm:fvs_modtw_series}, we let $\fixif{\pmodule}$ be the winner of this comparison.

\subsubsection{Computing Candidate Forests in Prime Nodes.}\label{sec:fvs_prime_candidate}

 To compute the $\pmodule$-candidate forest $\fixif{\pmodule}$ when $\pmodule$ is a prime node, we will use the cut-and-count-technique and dynamic programming along the given tree decomposition of the quotient graph $\pquotient$. Before going into the details of the dynamic programming, we will give the necessary formal definitions to describe the partial solutions of the dynamic programming and the subproblem that has to be solved. This will already allow us to define the induced forest $\fixif{\pmodule}$ and prove the correctness of the outer loop involving the modular decomposition. We first introduce some ``local'' versions of \cref{dfn:fvs_modtw_opt_sub} and \cref{dfn:fvs_modtw_promotion}.

\begin{dfn}\label{dfn:fvs_modtw_opt_sub_local}
  Let $\pmodule \in \modint(G)$ and $X \subseteq \pmodule$, we say that $X$ has \emph{$\pmodule$-substructure} if for all $\module \in \children(\pmodule)$ we have that $\marking_X(\module) \neq \zero$ implies $X \cap \module = Y^{\marking_X(\module)}_{\module}$.
\end{dfn}
Comparing the definition of \emph{$\pmodule$-substructure} to \emph{optimal substructure}, we see that in $\pmodule$-substructure we only consider the child modules and require the choice of a specified vertex, maximum independent set, or induced forest, respectively. Note that due to the previously discussed issue, $\fixif{\module}$ does not necessarily need to be a maximum induced forest.
\begin{dfn}\label{dfn:fvs_modtw_promotion_local}
  Let $\pmodule \in \modint(G)$ and $X \subseteq \pmodule$, we say that $X$ satisfies the \emph{$\pmodule$-promotion property} if for all modules $\module \in \children(\pmodule)$ with $|X \cap \module| \geq 2$ and $\marking_X(\nsib(\module)) = \{\zero\}$ it holds that $X \cap \module = \fixif{\module}$.
\end{dfn}
\cref{dfn:fvs_modtw_promotion_local}, unlike \cref{dfn:fvs_modtw_nice_local}, does not need to account for the current subgraph of $\pquotient$ as promotion is only checked for modules that have already been forgotten by the tree decomposition, i.e., all incident edges have already been added, and for non-introduced modules $\module$, we simply have $X \cap \module = \emptyset$.

We can now define the solution family considered by our algorithm.

\begin{dfn}
  The family $\dpfam_\pmodule$ consists of all $X \subseteq \pmodule$ such that $X$ is $\pmodule$-forest-nice wrt.\ $\pquotient$, has $\pmodule$-substructure, and satisfies the $\pmodule$-promotion property. Given $\ctarget \in [0, |\pmodule|]$, $\wtarget \in [0, \wfct(\pmodule)]$, $\vtarget \in [0, |\children(\pmodule)|]$, $\etarget \in [0, \vtarget - 1]$, the family $\dpfam_\pmodule^{\targets}$ consists of all $X \in \dpfam_\pmodule$ with
  \begin{itemize}
    \item $|X| = \ctarget$ and $\wfct(X) = \wtarget$,
    \item $|X^q| = \vtarget$ and $|E(\pquotient[X^q])| = \etarget$, where $X^q = \pproj(X)$.
  \end{itemize}
  We also define $\sols_\pmodule^\targets = \{X \in \dpfam_\pmodule^\targets \sep G[X] \text{ is a forest}\}$.
\end{dfn}

By pairing elements of $\dpfam_\pmodule^\targets$ with homogeneous cuts, we can use the cut-and-count-technique to decide whether $\sols_\pmodule^\targets$ is empty or not.

\begin{dfn}
  The family $\dpfamcuts_\pmodule$ consists of all $(X, (X_L, X_R)) \in \homcuts{\pmodule}(G)$ with $X \in \dpfam_\pmodule$.
  Similarly, $\dpfamcuts_\pmodule^{\targets}$ consists of all $(X, (X_L, X_R)) \in \homcuts{\pmodule}(G)$ with $X \in \dpfam_\pmodule^{\targets}$.
\end{dfn}

The crucial property of $\dpfamcuts_\pmodule^{\targets}$ is given by the following lemma.

\begin{lem}\label{thm:fvs_nonzero_implies_forest}
  Let $\pmodule \in \modint(G)$. It holds that $|\dpfamcuts_\pmodule^{\targets}| \equiv_{2^{\vtarget - \etarget + 1}} 2^{\vtarget - \etarget} |\sols_\pmodule^{\targets}|$.
\end{lem}

\begin{proof}
  Consider any $X \in \dpfam_\pmodule^\targets$ and let $X^q = \pproj(X)$. If $G[X]$ is a forest, then so is $\pquotient[X^q]$ by \cref{thm:forest_quotient} and we have that $X$ contributes exactly $2^{\vtarget - \etarget}$ objects to $\dpfamcuts_\pmodule^\targets$ by \cref{thm:hom_cut} and \cref{thm:forest_check}. By \cref{thm:forest_hom_cuts}, we see that if $G[X]$ is not a forest, then $X$ contributes a multiple of $2^{\vtarget - \etarget + 1}$ objects to $\dpfamcuts_\pmodule^\targets$, which therefore cancel. \qed
\end{proof}

From the sets $\dpfamcuts_\pmodule^{\targets}$ for a fixed $\pmodule \in \modint(G)$, we can finally give the recursive definition of the $\pmodule$-candidate forest $\fixif{\pmodule}$.

\begin{dfn}
  \label{dfn:recursive_candidate_forest}
  Let $\pmodule \in \modint(G)$ such that $\pquotient$ is prime. The set of \emph{attained cost-weight-pairs} $\cwpairs{\pmodule}$ consists of all pairs $(\ctarget, \wtarget)$ such that there exist $\vtarget$ and $\etarget$ with $|\dpfamcuts_\pmodule^{\targets}| \not\equiv_{2^{\vtarget - \etarget + 1}} 0$. We denote the lexicographic maximum pair in $\cwpairs{\pmodule}$ by $(\ctarget_{max}, \wtarget_{max})$. \cref{thm:fvs_nonzero_implies_forest} guarantees the existence of an induced forest $Y$ of $G[\pmodule]$ with $|Y| = \ctarget_{max}$ and $\wfct(Y) = \wtarget_{max}$. If $\ctarget_{max} > |\fixis{\pmodule}|$, then the \emph{$\pmodule$-candidate forest} $\fixif{\pmodule}$ is an arbitrary induced forest among these, else we greedily extend $\fixis{\pmodule}$ by some vertices, without introducing cycles, to obtain $\fixif{\pmodule}$. We set $\cif{\pmodule} = |\fixif{\pmodule}|$ and $\wif{\pmodule} = \wfct(\fixif{\pmodule})$.
\end{dfn}

The algorithm does not know the exact set $\fixif{\pmodule}$, hence no issue is caused by the arbitrary choice, but the algorithm knows the values $\cif{\pmodule}$ and $\wif{\pmodule}$. The set $\fixif{\pmodule}$ is only used for the analysis of the algorithm. We will see that the choice of $\fixif{\pmodule}$ is unique when $\wfct\restrict{\pmodule}$ isolates the optimum induced forests of $G[\pmodule]$, else the choice might not be unique. Only in the latter case can $\ctarget_* \leq |\fixis{\pmodule}|$ occur, but since $\pquotient$ is prime, the graph $G[\pmodule]$ must contain some edges and hence there exists a larger induced forest that is not an independent set.

Note that $\fixif{\pmodule}$ is always an induced forest, but $G[\fixif{\pmodule}]$ does not necessarily contain an edge, i.e., $\fixif{\pmodule}$ may be an independent set or even a single vertex if $\pquotient$ is a parallel node or singleton node. This means that for some $X \subseteq V(G)$ with $X \cap \pmodule = \fixif{\pmodule}$, we only know $\marking_X(\pmodule) \leq \twoedge$ and not necessarily $\marking_X(\pmodule) = \twoedge$.

The complete outer DP is summarized in \cref{algo:outer_dp}.
 \begin{algorithm}
   \uIf{$\pmodule$ is a parallel node}{
     $\fixif{\pmodule} := \bigcup_{\module \in \children(\pmodule)} \fixif{\module}$\;
   }
   \uElseIf{$\pmodule$ is a series node}{
     pick $Y_1$ among all $\fixif{\module}$, $\module \in \children(\pmodule)$, to lex.\ maximize $(|Y_1|, \wfct(Y_1))$\;
     pick $Y_2 \in \{\fixis{\module_1} \cup \fixv{\module_2} \sep \module_1 \neq \module_2 \in \children(\pmodule)\}$ to lex.\ max.\ $(|Y_2|, \wfct(Y_2))$\;
     pick $\fixif{\pmodule}$ as a winner of the lex.\ comparison $(|Y_1|, \wfct(Y_1))$ vs $(|Y_2|, \wfct(Y_2))$\;
   }
   \Else{
     compute $|\csols^{\targets}_\pmodule|$ for all $\targets$ using treewidth-based DP\;
     construct $\cwpairs{\pmodule} = \left\{(\ctarget, \wtarget) \sep \text{ there are $\vtarget, \etarget$ such that } |\csols^{\targets}_\pmodule| \not\equiv_{2^{\vtarget - \etarget + 1}} 0 \right\}$\;
     let $(\ctarget_{max}, \wtarget_{max}) \in \cwpairs{\pmodule}$ be the lexicographic maximum\;
     \uIf{$\ctarget_{max} > |\fixis{\pmodule}|$}{
     		pick any $\fixif{\pmodule}$ among induced forests $Y$ of $G[\pmodule]$ with $|Y| = \ctarget_{max}$ and $\wfct(Y) = \wtarget_{max}$\;
     }
     \Else{
     		obtain $\fixif{\pmodule}$ by greedily extending $\fixis{\pmodule}$ by vertices without creating cycles\;
     }
   }
   \caption{Outer DP to compute $\fixif{\pmodule}$.}
   \label{algo:outer_dp}
 \end{algorithm}

\subsubsection{Correctness of Outer DP.}

Assuming an algorithm that computes the values $|\dpfamcuts_{\pmodule}^{\targets}|$ for all prime $\pquotient$ and all $\targets$, we obtain an algorithm that implicitly computes $\fixif{\pmodule}$ for all $\pmodule \in \modtree(G)$ by starting with $\fixif{\{v\}} = \{v\}$ for all $v \in V$ and performs bottom-up dynamic programming along the modular decomposition tree using the appropriate algorithm based on the node type. While the precise set $\fixif{\pmodule}$ is not known to the algorithm, it knows the value $\cif{\pmodule} = |\fixif{\pmodule}|$. The algorithm returns positively if $\cif{V} \geq \budget$ and negatively otherwise. As we ensure that $\fixif{\pmodule}$ is an induced forest for all $\pmodule \in \modtree(G)$, the algorithm does not return false positives. The next lemma concludes the discussion of the outer DP and implies that the algorithm answers correctly assuming that the weight function $\wfct$ isolates the maximum induced forests of $G$.

\begin{lem}[Main Correctness Lemma]\label{thm:fvs_mtw_main_correctness}
  Suppose that $\wfct$ max-isolates $X_*$ in $\optfamily(G)$. The following properties hold for all $\pmodule \in \modtree(G)$:
  \begin{enumerate}
    \item $\state(\pmodule) := \marking_{X_*}(\pmodule) \neq \zero$ implies that $X_* \cap \pmodule = Y^{\state(\pmodule)}_\pmodule$, ($\pmodule$-substructure for all $\pmodule$)
    \item $\state(\pmodule) = \marking_{X_*}(\pmodule) = \twoedge$ implies that $\fixif{\pmodule}$ is a maximum induced forest of $G[\pmodule]$,
    \item $\state(\pmodule) = \marking_{X_*}(\pmodule) = \twoedge$ implies that $X_* \cap \pmodule \in \dpfam_{\pmodule}$.
  \end{enumerate}
\end{lem}

\begin{proof}
  Notice that for singleton modules only the first property is relevant and is trivially true. By \cref{thm:forest_nice} and \cref{thm:fvs_opt_sub}, $X_*$ is forest-nice, has optimal substructure and the promotion property. By \cref{thm:fvs_descendant_isolation}, it follows that $\wfct\big|_{\pmodule}$ max-isolates $X_* \cap \pmodule$ in $\optfamily(G, \state(\pmodule))$ for all $\pmodule \in \modtree(G)$. Since $X_*$ is forest-nice, $X_* \cap \pmodule$ must be $\pmodule$-forest-nice for all $\pmodule \in \modint(G)$ as the quotient graph $\pquotient$ captures when two sibling modules $\module, \module' \in \children(\pmodule)$ are adjacent.
  
  We proceed by proving the first property whenever $\state(\pmodule) \neq \twoedge$. Fix some $\pmodule$ with $\state(\pmodule) \in \{\one, \twoind\}$. We have $X_* \cap \pmodule, Y^{\state(\pmodule)}_{\pmodule} \in \optfamily(G[\pmodule], \state(\pmodule))$ by optimal substructure and definition. By choice of $Y^{\state(\pmodule)}_{\pmodule}$, we have that $\wfct(X_* \cap \pmodule) \leq \wfct(Y^{\state(\pmodule)}_{\pmodule})$. By max-isolation of $X_* \cap \pmodule$ it follows that $\wfct(X_* \cap \pmodule) = \wfct(Y^{\state(\pmodule)}_{\pmodule})$ and even $X_* \cap \pmodule = Y^{\state(\pmodule)}_{\pmodule}$.
    
  The remainder of the proof is an induction along the modular decomposition tree, as the base case we consider modules $\pmodule \in \modint(G)$ with $\state(\pmodule) = \twoedge$ and $\state(\module) \neq \twoedge$ for all $\module \in \children(\pmodule)$. For the base case, we have already shown that $X_* \cap \module = Y^{\state(\module)}_\module$ for all $\module \in \children(\pmodule)$, hence $X_* \cap \pmodule$ has $\pmodule$-substructure in this case. 
  
  We continue with the $\pmodule$-promotion property in the base case. Suppose it is violated for some $\module \in \children(\pmodule)$, i.e., $\state(\nsib(\module)) \subseteq \{\zero\}$ and $X_* \cap \module = \fixis{\module} \neq \fixif{\module}$ (using $\pmodule$-substructure). By definition of $\fixif{\module}$, we have that $\fixif{\module}$ is an induced forest of $G[\module]$ and $\fixif{\module} \neq \fixis{\module}$ if and only if $|\fixif{\module}| > |\fixis{\module}|$. We claim that $\module$ must also violate the promotion property of $X_*$. For this it remains to establish that $\state(\nall(\module)) \subseteq \{\zero\}$. We have $\state(\nsib(\module)) \subseteq \{\zero\}$ by assumption, this shows that $\state(\module') = \zero$ for all $\module' \in \nall(\module)$ with $\module' \subseteq \pmodule$. Every module $\module' \in \nall(\module)$ with $\module' \not\subseteq \pmodule$ must be disjoint from $\pmodule$ and hence $\module' \in \nall(\pmodule)$ which implies that $\state(\module') = \zero$ since $X_*$ is forest-nice.
  
  For the base case, we have now established that $X_* \cap \pmodule \in \dpfam_{\pmodule}$, as we have verified that $X_* \cap \pmodule$ is $\pmodule$-forest-nice wrt.\ $\pquotient$, has $\pmodule$-substructure, and has the $\pmodule$-promotion property. We can now proceed by showing the first and second property for the base case when $\state(\pmodule) = \twoedge$. Note that the second property follows from the first one by optimal substructure of $X_*$, so we only have to prove the first property.
  
  If $\pquotient$ is a parallel or series node, then the analysis in \cref{sec:fvs_parallel_series} shows that $\fixif{\pmodule} \in \optfamily(G[\pmodule])$. Since also $X_* \cap \pmodule \in \optfamily(G[\pmodule])$ and both maximize their weight (by definition and max-isolation), the isolation of $X_* \cap \pmodule$ implies $X_* \cap \pmodule = \fixif{\pmodule}$. If $\pquotient$ is a prime node, then we set $X^q_* = \pproj(X_* \cap \pmodule)$ and $\ctarget_* = |X_* \cap \pmodule|$, $\wtarget_* = \wfct(X_* \cap \pmodule)$, $\vtarget_* = |X_*^q|$, $\etarget_* = |E(\pquotient[X_*^q])|$. Hence, we have that $X_* \cap \pmodule \in \dpfam^{\ctarget_*, \wtarget_*, \vtarget_*, \etarget_*}_\pmodule$ and $X_* \cap \pmodule \in \sols^{\ctarget_*, \wtarget_*, \vtarget_*, \etarget_*}_\pmodule$. By max-isolation of $X_* \cap \pmodule$, we therefore have $|\sols^{\ctarget_*, \wtarget_*, \vtarget_*, \etarget_*}_\pmodule| = 1$ and \cref{thm:fvs_nonzero_implies_forest} shows that $|\dpfamcuts^{\ctarget_*, \wtarget_*, \vtarget_*, \etarget_*}_\pmodule| \not\equiv_{2^{\vtarget_* - \etarget_* + 1}} 0$, so $(\ctarget_*, \wtarget_*) \in \cwpairs{\pmodule}$. Also, $(\ctarget_*, \wtarget_*)$ must be the lexicographic maximum in $\cwpairs{\pmodule}$. Therefore, \cref{dfn:recursive_candidate_forest} must pick $\fixif{\pmodule} = X_* \cap \pmodule$; we must have $|X_* \cap \pmodule| > |\fixis{\pmodule}|$, since $G[\pmodule]$ contains an edge and $X_* \cap \pmodule \in \optfamily(G[\pmodule])$. This concludes the proof of the base case.
  
  Now, when proving the three properties for some $\pmodule \in \modint(G)$, we can inductively assume that they hold for all $\module \in \children(\pmodule)$. The argument for the inductive step is essentially the same as for the base case, however $\state(\module) = \twoedge$ can occur now, but for this case we can apply the already proven properties. The first two properties for the child modules allow us to establish $X_* \cap \pmodule \in \dpfam_\pmodule$ even in the inductive step. From that point on, the same argument considering the sets $\dpfamcuts^{\targets}_\pmodule$ can be followed to also obtain the first and second property for $\pmodule$. \qed
\end{proof}

\subsection{Dynamic Programming along Tree Decomposition}\label{sec:fvs_prime_dp}

\newcommand{\iso}{\mathrm{iso}}

We now need to show how to compute the values $|\dpfamcuts_{\pmodule}^{\targets}|$ modulo $2^{\vtarget - \etarget + 1}$ for all $\targets$ when $\pquotient$ is prime, from which we can then obtain the $\pmodule$-candidate forest $\fixif{\pmodule}$ and proceed through the modular decomposition. We will compute these values by performing dynamic programming along the tree decomposition of the quotient graph $\pquotient = G[\pmodule] / \children(\pmodule)$.

\subsubsection*{Precomputed Data.} Let us fix some $\pmodule \in \modint(G)$ and recap the data that is available from solving the previous subproblems. For every $\module \in \children(\pmodule)$, we know the values
\begin{itemize}
  \item $\cv{\module} = |\fixv{\module}| = 1$, $\wv{\module} = \wfct(\fixv{\module})$,
  \item $\cis{\module} = |\fixis{\module}|$, $\wis{\module} = \wfct(\fixis{\module})$,
  \item $\cif{\module} = |\fixif{\module}|$, $\wif{\module} = \wfct(\fixif{\module})$.
\end{itemize}
The algorithm also knows the sets $\fixv{\module}$ and $\fixis{\module}$, but not the sets $\fixif{\module}$, they will be used in the analysis however. Furthermore, we are given a tree decomposition $(\TT^q_{\pmodule}, (\bag^q_t)_{t \in V(\TT^q_{\pmodule})})$ of the quotient graph $\pquotient$ of width $k$ which can be assumed to be very nice by \cref{thm:very_nice_tree_decomposition}. To lighten the notation, we do not annotate the bags $\bag^q_t$ with $\pmodule$, but keep in mind that there is a different tree decomposition for each quotient graph.

\begin{dfn}
  Let $t \in V(\TT^q_\pmodule)$ be a node of the tree decomposition $\TT^q_\pmodule$. The set of \emph{relaxed solutions} $\rsols_{t, \pmodule}$ consists of the vertex subsets $X \subseteq V_t = \pprojinv(V^q_t)$ that satisfy the following properties:
  \begin{itemize}
   \item $X$ is $\pmodule$-forest-nice with respect to $G^q_t$, 
   \item $X$ has $\pmodule$-substructure,
   \item $\forall \module \in \children(\pmodule)\colon \marking_X(\module) = \twoedge \rightarrow (\module \subseteq V_t \setminus \bag_t \vee \text{$G[\module]$ is a clique of size at least 2})$,
   \item $\forall \qvertex \in V^q_t \setminus \bag^q_t\colon (|X \cap \module| \geq 2 \wedge \deg_{G^q_t[\pproj(X)]}(v^q_\module) = 0) \rightarrow X \cap \module = \fixif{\module}$.
  \end{itemize}
\end{dfn}

Let $\rvertex$ be the root node of the tree decomposition $\TT^q_\pmodule$, we want this definition to achieve $\rsols_{\rvertex, \pmodule} = \dpfam_\pmodule$. Hence, the first two properties are a natural requirement. The third and fourth property lead to the $\pmodule$-promotion property at the root node $\rvertex$ and are more intricate to facilitate the dynamic program. To be precise, since the the bag $\bag^q_{\rvertex}$ at the root node $\rvertex$ is empty, the third property is trivially satisfied and the fourth property turns into the $\pmodule$-promotion property. 

We exclude the current bag from consideration, because we only want to check whether a module $\module$ is isolated in $X$ once all incident edges have been introduced. This is certainly the case when $\module$ leaves the current bag, i.e., it is forgotten. If $\module$ is isolated at this point, we can safely replace the independent set $\fixis{\module}$ inside $\module$ by the induced forest $\fixif{\module}$, which cannot decrease the size of $X$. This means, with the exception of modules inducing a clique, that no module $\module$ in the current bag satisfies $\marking_X(\module) = \twoedge$.

The naive dynamic programming routine would not use promotion and track in which modules of the current bag the solution chooses an induced forest (and not just an independent set). By using promotion, we can save this state and only handle the remaining states, namely choosing no vertex, a single vertex, or an independent set. Thereby, we obtain an improved running time.

Due to \cref{thm:fvs_nonzero_implies_forest}, we want to count for each $X \in \rsols_{t, \pmodule}$ the number of consistent homogeneous cuts. Before considering cuts, each module $\module$ in the considered bag has four possible states. The intersection with $X$ can be empty, contain a single vertex, or contain at least two vertices, and in the latter case we distinguish whether $X$ intersects a neighboring module or not. To count the homogeneous cuts naively, we would split all states except the empty state into two states, one for each side of a cut, thus obtaining seven total states. However, it turns out that tracking the cut side is not necessary when $X$ intersects $\module$ in at least two vertices. When $\module$ is isolated, we can simply count it twice, and otherwise $\module$ inherits the cut side from the unique neighboring module that is also intersected by $X$. Hence, five states suffice and we define the cut solutions accordingly. 

\begin{dfn}
  Let $t \in V(\TT^q_{\pmodule})$ be a node of the tree decomposition $\TT^q_{\pmodule}$. The set of \emph{cut solutions} $\csols_{t, \pmodule}$ consists of pairs $(X, (X_L, X_R))$ such that $X \in \rsols_{t, \pmodule}$ and $(X_L, X_R)$ is $\pmodule$-homogeneous and a consistent cut of $G_t[X \setminus (\iso_t(X) \cap \bag_t)]$, where $\iso_t(X) = \bigcup \{\module \in \children(\pmodule) \sep |X \cap \module| \geq 2, \deg_{G^q_t[\pproj(X)]}(\qvertex) = 0\}$.
\end{dfn}
In the case of isolated modules, we consider it easier to account for the cut side when forgetting the module. Hence, the cuts considered in the definition of $\csols_{t, \pmodule}$ do not cover such modules that belong to the current bag $\bag_t$. Again, for the root node $\rvertex$ of the tree decomposition $\TT^q_{\pmodule}$ this extra property will be trivially satisfied as the associated bag is empty. The definition is again built in such a way that $\csols_{\rvertex, \pmodule} = \csols_\pmodule$.

Our dynamic programming algorithm has to track certain additional data of a solution $X$, namely its size $\ctarget = |X|$, its weight $\wtarget = \wfct(X)$ for the isolation lemma, the number $\vtarget = |\pproj(X)|$ of intersected modules, and the number $\etarget = |E(G^q_t[\pproj(X)])|$ of induced edges in the currently considered subgraph $G^q_t$ of the quotient graph $\pquotient$. We need $\vtarget$ and $\etarget$ to apply \cref{thm:forest_hom_cuts}. Accordingly, we define
$\rsols_{t, \pmodule}^{\targets} = \{X \in \rsols_{t, \pmodule} \sep \ctarget = |X \setminus \bag_t|, \wtarget = \wfct(X \setminus \bag_t), \vtarget = |\pproj(X) \setminus \bag^q_t|, \etarget = |E(G^q_t[\pproj(X)])|\}$ and $\csols_{t, \pmodule}^{\targets} = \{(X, (X_L, X_R)) \in \csols_{t, \pmodule} \sep X \in \rsols_{t, \pmodule}^{\targets}\}$. Note that we exclude the current bag in these counts, except for $\etarget$, hence we have to update these counts when we forget a module. This choice simplifies some recurrences in the algorithm, otherwise updating the counts would be a bit cumbersome due to promotion. 

Finally, we can define the table that is computed at each node $t \in V(\TT^q_{\pmodule})$ by our dynamic programming algorithm. Every module $\module$ in the current bag has one of five states for a given solution $X$, these states are denoted by $\states = \{\zero, \one_L, \one_R, \two_0, \two_1\}$. The bold number refers to the size of the intersection $X \cap \module$, i.e., $\zero$ if $X \cap \module = \emptyset$, $\one$ if $|X \cap \module| = 1$, and $\two$ if $|X \cap \module| \geq 2$. For $\one$, we additionally track whether the module belongs to the left ($\one_L$) or right side ($\one_R$) of the considered homogeneous cut. For $\two$, we additionally track how many neighboring modules are intersected by $X$, due to the definition of $\pmodule$-forest-nice this number is either zero ($\two_0$) or one ($\two_1$). As argued before, we will not have any modules $\module$ with $\marking_X(\module) = \twoedge$ in the current bag unless $\module$ induces a clique. 

We remark that there is an edge case when the graph $G[\module]$ is a clique of size at least 2, as in that case the maximum independent sets of $G[\module]$ are simply singletons which are captured by the states $\one_L$ and $\one_R$. As we do not track the degree of such states, we cannot safely perform promotion for them. Instead we directly introduce induced forests inside $\module$ in this exceptional case with the state $\two_1$.

\begin{dfn}
  Let $t \in V(\TT^q_{\pmodule})$ be a node of the tree decomposition $\TT^q_{\pmodule}$. A function $f \colon \bag^q_t \rightarrow \states$ is called a \emph{$t$-signature}. Let $(X, (X_L, X_R)) \in \csols_{t, \pmodule}$ and $X^q = \pproj(X)$. We say that $(X,(X_L,X_R))$ is \emph{compatible} with a $t$-signature $f$ if the following properties hold for every $v^q_\module \in \bag^q_t$:
  \begin{itemize}
   \item $f(v^q_\module) = \zero$ implies that $\marking_X(\module) = \zero$,
   \item $f(v^q_\module) = \one_L$ implies that $\marking_X(\module) = \one$ and $X \cap \module \subseteq X_L$,
   \item $f(v^q_\module) = \one_R$ implies that $\marking_X(\module) = \one$ and $X \cap \module \subseteq X_R$,
   \item $f(v^q_\module) = \two_0$ implies that $\marking_X(\module) = \twoind$ and $\deg_{G^q_t[X^q]}(v^q_\module) = 0$,
   \item $f(v^q_\module) = \two_1$ and $G[\module]$ is not a clique implies that $\marking_X(\module) = \twoind$ and $\deg_{G^q_t[X^q]}(v^q_\module) = 1$,
   \item $f(v^q_\module) = \two_1$ and $G[\module]$ is a clique implies that $\marking_X(\module) = \twoedge$.
  \end{itemize}
  For a $t$-signature $f$, we let $\dpsols_{t,\pmodule}(f)$ denote the set of all $(X, (X_L, X_R)) \in \csols_{t,\pmodule}$ that are compatible with $f$. Similarly, we define $\dpsols_{t,\pmodule}^{\targets}(f)$ for given $\ctarget \in [0, \cfct(\pmodule)]$, $\wtarget \in [0, \wfct(\pmodule)]$, $\vtarget \in [0, |\pmodule|]$, and $\etarget \in [0, \vtarget - 1]$.
\end{dfn}

Fix a parent module $\pmodule \in \modint(G)$ and for every node $t \in V(\TT^q_{\pmodule})$, $t$-signature $f$, and appropriate $\targets$, define the value $\dppoly_t^{\targets}(f) = |\dpsols_{t, \pmodule}^{\targets}(f)|$. Whenever at least one of $\targets$ is negative, we assume that $\dppoly_t^{\targets}(f) = 0$. We will now describe the dynamic programming recurrences to compute $\dppoly_t^{\targets}(f)$ for all choices of $t$, $f$, $\targets$ based on the type of the node $t$ in the very nice tree decomposition $\TT^q_{\pmodule}$.

\subsubsection*{Leaf bag.} We have that $V^q_t = \bag^q_t = \emptyset$ and $t$ has no child. Therefore, the only candidate is $(\emptyset, (\emptyset, \emptyset))$ and we simply need to check if the trackers $\targets$ agree with that:
\begin{equation*}
  \dppoly_t^{\targets}(f) = [\ctarget = \wtarget = \etarget = \vtarget = 0]
\end{equation*}

\subsubsection*{Introduce vertex bag.} We have that $\bag^q_t = \bag^q_s \cup \{\qvertex\}$, where $\qvertex \notin \bag^q_s$ and $s$ is the only child of $t$. For the sake of the write-up, we assume that $f$ is an $s$-signature here. The recurrence is straightforward with the exception of handling the clique case:
\begin{equation*}
  \dppoly_t^{\targets}(f[\qvertex \mapsto \state]) = 
  \begin{cases}
    \dppoly_s^{\targets}(f), & \text{if } \state \in \{\zero, \one_L, \one_R\}, \\
    [\text{$G[\module]$ is not a clique}] \dppoly_s^{\targets}(f), & \text{if } \state = \two_0, \\
    [|\module| > 1 \text{ and $G[\module]$ is a clique}] \dppoly_s^{\targets}(f), & \text{if } \state = \two_1.
  \end{cases}
\end{equation*}
If $G[\module]$ is a clique, then $\marking_X(\module) = \twoind$ can never be satisfied. So, we will directly generate solutions with $\marking_X(\module) = \twoedge$ in this case. If $G[\module]$ is not a clique, such solutions will only be generated at forget nodes by \emph{promotion}. Recall that no edges incident to $\module$ have been introduced yet, which in particular rules out the case that $f(\qvertex) = \two_1$ when $G[\module]$ is not a clique, and the trackers are only updated when we forget a module.

\subsubsection*{Introduce edge bag.} We have that $\{v^q_{\module_1}, v^q_{\module_2}\} \subseteq \bag^q_t = \bag^q_s$, where $\{v^q_{\module_1}, v^q_{\module_2}\}$ denotes the introduced edge and $s$ is the only child of $t$.
Define helper functions $\newedge, \cons \colon \states \times \states \rightarrow \{0,1\}$ by $\newedge(\state_1, \state_2) = [\state_1 \neq \zero \wedge \state_2 \neq \zero]$ and $\cons$ is given by the following table:
\begin{equation*}
  \begin{array}{l|ccccc}
    \cons  & \zero & \one_L & \one_R & \two_0 & \two_1 \\
    \hline
    \zero  & 1 & 1 & 1 & 1 & 1 \\
    \one_L & 1 & 1 & 0 & 0 & 1 \\
    \one_R & 1 & 0 & 1 & 0 & 1 \\
    \two_0 & 1 & 0 & 0 & 0 & 0 \\
    \two_1 & 1 & 1 & 1 & 0 & 0 
  \end{array}
\end{equation*}
The $\cons$-function is used to filter partial solutions that have incompatible states at the newly introduced edge. There are three reasons why states might be incompatible: they belong to different sides of the cut, they directly induce a cycle, or they do not correctly account for the degree in the graph induced by the partial solution.

Furthermore, given a $t$-signature $f$, we define the $s$-signature $\tilde{f}$ as follows. We set $\tilde{f} := f$ if $\cons(f(v^q_{\module_1}), f(v^q_{\module_2})) = 0$ or $\newedge(f(v^q_{\module_1}), f(v^q_{\module_2})) = 0$ or $\two_1 \notin \{f(v^q_{\module_1}), f(v^q_{\module_2})\}$. Otherwise, the introduced edge changes the state from $\two_0$ to $\two_1$ at one of its endpoints, i.e., without loss of generality $f(v^q_{\module_1}) = \two_1$ and $f(v^q_{\module_2}) \in \{\one_L, \one_R\}$ (else, swap role of $\module_1$ and $\module_2$) and we set $\tilde{f} := f[v^q_{\module_1} \mapsto \two_0]$. Finally, the recurrence is given by 
\begin{equation*}
  \dppoly_t^{\targets}(f) = \cons(f(v^q_{\module_1}), f(v^q_{\module_2})) \dppoly_s^{\ctarget, \wtarget, \vtarget, \etarget - \newedge(f(v^q_{\module_1}), f(v^q_{\module_2}))}(\tilde{f}).
\end{equation*}
Observe that we update the edge count, if necessary, in this recurrence. We remark that if $f(v^q_{\module_1}) = \two_1$ and $f(v^q_{\module_2}) \in \{\one_L, \one_R\}$ and $G[\module_1]$ is a clique, we should filter as well, because this means $\marking_X(\module_1) = \twoedge$ and hence $v^q_{\module_1}$ should not receive incident edges in $G^q_t[\pproj(X)]$. One could explicitly adapt the recurrence for this case or instead, as we do, observe that since $\marking_X(\module_1) = \twoind$ is impossible, all entries $\dppoly_s^{\targets}(\tilde{f})$ will be zero due to $\tilde{f}(v^q_{\module_1}) = \two_0$ and hence we do not generate any partial solutions for this case anyway.

\subsubsection*{Forget vertex bag.} We have that $\bag^q_t = \bag^q_s \setminus \{\qvertex\}$, where $\qvertex \in \bag^q_s$ and $s$ is the only child of $t$. Recall that $\cis{\module}$, $\cif{\module}$, $\wv{\module}$, $\wis{\module}$, $\wif{\module}$ denote the size or weight of a singleton set, maximum independent set, or the candidate forest inside $\module$, respectively. The recurrence is given by:
\begin{equation*}
  \begin{array}{lclcl}
    \dppoly_t^{\targets}(f) & = & & & \dppoly_s^{\targets}(f[v^q_\module \mapsto \zero]) \\
    & + & & & \dppoly_s^{\ctarget - 1, \wtarget - \wv{\module}, \vtarget - 1, \etarget}(f[v^q_\module \mapsto \one_L]) \\
    & + & & & \dppoly_s^{\ctarget - 1, \wtarget - \wv{\module}, \vtarget - 1, \etarget}(f[v^q_\module \mapsto \one_R]) \\
    & + & 2 & \cdot & \dppoly_s^{\ctarget - \cif{\module}, \wtarget - \wif{\module}, \vtarget - 1, \etarget}(f[v^q_\module \mapsto \two_0]) \\
    & + & [\text{$G[\module]$ is not a clique}] & \cdot & \dppoly_s^{\ctarget - \cis{\module}, \wtarget - \wis{\module}, \vtarget - 1, \etarget}(f[v^q_\module \mapsto \two_1]) \\
    & + & 2[\text{$|\module| > 1$ and $G[\module]$ is a clique}] & \cdot & \dppoly_s^{\ctarget - \cif{\module}, \wtarget - \wif{\module}, \vtarget - 1, \etarget}(f[v^q_\module \mapsto \two_1]) 
  \end{array}
\end{equation*}
As $\module$ leaves the current bag, we need to update the trackers $\ctarget$, $\wtarget$, and $\vtarget$. The first three cases are straightforward, but the latter three deserve an explanation. If $\module$ had state $\two_0$ before, then $\module \subseteq \iso_s(X)$ and $G[\module]$ cannot be a clique, so we want to promote the independent set in $\module$ to an induced forest and also track the cut side now. Since $\module$ remains isolated, both cut sides are possible, explaining the factor 2. If $G[\module]$ is not a clique and $\module$ had state $\two_1$ before, then we keep the independent set in $\module$ and its cut side is already tracked. If instead $G[\module]$ is a clique and had state $\two_1$ before, then $M \subseteq \iso_s(X)$ and we are taking an edge (= maximum induced forest) inside $\module$ and we need to track its cut side now. 

\subsubsection*{Join bag.} We have that $\bag^q_t = \bag^q_{s_1} = \bag^q_{s_2} = V^q_{s_1} \cap V^q_{s_2}$, where $s_1$ and $s_2$ are the two children of $t$. To state the recurrence for the join bag, we first introduce the \emph{induced forest join} $\oplus_{\indfor} \colon \states \times \states \rightarrow \states \cup \{\bot\}$, where $\bot$ stands for an undefined value, which is defined by the following table:
\begin{equation*}
  \begin{array}{l|ccccc}
    \oplus_{\indfor}  & \zero & \one_L & \one_R & \two_0 & \two_1 \\
    \hline
    \zero  & \zero & \bot & \bot & \bot & \bot \\
    \one_L & \bot & \one_L & \bot & \bot & \bot \\
    \one_R & \bot & \bot & \one_R & \bot & \bot \\
    \two_0 & \bot & \bot & \bot & \two_0 & \two_1 \\
    \two_1 & \bot & \bot & \bot & \two_1 & \bot 
  \end{array}
\end{equation*}
When combining two partial solutions, one coming from child $s_1$ and the other one coming from $s_2$, we want to ensure that they have essentially the same states on $\bag^q_t = V^q_{s_1} \cap V^q_{s_2}$. However for the state $\two_1$ (if the considered modules does not induce a clique), we need to decide which child contributes the incident edge in the quotient graph and ensure that the other child does not contribute an additional edge. This is implemented by the operation $\oplus_{\indfor}$. Given some set $S$ and functions $f, g \colon S \rightarrow \states$, we abuse notation and let $f \oplus_{\indfor} g \colon S \rightarrow \states \cup \{\bot\}$ denote the function obtained from $f$ and $g$ by pointwise application of $\oplus_{\indfor}$. We also define $\oplus_{\two} = \oplus_{\indfor} \big|_{\{\two_0, \two_1\} \times \{\two_0, \two_1\}}$ and similarly extend it to functions.

\newcommand{\cliquebag}{\widetilde{\bag}}
For any module $\module$ with $\qvertex \in \bag^q_t$ that induces a clique, the state $\two_1$ behaves differently and should agree on both children. Hence, we define $\cliquebag^q_t = \{\qvertex \in \bag^q_t \sep G[\module] \text{ is a clique}\}$. We can now state a first version of the recurrence, which will be transformed further to enable efficient computation. The preliminary recurrence is given by
\begin{equation*}
  \dppoly_t^{\targets}(f) = \sum_{\substack{\ctarget_1 + \ctarget_2 = \ctarget \\ \wtarget_1 + \wtarget_2 = \wtarget}} \sum_{\substack{\vtarget_1 + \vtarget_2 = \vtarget \\ \etarget_1 + \etarget_2 = \etarget}} \smashoperator[r]{\sum_{\substack{f_1, f_2 \colon \bag^q_t \setminus \cliquebag^q_t \rightarrow \states \colon \\ f_1 \oplus_{\indfor} f_2 = f}}} A_{s_1}^{\ctarget_1, \wtarget_1, \vtarget_1, \etarget_1}(f_1 \cup f\restrict{\cliquebag^q_t}) A_{s_2}^{\ctarget_2, \wtarget_2, \vtarget_2, \etarget_2}(f_2 \cup f\restrict{\cliquebag^q_t}),
\end{equation*}
where we ensure that all states agree for modules inducing cliques and otherwise apply the induced forest join $\oplus_{\indfor}$.

\newcommand{\eqdomain}{D^=}
\newcommand{\neqdomain}{D^{\neq}}

To compute this recurrence quickly, we separately handle the part of $\oplus_{\indfor}$ that essentially checks for equality and reduce the remaining part to already known the results. Given a $t$-signature $f \colon \bag^q_t \rightarrow \states$, we define $\eqdomain_t(f) := \cliquebag^q_t \cup f^{-1}(\{\zero, \one_L, \one_R\})$ and $\neqdomain_t(f) := \bag^q_t \setminus \eqdomain_t(f)$. We decompose $f$ into $f^= := f \restrict{\eqdomain_t(f)}$ and $f^{\neq} := f \restrict{\neqdomain_t(f)}$. 

\newcommand{\boldx}{\vec{\mathbf{x}}}

We fix the values $\targets$ and a function $g \colon S \rightarrow \states$ where $\cliquebag^q_t \subseteq S \subseteq \bag^q_t$ is some subset of the current bag containing the clique modules. We claim that the entries $\dppoly_t^{\targets}(f)$ for all $t$-signatures $f$ with $f^= = g$ (including $\eqdomain_t(f) = S$) can be computed in time $\Oh^*(2^{|\bag^q_t \setminus S|})$. We branch on $\boldx_1 = (\ctarget_1, \wtarget_1, \vtarget_1, \etarget_1)$, which determines the values $\boldx_2 = (\ctarget_2, \wtarget_2, \vtarget_2, \etarget_2)$, and define the auxiliary table $T_g^{\boldx_1, \boldx_2}$ indexed by $h\colon \bag^q_t \setminus S \rightarrow \{\two_0, \two_1\}$ as follows
\begin{equation*}
  T_g^{\boldx_1, \boldx_2}(h) = {\sum_{\substack{h_1, h_2 \colon \bag^q_t \setminus S \rightarrow \{\two_0, \two_1\} \colon \\ h_1 \oplus_{\two} h_2 = h}}} \dppoly_{s_1}^{\ctarget_1, \wtarget_1, \vtarget_1, \etarget_1}(g \cup h_1) \dppoly_{s_2}^{\ctarget_2, \wtarget_2, \vtarget_2, \etarget_2}(g \cup h_2).
\end{equation*}
Since $\oplus_{\two}$ is essentially the same as addition over $\{0,1\}$ with $1 + 1$ being undefined, we can compute all entries of $T_g^{\boldx_1, \boldx_2}$ in time $\Oh^*(2^{|\bag^q_t \setminus S|})$ by the work of, e.g., van Rooij~\cite[Theorem 2]{Rooij21} using fast subset convolution and the fast fourier transform. Then, for every $t$-signature $f$ with $f^= = g$, we obtain $\dppoly_t^{\targets}(f)$ by summing $T_g^{\boldx_1, \boldx_2}(f^{\neq})$ over all $\boldx_1 + \boldx_2 = (\ctarget, \wtarget, \vtarget, \etarget)$. Since there are only polynomially many choices for $\boldx_1$ and $\boldx_2$, this proves the claim.

In conclusion, to compute $\dppoly_t^{\targets}(f)$ for all $\targets$, $f$, we need time
\begin{align*}
  \sum_{\cliquebag^q_t \subseteq S \subseteq \bag^q_t} \sum_{g \colon S \rightarrow \{\zero, \one_L, \one_R\}} \Oh^*(2^{|\bag^q_t \setminus S|}) & \leq \sum_{S \subseteq \bag^q_t} \Oh^*(3^{|S|}2^{|\bag^q_t \setminus S|}) = \Oh^*\left(\sum_{i = 0}^{|\bag^q_t|} \binom{|\bag^q_t|}{i} 3^i 2^{|\bag^q_t| - i}\right) \\
  & = \Oh^*((3 + 2)^{|\bag^q_t|}) = \Oh^*(5^k).
\end{align*}

\begin{lem}\label{thm:fvs_mtw_prime_algo}
	Let $\pmodule \in \modint(G)$ be a prime node and $\wfct\colon V \rightarrow [2|V|]$ a weight function. Given a tree decomposition of $\pquotient$ of width $k$ and the sets $\fixv{\module}$, $\fixis{\module}$ and values $\cif{\module}$, $\wif{\module}$ for all $\module \in \children(\pmodule)$, the values $|\dpfamcuts^\targets_\pmodule|$ can be computed in time $\Oh^*(5^k)$ for all $\targets$.
\end{lem}

\begin{proof}
	From the sets $\fixv{\module}$ and $\fixis{\module}$, we directly obtain the values $\wv{\module}$, $\cis{\module}$, $\wis{\module}$ for all $\module \in \children(\pmodule)$. We then transform the given tree decomposition into a very nice tree decomposition $(\TT^q_{\pmodule}, (\bag^q_t)_{t \in V(\TT^q_{\pmodule})})$ using \cref{thm:very_nice_tree_decomposition} and run the described dynamic programming algorithm described before to compute the values $\dppoly_{\rvertex}^{\targets}(\emptyset)$, where $\rvertex$ is the root of $\TT^q_{\pmodule}$, for all appropriate values of $\targets$. Assuming the correctness of the recurrences, we have that $\dppoly_{\rvertex}^{\targets}(\emptyset) = |\dpsols_{\rvertex}^{\targets}(\emptyset)| = |\dpfamcuts_\pmodule^{\targets}|$ by definition and the degeneration of the conditions at $\rvertex$.
	
  For the running time, note that for every $t \in V(\TT^q_{\pmodule})$, there are at most $\Oh^*(5^k)$ table entries $\dppoly_{t}^{\targets}(f)$ and the recurrences can be computed in polynomial time except for the case of join bags. In the case of a join bag, we have shown how to compute all table entries simultaneously in time $\Oh^*(5^k)$. By \cref{thm:very_nice_tree_decomposition} the tree decomposition $\TT^q_{\pmodule}$ has a polynomial number of nodes, hence the running time follows and it remains to sketch the correctness of the dynamic programming recurrences.
	
  For leaf bags, the correctness follows by observing that $\dpsols_t(\emptyset) = \csols_{t, \pmodule} = \{(\emptyset, (\emptyset, \emptyset))\}$. So, we start by considering introduce vertex bags. We set up a bijection between $\dpsols_{t}(f[\qvertex \mapsto \state])$ and $\dpsols_{s}(f)$ depending on $\state \in \states$. We map $(X, (X_L, X_R)) \in \dpsols_{s}(f)$ to 
	\begin{itemize}
  \item $(X, (X_L, X_R))$ if $\state = \zero$,
  \item $(X \cup \fixv{\module}, (X_L \cup \fixv{\module}, X_R))$ if $\state = \one_L$ ($\one_R$ is analogous),
  \item $(X \cup \fixis{\module}, (X_L, X_R))$ if $\state = \two_0$ and $G[\module]$ is not a clique,
  \item $(X \cup \fixif{\module}, (X_L, X_R))$ if $\state = \two_1$ and $G[\module]$ is a clique of size at least 2.
 \end{itemize}
 In the last two cases, we have $\module \subseteq \iso_t(X)$, so we do not need to track the cut side. Using $\pmodule$-substructure it is possible to verify that these mappings constitute bijections. The case that $\state = \two_1$ and $G[\module]$ is not a clique is impossible, since no edges incident to $\qvertex$ are introduced yet. The case that $\state = \two_0$ and $G[\module]$ is a clique is impossible, since any subset of $\module$ of size at least two has to induce an edge.
 
 For introduce edge bags, we highlight the case that $\tilde{f}(v^q_{\module_1}) = \two_0$ and $f(v^q_{\module_1}) = \two_1$, where $\module_1$ needs to inherit the cut side from $\module_2$. Formally, a partial solution $(X, (X_L, X_R)) \in \dppoly_s^{\ctarget, \wtarget, \vtarget, \etarget - 1}(\tilde{f})$ with $f(v^q_{\module_2}) = \one_L$ is bijectively mapped to $(X, (X_L \cup (X \cap \module), X_R)) \in \dppoly_t^{\targets}(f)$ and analogously when $f(v^q_{\module_2}) = \one_R$. We have already argued the correct handling of the clique case when presenting the recurrence. The remaining cases are straightforward.
 
 We proceed with forget vertex bags. First, we observe that all considered cases are disjoint, hence no overcounting occurs. The handling of the cases $\zero$, $\one_L$, and $\one_R$ is standard and we omit further explanation. For isolated modules, we need to track the cut side when we forget them, since both sides are possible, we multiply with the factor 2. Furthermore, we need to perform the promotion when we forget a module with state $\two_0$. The most involved case is $\fixif{\module} \neq \fixis{\module}$ and $G[\module]$ is not a clique, then we perform promotion on the isolated module $\module$, swapping $\fixis{\module}$ with $\fixif{\module}$, and now have to track the cut side of $\module$, again yielding the factor 2. Formally, if $f$ is a $t$-signature and $(X, (X_L, X_R)) \in \dpsols_s(f[\qvertex \mapsto \two_0])$, then $G[\module]$ is not a clique and we obtain the solutions $((X \setminus \module) \cup \fixif{\module}, (X_L \cup \fixif{\module}, X_R)) \in \dpsols_t(f)$ and $((X \setminus \module) \cup \fixif{\module}, (X_L, X_R \cup \fixif{\module})) \in \dpsols_t(f)$.
 
 For the join bags, we have that $V^q_{s_1} \cap V^q_{s_2} = \bag^q_t$, so the behavior on the intersection is completely described by the signature $f$. Every $(X,(X_L,X_R)) \in \dpsols_t(f)$ splits into a solution $(X^1, (X^1_L, X^2_L)) \in \dpsols_{s_1}(f_1)$ at $s_1$ and a solution $(X^2, (X^2_L, X^2_R)) \in \dpsols_{s_2}(f_2)$ at $s_2$, where for $i \in [2]$ we set $X^i = X \cap V_{s_i}$, $X^i_L = (X_L \cap V_{s_i}) \setminus (\iso_{s_i}(X^i) \cap \bag_{s_i})$, $X^i_R = (X_R \cap V_{s_i}) \setminus (\iso_{s_i}(X^i) \cap \bag_{s_i})$ and 
 \begin{equation*}
   f_i(\qvertex) = \begin{cases}
     f(\qvertex), & \text{if } f(\qvertex) \neq \two_1, \\
     \two_1,      & \text{if $f(\qvertex) = \two_1$ and $G[\module]$ is clique of size $\geq 2$}, \\
     \two_d,      & \text{if $f(\qvertex) = \two_1$ and $G[\module]$ is not a clique and $\deg_{G^q_{s_i}[\pproj(X^i)]}(\qvertex) = d$}.
                   \end{cases}
 \end{equation*}
 For a non-clique module with state $\two_1$, the edge leading to degree 1 is present at one of the child nodes $s_1$ or $s_2$, but not at the other one. At the child, where the edge is not present, the module has state $\two_0$ and is isolated, therefore we do not track the cut side and hence have to account for this in the definitions of $X^i_L$ and $X^i_R$. This map can be seen to be a bijection between $\dpsols_t(f)$ and $\bigcup_{f_1, f_2} \dpsols_{s_1}(f_1) \times \dpsols_{s_2}(f_2)$, where the union is over all $f_1, f_2 \colon \bag^q_t \rightarrow \states$ such that $f\restrict{\cliquebag^q_t} = f_1\restrict{\cliquebag^q_t} = f_2\restrict{\cliquebag^q_t}$ and $f\restrict{\bag^q_t \setminus \cliquebag^q_t} = f_1\restrict{\bag^q_t \setminus \cliquebag^q_t} \oplus_{\indfor} f_2\restrict{\bag^q_t \setminus  \cliquebag^q_t}$, which is implemented by the join-recurrence once we account for the trackers $\ctarget$, $\wtarget$, $\vtarget$, and $\etarget$; as every edge is introduced exactly once and the other trackers are only computed for forgotten vertices, no overcounting happens here and we only have to consider how the trackers are distributed between $s_1$ and $s_2$. We also remark that the correctness here requires that the promotion property is only applied to forgotten modules which have received all incident edges already. \qed
\end{proof}

Finally, we have assembled all ingredients to prove the desired theorem.

\begin{thm}
	There exists a Monte-Carlo algorithm that, given a tree decomposition of width $k$ for every prime quotient graph in the modular decomposition of $G$, solves \FVS in time $\Oh^*(5^k)$. The algorithm cannot give false positives and may give false negatives with probability at most $1/2$.
\end{thm}

\begin{proof}
	Solving the complementary problem \IF, we begin by computing the sets $\fixv{\pmodule}$ and $\fixis{\pmodule}$ for all $\pmodule \in \modtree(G)$ in time $\Oh^*(2^k)$ using \cref{thm:is_mtw_algo}. We sample a weight function $\wfct\colon V \rightarrow [2n]$ uniformly at random, which max-isolates $\optfamily(G)$ with probability at least $1/2$ by \cref{thm:max_isolation}. We generate the sets $\fixif{\pmodule}$ for the base cases $\pmodule = \{v\}$, $v \in V$. 
	
	By bottom-up dynamic programming along the modular decomposition, we inductively compute the values $\cif{\pmodule}$ and $\wif{\pmodule}$, $\pmodule \in \modint(G)$, given the values $\cif{\module}$ and $\wif{\module}$ for all $\module \in \children(\pmodule)$. To do so, we distinguish whether $\pmodule$ is a parallel, series, or prime node. In the first two cases, we can compute these values in polynomial time by \cref{sec:fvs_parallel_series}. 
	
	In the prime case, we compute the values $|\dpfamcuts^\targets_\pmodule|$ in time $\Oh^*(5^k)$ using \cref{thm:fvs_mtw_prime_algo}. From these values, we can obtain the values $\cif{\pmodule}$ and $\wif{\pmodule}$ by the description in \cref{sec:fvs_prime_candidate} in polynomial time. As the modular decomposition has a polynomial number of nodes, the running time follows.
	
	If $\cif{V} \geq \budget$, then the algorithm returns true and otherwise the algorithm returns false. It remains to prove the correctness of this step, assuming that the weight function $\wfct$ is isolating. By \cref{thm:fvs_mtw_main_correctness}, we have that $\fixif{V}$ is a maximum induced forest of $G[V] = G$ if $\wfct$ is isolating and since $\cif{V}
	 = |\fixif{V}|$ this shows that the algorithm is correct in this case. Since we always ensure that $\fixif{V}$ is an induced forest, but not necessarily maximum, even if $\wfct$ is not isolating, the algorithm cannot return false positives. \qed
\end{proof}

\section{Lower Bounds}\label{sec:modtw_lb}

In this section, we prove the tight lower bounds for \CVC and \FVS parameterized by twinclass-pathwidth, cf.~\cref{thm:intro_lower_bounds}. The construction principle follows the style of Lokshtanov et al.~\cite{LokshtanovMS18}. On a high level, that means the resulting graphs can be interpreted as a \emph{matrix of blocks}, where each block spans several rows and columns. Every row is a long \emph{path-like gadget} that simulates a constant number of variables of the \SAT instance and which contributes 1 unit of twinclass-pathwidth. The number of simulated variables is tied to the running time we want to rule out. For technical reasons, we consider bundles of rows simulating a \emph{variable group} of appropriate size. Every column corresponds to a clause and consists of gadgets that \emph{decode} the states on the path gadgets and check whether the resulting assignment satisfies the clause.

In both lower bounds, the main technical contribution is the design of the \emph{path gadgets}. Whereas the design of the \emph{decoding gadgets} can be adapted from known constructions. The main challenge in the construction of the path gadgets is that the appearance of twinclasses \emph{restricts} the design space: we \emph{cannot} attach separate gadgets to each vertex in the twinclass, but only gadgets to read the state of the twinclass as a \emph{whole}. To interface with the decoding gadgets, each path gadget contains a \emph{clique-like} center containing one vertex per desired state of the path gadget. An additional complication is the \emph{transitioning} of the state throughout a long path, where the presence of twinclasses means that we have \emph{less control} over the transitioning compared to the sparse case, e.g., when simply parameterizing by pathwidth.

\subsection{Connected Vertex Cover}
\label{sec:modtw_cvc_lb}

This subsection is devoted to proving that \CVC parameterized by twinclass-pathwidth cannot be solved in time $\Oh^*((5-\eps)^{\tcpw(G)})$ for some $\eps > 0$ unless the \SETH fails. We first design the path gadget and analyze it in isolation and afterwards we present the complete construction. The decoding gadgets are directly adapted from the lower bound for \CVC parameterized by pathwidth given by Cygan et al.~\cite{CyganNPPRW11arxiv}.

\newpage
\subsubsection{Path Gadget Construction and Analysis}

\paragraph{Root.} We create a vertex $\rvertex$ called the \emph{root} and attach a vertex $\rvertex'$ of degree 1 to ensure that every connected vertex cover contains $\rvertex$. Given a subset $X \subseteq V(G)$ with $\rvertex \in X$, a vertex $v \in X$ is \emph{root-connected} in $X$ if there is a $v,\rvertex$-path in $G[X]$. We just say \emph{root-connected} if $X$ is clear from the context. Note that $G[X]$ is connected if and only if all vertices of $X$ are root-connected in $X$. 

\paragraph{States.} We define the three atomic states $\atoms = \{\zero, \one_0, \one_1\}$ and define the two predicates $\sol, \conn \colon \atoms \rightarrow \{0,1\}$ by $\sol(\bolda) = [\bolda \in \{\one_0, \one_1\}]$ and $\conn(\bolda) = [\bolda = \one_1]$. The atom $\zero$ means that a vertex is not inside the partial solution; $\one_1$ and $\one_0$ indicate that a vertex is inside the partial solution and the subscript indicates whether it is root-connected or not. Building on these atomic states, we define five states consisting of four atomic states each: 
\begin{itemize}
  \item $\state^1  = (\zero_{\phantom{0}}, \zero_{\phantom{0}}, \one_1, \one_1)$, 
 \item $\state^2  = (\one_0, \zero_{\phantom{0}}, \one_1, \one_0)$, 
 \item $\state^3  = (\one_1, \zero_{\phantom{0}}, \one_0, \one_0)$, 
 \item $\state^4  = (\one_0, \one_0, \one_1, \zero_{\phantom{0}})$, 
 \item $\state^5  = (\one_1, \one_1, \one_0, \zero_{\phantom{0}})$.
\end{itemize}
Why the states are numbered in this way will become clear later. 
We collect the five states in the set $\states = \{\state^1, \ldots, \state^5\}$ and use the notation $\state^\ell_i \in \atoms$, $i \in [4]$, $\ell \in [5]$, to refer to the $i$-th coordinate of state $\state^\ell$.

\paragraph{Path gadget.} The path gadget $P$ is constructed as follows. We create 15 \emph{central} vertices $w_{\ell, i}$, $\ell \in [5]$, $i \in [3]$, in 5 sets $W_\ell = \{w_{\ell, 1}, w_{\ell, 2}, w_{\ell, 3}\}$ of size 3 and each set will form a twinclass. We create 2 \emph{input} vertices $u_1, u_2$, 4 \emph{cost} vertices $w_{+,1}, \ldots, w_{+,4}$, 5 \emph{clique} vertices $v_1, \ldots, v_5$, and 5 \emph{complement} vertices $\bar{v}_1, \ldots, \bar{v}_5$. Furthermore, for every $f \in [4]$, we create 2 \emph{auxiliary} vertices $a_{1,f}, a_{2,f}$, 2 \emph{indicator} vertices $b_{0,f}, b_{1,f}$, and 2 \emph{connectivity} vertices $c_{0,f}, c_{1,f}$. Finally, we create 4 further auxiliary vertices $\bar{a}_{1,1}, \bar{a}_{2,1}, \bar{a}_{1,2}, \bar{a}_{2,2}$ and 4 further connectivity vertices $\bar{c}_{0,1}, \bar{c}_{1,1}, \bar{c}_{0,2}, \bar{c}_{1,2}$. The vertices $a_{1,4}$ and $\bar{a}_{1,2}$ will also be called \emph{output} vertices.

We add edges such that the central sets $W_\ell$, $\ell \in [5]$, are pairwise adjacent twinclasses, i.e. they induce a complete 5-partite graph, and such that the clique vertices $v_\ell$, $\ell \in [5]$, form a clique. Each complement vertex $\bar{v}_\ell$, $\ell \in [5]$, is made adjacent to $W_\ell$ and to $v_\ell$. The cost vertices $w_{+,1}$ and $w_{+,2}$ are made adjacent to $W_1$; $w_{+,3}$ is made adjacent to $W_2$; and $w_{+,4}$ is made adjacent to $W_3$. 

For every $f \in [4]$, we add edges $\{a_{1,f}, a_{2,f}\}$, $\{a_{2,f}, b_{1,f}\}$, $\{b_{1,f}, b_{0,f}\}$, $\{b_{0,f}, a_{1,f}\}$, forming a $C_4$, and the edges $\{a_{1,f}, c_{1,f}\}$ and $\{c_{0,f}, c_{1,f}\}$. For every $i \in [2]$, we add edges $\{\bar{a}_{1,i}, \bar{a}_{2,i}\}$, $\{\bar{a}_{1,i}, \bar{c}_{1,i}\}$, $\{\bar{c}_{0,i}, \bar{c}_{1,i}\}$. The input vertices $u_1$ and $u_2$ are made adjacent to each $a_{1,f}$ for $f \neq 4$ and they are made adjacent to $\bar{a}_{1,1}$.

All vertices except $\{a_{1,f} \sep f \in [4]\} \cup \{\bar{a}_{1,i}, \bar{a}_{2,i} \sep i \in [2]\} \cup \{u_1, u_2\}$ are made adjacent to the root $\rvertex$. Finally, we describe how to connect the central vertices to the rest. Each twinclass $W_\ell$, $\ell \in [5]$, is made adjacent to $b_{[\state^\ell_2 = \zero], f}$ and to $c_{[\state^\ell_1 = \state^\ell_2],f}$ for all $f \in [4]$ and $W_\ell$ is also made adjacent to $\bar{c}_{[\state^\ell_1 \neq \one_0], 1}$ and $\bar{c}_{[\state^\ell_1 \neq \one_1], 2}$. The construction is depicted in \cref{fig:cvc_modtw_path_1} and \cref{fig:cvc_modtw_path_2}.
\begin{figure}[h]
  \centering
  \scalebox{0.9}{\input{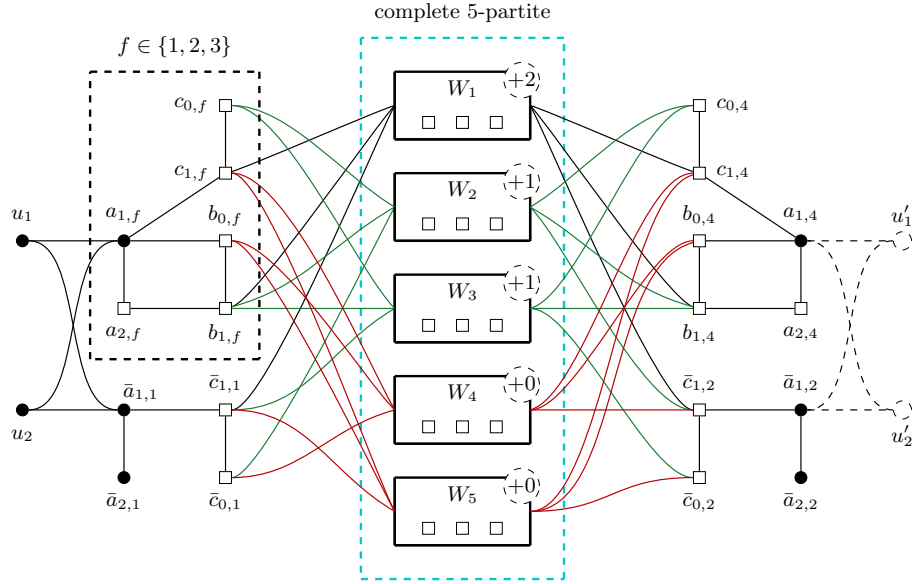}}
  \caption{Vertices depicted with a rectangle are adjacent to the root vertex $\rvertex$. The graph in the black dashed rectangle appears thrice with the same connections to the remaining vertices. The vertices inside the cyan dashed rectangle induce a complete 5-partite graph. The dashed circles at the central vertices indicate the number of cost vertices attached to this set and the dashed vertices and edges at the right indicate how to connect to the next copy of the path gadget.}
  \label{fig:cvc_modtw_path_1}
  \vspace*{-0.3cm}
\end{figure}

We emphasize that the graphs $P[\{a_{1,f}, a_{2,f}, b_{0,f}, b_{1,f}, c_{0,f}, c_{1,f}\} \cup \bigcup_{\ell \in [5]} W_\ell]$, $f \in [4]$, are all isomorphic to each other, however the first three are also adjacent to the input vertices $u_1$ and $u_2$, whereas the fourth one is not. To study the path gadget $P$, we mostly consider the parts in \cref{fig:cvc_modtw_path_1}; the parts in \cref{fig:cvc_modtw_path_2} are considerably simpler and will later allow us to simply attach the standard decoding gadget already used by Cygan et al.~\cite{CyganNPPRW11arxiv} for \CVC parameterized by pathwidth.

\begin{figure}[h]
  \centering
  \scalebox{0.9}{\begin{tikzpicture}
	\begin{pgfonlayer}{nodelayer}
		\node [style=tiny rectangle] (12) at (-1, 5.5) {};
		\node [style=tiny rectangle] (13) at (0, 5.5) {};
		\node [style=tiny rectangle] (14) at (1, 5.5) {};
		\node [style=none] (15) at (-2, 7) {};
		\node [style=none] (16) at (-2, 6) {};
		\node [style=none] (17) at (-2, 5) {};
		\node [style=none] (18) at (2, 5) {};
		\node [style=none] (19) at (2, 6) {};
		\node [style=none] (20) at (0, 6.5) {$W_1$};
		\node [style=tiny rectangle] (21) at (-1, 2.5) {};
		\node [style=tiny rectangle] (22) at (0, 2.5) {};
		\node [style=tiny rectangle] (23) at (1, 2.5) {};
		\node [style=none] (24) at (-2, 2) {};
		\node [style=none] (25) at (2, 2) {};
		\node [style=tiny rectangle] (26) at (-1, -0.5) {};
		\node [style=tiny rectangle] (27) at (0, -0.5) {};
		\node [style=tiny rectangle] (28) at (1, -0.5) {};
		\node [style=none] (29) at (-2, -1) {};
		\node [style=none] (30) at (2, -1) {};
		\node [style=tiny rectangle] (31) at (-1, -3.5) {};
		\node [style=tiny rectangle] (32) at (0, -3.5) {};
		\node [style=tiny rectangle] (33) at (1, -3.5) {};
		\node [style=none] (34) at (-2, -4) {};
		\node [style=none] (35) at (2, -4) {};
		\node [style=tiny rectangle] (36) at (-1, -6.5) {};
		\node [style=tiny rectangle] (37) at (0, -6.5) {};
		\node [style=tiny rectangle] (38) at (1, -6.5) {};
		\node [style=none] (39) at (-2, -7) {};
		\node [style=none] (40) at (2, -7) {};
		\node [style=none] (41) at (-3, 8) {};
		\node [style=none] (42) at (3, 8) {};
		\node [style=none] (43) at (3, -8) {};
		\node [style=none] (44) at (-3, -8) {};
		\node [style=none] (78) at (2, 7) {};
		\node [style=none] (79) at (-2, 4) {};
		\node [style=none] (80) at (-2, 3) {};
		\node [style=none] (81) at (2, 3) {};
		\node [style=none] (82) at (0, 3.5) {$W_2$};
		\node [style=none] (84) at (2, 4) {};
		\node [style=none] (85) at (-2, 1) {};
		\node [style=none] (86) at (-2, 0) {};
		\node [style=none] (87) at (2, 0) {};
		\node [style=none] (88) at (0, 0.5) {$W_3$};
		\node [style=none] (90) at (2, 1) {};
		\node [style=none] (91) at (-2, -2) {};
		\node [style=none] (92) at (-2, -3) {};
		\node [style=none] (93) at (2, -3) {};
		\node [style=none] (94) at (0, -2.5) {$W_4$};
		\node [style=none] (96) at (2, -2) {};
		\node [style=none] (97) at (-2, -5) {};
		\node [style=none] (98) at (-2, -6) {};
		\node [style=none] (99) at (2, -6) {};
		\node [style=none] (100) at (0, -5.5) {$W_5$};
		\node [style=none] (102) at (2, -5) {};
		\node [style=tiny rectangle] (103) at (-4, 7) {};
		\node [style=tiny rectangle] (104) at (-4, 5) {};
		\node [style=tiny rectangle] (105) at (-4, 3) {};
		\node [style=tiny rectangle] (106) at (-4, 0) {};
		\node [style=tiny rectangle] (107) at (4.5, 6) {};
		\node [style=tiny rectangle] (108) at (4.5, 3) {};
		\node [style=tiny rectangle] (109) at (4.5, 0) {};
		\node [style=tiny rectangle] (110) at (4.5, -3) {};
		\node [style=tiny rectangle] (111) at (4.5, -6) {};
		\node [style=tiny rectangle] (112) at (7, 6) {};
		\node [style=tiny rectangle] (113) at (7, 3) {};
		\node [style=tiny rectangle] (114) at (7, 0) {};
		\node [style=tiny rectangle] (115) at (7, -3) {};
		\node [style=tiny rectangle] (116) at (7, -6) {};
		\node [style=none] (117) at (6, 8) {};
		\node [style=none] (118) at (8, 8) {};
		\node [style=none] (119) at (6, -8) {};
		\node [style=none] (120) at (8, -8) {};
		\node [style=none] (121) at (4.5, 6.75) {$\bar{v}_1$};
		\node [style=none] (122) at (4.5, 3.75) {$\bar{v}_2$};
		\node [style=none] (123) at (4.5, 0.75) {$\bar{v}_3$};
		\node [style=none] (124) at (4.5, -2.25) {$\bar{v}_4$};
		\node [style=none] (125) at (4.5, -5.25) {$\bar{v}_5$};
		\node [style=none] (126) at (7, 6.75) {$v_1$};
		\node [style=none] (127) at (7, 3.75) {$v_2$};
		\node [style=none] (128) at (7, 0.75) {$v_3$};
		\node [style=none] (129) at (7, -2.25) {$v_4$};
		\node [style=none] (130) at (7, -5.25) {$v_5$};
		\node [style=none] (131) at (-4, 7.75) {$w_{+,1}$};
		\node [style=none] (132) at (-4, 5.75) {$w_{+,2}$};
		\node [style=none] (133) at (-4, 3.75) {$w_{+,3}$};
		\node [style=none] (134) at (-4, 0.75) {$w_{+,4}$};
		\node [style=none] (135) at (7, 8.75) {clique};
		\node [style=none] (136) at (0, 8.75) {complete 5-partite};
		\node [style=none] (137) at (9, 7) {};
		\node [style=none] (138) at (8.75, 6.25) {$\vdots$};
		\node [style=none] (139) at (9, 5) {};
		\node [style=none] (140) at (9, 4) {};
		\node [style=none] (142) at (9, 2) {};
		\node [style=none] (146) at (9, -7) {};
		\node [style=none] (147) at (9, -5) {};
		\node [style=none] (148) at (8.75, 3.25) {$\vdots$};
		\node [style=none] (149) at (8.75, 0.25) {$\vdots$};
		\node [style=none] (150) at (8.75, -2.75) {$\vdots$};
		\node [style=none] (151) at (8.75, -5.75) {$\vdots$};
		\node [style=none] (152) at (9, 1) {};
		\node [style=none] (153) at (9, -1) {};
		\node [style=none] (154) at (9, -2) {};
		\node [style=none] (155) at (9, -4) {};
	\end{pgfonlayer}
	\begin{pgfonlayer}{edgelayer}
		\draw [style=thick] (15.center) to (16.center);
		\draw [style=thick] (16.center) to (17.center);
		\draw [style=thick] (17.center) to (18.center);
		\draw [style=thick] (18.center) to (19.center);
		\draw [style=thick] (24.center) to (25.center);
		\draw [style=thick] (29.center) to (30.center);
		\draw [style=thick] (34.center) to (35.center);
		\draw [style=thick] (39.center) to (40.center);
		\draw [style=dashed cyan thick] (41.center) to (44.center);
		\draw [style=dashed cyan thick] (44.center) to (43.center);
		\draw [style=dashed cyan thick] (43.center) to (42.center);
		\draw [style=thick] (19.center) to (18.center);
		\draw [style=thick] (18.center) to (17.center);
		\draw [style=thick] (17.center) to (16.center);
		\draw [style=thick] (16.center) to (15.center);
		\draw [style=thick] (25.center) to (24.center);
		\draw [style=thick] (30.center) to (29.center);
		\draw [style=thick] (35.center) to (34.center);
		\draw [style=thick] (40.center) to (39.center);
		\draw [style=thick] (78.center) to (19.center);
		\draw [style=thick] (15.center) to (78.center);
		\draw [style=thick] (79.center) to (80.center);
		\draw [style=thick] (80.center) to (79.center);
		\draw [style=thick] (84.center) to (81.center);
		\draw [style=thick] (79.center) to (84.center);
		\draw [style=thick] (80.center) to (24.center);
		\draw [style=thick] (25.center) to (81.center);
		\draw [style=thick] (81.center) to (25.center);
		\draw [style=thick] (24.center) to (80.center);
		\draw [style=thick] (85.center) to (86.center);
		\draw [style=thick] (86.center) to (85.center);
		\draw [style=thick] (90.center) to (87.center);
		\draw [style=thick] (85.center) to (90.center);
		\draw [style=thick] (86.center) to (29.center);
		\draw [style=thick] (30.center) to (87.center);
		\draw [style=thick] (87.center) to (30.center);
		\draw [style=thick] (29.center) to (86.center);
		\draw [style=thick] (91.center) to (92.center);
		\draw [style=thick] (92.center) to (91.center);
		\draw [style=thick] (96.center) to (93.center);
		\draw [style=thick] (91.center) to (96.center);
		\draw [style=thick] (92.center) to (34.center);
		\draw [style=thick] (35.center) to (93.center);
		\draw [style=thick] (93.center) to (35.center);
		\draw [style=thick] (34.center) to (92.center);
		\draw [style=thick] (97.center) to (98.center);
		\draw [style=thick] (98.center) to (97.center);
		\draw [style=thick] (102.center) to (99.center);
		\draw [style=thick] (97.center) to (102.center);
		\draw [style=thick] (98.center) to (39.center);
		\draw [style=thick] (40.center) to (99.center);
		\draw [style=thick] (99.center) to (40.center);
		\draw [style=thick] (39.center) to (98.center);
		\draw [style=thin, in=180, out=-15] (103) to (16.center);
		\draw [style=thin, in=-180, out=15] (104) to (16.center);
		\draw [style=thin] (105) to (80.center);
		\draw [style=thin] (106) to (86.center);
		\draw [style=dashed cyan thick] (120.center) to (118.center);
		\draw [style=dashed cyan thick] (118.center) to (117.center);
		\draw [style=dashed cyan thick] (117.center) to (119.center);
		\draw [style=dashed cyan thick] (119.center) to (120.center);
		\draw [style=dashed cyan thick] (41.center) to (42.center);
		\draw [style=thin] (19.center) to (107);
		\draw [style=thin] (107) to (112);
		\draw [style=thin] (81.center) to (108);
		\draw [style=thin] (108) to (113);
		\draw [style=thin] (87.center) to (109);
		\draw [style=thin] (109) to (114);
		\draw [style=thin] (93.center) to (110);
		\draw [style=thin] (110) to (115);
		\draw [style=thin] (99.center) to (111);
		\draw [style=thin] (111) to (116);
		\draw [style=thin dashed] (112) to (137.center);
		\draw [style=thin dashed] (112) to (139.center);
		\draw [style=thin dashed] (113) to (140.center);
		\draw [style=thin dashed] (113) to (142.center);
		\draw [style=thin dashed] (116) to (147.center);
		\draw [style=thin dashed] (116) to (146.center);
		\draw [style=thin dashed] (114) to (152.center);
		\draw [style=thin dashed] (114) to (153.center);
		\draw [style=thin dashed] (115) to (154.center);
		\draw [style=thin dashed] (115) to (155.center);
	\end{pgfonlayer}
\end{tikzpicture}}
  \caption{The remaining parts of the path gadget $P$ which will be connected to the decoding gadget. All vertices that are depicted with a rectangle are adjacent to the root vertex $\rvertex$. The vertices inside the cyan dashed rectangle induce a complete 5-partite graph or a clique respectively. Only the clique vertices have neighbors outside of $P$.} 
  \label{fig:cvc_modtw_path_2}
  \vspace*{-0.5cm}
\end{figure}
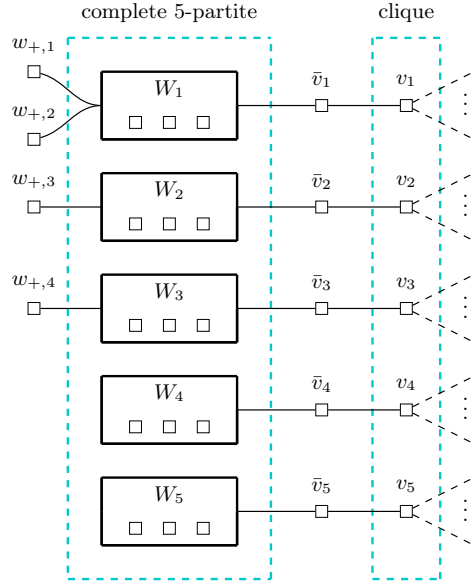

For the upcoming lemmas, we assume that $G$ is a graph that contains $P + \rvertex$ as an induced subgraph and that only the input vertices $u_1, u_2$, the output vertices $a_{1,4}, \bar{a}_{1,2}$, and the clique vertices $v_\ell, \ell \in [5]$, have neighbors outside this copy of $P + \rvertex$. Furthermore, we assume that $\{u_1, u_2\}$ is a twinclass in $G$. Let $X$ be a vertex cover of $G$ with $\rvertex \in X$. We study the behavior of such vertex covers on $P$; we will abuse notation and write $X \cap P$ instead of $X \cap V(P)$.

Observe that the set 
\begin{align*}
  M = \, & \{\{a_{1,f}, a_{2,f}\}, \{b_{0,f}, b_{1,f}\}, \{c_{0,f}, c_{1,f}\} \sep f \in [4]\} \\
  \cup\, & \{\{\bar{a}_{1,i}, \bar{a}_{2,i}\}, \{\bar{c}_{0,i}, \bar{c}_{1,i}\} \sep i \in [2]\} \\
  \cup\, & \{\{v_\ell, \bar{v}_\ell\} \sep \ell \in [5]\}
\end{align*}
is a matching in $P$ of size $4 \cdot 3 + 2 \cdot 2 + 5 = 21$.

\begin{lem}\label{thm:cvc_modtw_path_gadget_lb}
  We have that $|\{\ell \in [5] \sep W_\ell \subseteq X\}| \geq 4$ and $|X \cap P| \geq |M| + 4 \cdot 3 = 33$.
  If $G[X]$ is connected, then $|X \cap P| \geq |M| + 4 \cdot 3 + 2 = 35$ and in case of equality, $|X \cap \{u_1, u_2, w_{+,1}, \ldots, w_{+,4}\}| = 2$ and there is a unique $\ell \in [5]$ such that $W_\ell \not\subseteq X$.
\end{lem}

\begin{proof}
  The vertex set $\bigcup_{\ell \in [5]} W_\ell$ induces a complete $5$-partite graph disjoint from the matching $M$. Any vertex cover must contain at least 4 of the 5 partition classes completely, otherwise there is an edge that is not covered, and since each class is of size 3, this accounts for $4 \cdot 3 = 12$ further vertices. This shows that $|X \cap P| \geq |M| + 4 \cdot 3 = 33$. 

If $X$ completely contains all $W_\ell$, $\ell \in [5]$, then it immediately follows that $|X \cap P| \geq 36$, so if $|X \cap P| = 35$, then there is an unique $\ell \in [5]$ such that $W_\ell \not\subseteq X$. If $\ell = 1$, then we must have $w_{+,1}, w_{+,2} \in X$, so $|X \cap P| \geq 35$. Before we proceed with the remaining proof, notice that $A_f = \{a_{1,f}, a_{2,f}, b_{0,f}, b_{1,f}\}$ induces a $C_4$ for all $f \in [4]$, so if $|X \cap A_f| = 2$, then $X \cap A_f \in \{\{a_{1,f}, b_{1,f}\}, \{a_{2,f}, b_{0,f}\}\}$, i.e., $X$ must pick an antipodal pair from $A_f$.


For the remainder of the proof, assume that $G[X]$ is connected. Suppose that $X \cap \{u_1, u_2\} = \emptyset$, then $a_{1,f} \in X$ for all $f \in [3]$ and $a_{1,f}$ must be root-connected in $X$. If $\ell \in \{2,3\}$, then $b_{1,f}, c_{0,f} \in X$, so whichever neighbor of $a_{1,f}$ we choose for the sake of root-connectedness, the size of $X$ increases by one for every $f \in [3]$. If $\ell \in \{4,5\}$, then $b_{0,f} \in X$, so $a_{1,f}$ is root-connected, but we need to pick another vertex of $A_f$ to cover the remaining edge induced by $A_f$, again increasing the size of $X$. In summary, we obtain $|X \cap P| \geq 36$ if $\ell > 1$ and $X \cap \{u_1, u_2\} = \emptyset$.


Suppose that $|X \cap \{u_1, u_2\}| = 1$ and without loss of generality $u_1 \in X$ and $u_2 \notin X$. Again, we must have $a_{1,f} \in X$ for all $f \in [3]$. If $\ell \in \{2, 3\}$, we have that $w_{+,3} \in X$ or $w_{+,4} \in X$. If $\ell \in \{4, 5\}$, we again see that $|X \cap A_f| \geq 3$ for all $f \in [3]$ and hence $|X \cap P| \geq 37$, so $|X \cap P| \geq 35$ in either case. 

By the previous arguments, we see that $|X \cap P| = 35$ and $X \cap \{u_1, u_2\} = \emptyset$ implies that $\ell = 1$; $|X \cap P| = 35$ and $|X \cap \{u_1, u_2\}| = 1$ implies that $\ell \in \{2,3\}$; $|X \cap P| = 35$ and $|X \cap \{u_1, u_2\}| = 2$ implies that $\ell \in \{4,5\}$. So, the equation $|X \cap \{u_1, u_2, w_{+,1}, \ldots, w_{+,4}\}| = 2$ follows. \qed
\end{proof}

 We want to study the connected vertex covers on $P$ locally, but connectivity is not a local property. However, through our assumption, we know that any vertex in $G[X]$ that is not root-connected in $X \cap (P + \rvertex)$ has to be root-connected through the input or output vertices. In particular, although the clique vertices $v_\ell, \ell \in [5]$, may be adjacent to vertices outside of $P + \rvertex$, any path leaving $P + \rvertex$ through some clique vertex immediately yields a path to $\rvertex$ in $P + \rvertex$, since the clique vertices are adjacent to $\rvertex$. This motivates that we should distinguish whether a vertex in $P + \rvertex$ is root-connected already in $P + \rvertex$ or via a path that leaves $P$.

 Let $Y \subseteq V(G)$, we define $\statemap_Y \colon V(G) \rightarrow \atoms$ by
 \begin{equation*}
   \statemap_Y(v) = \begin{cases}
                      \zero & \text{if } v \notin Y, \\
                      \one_0 & \text{if } v \in Y \text{ and $v$ is not root-connected in $Y \cup \{\rvertex\}$}, \\
                      \one_1 & \text{if } v \in Y \text{ and $v$ is root-connected in $Y \cup \{\rvertex\}$}. \\
                    \end{cases}
 \end{equation*}
 For $Y \subseteq V(P)$, we define $\statemap(Y) = (\statemap_Y(u_1), \statemap_Y(u_2), \statemap_Y(\bar{a}_{1,2}), \statemap_Y(a_{1,4}))$.
 
 We say that a vertex subset $Y \subseteq V(G)$ is \emph{canonical} with respect to the twinclass $\{u_1, u_2\}$ if $u_2 \in Y$ implies $u_1 \in Y$; we will just say that $Y$ is canonical if $\{u_1, u_2\}$ is clear from the context. Since $\{u_1, u_2\}$ is a twinclass, we can always assume that we are working with a canonical subset.
 
 \begin{lem}\label{thm:cvc_modtw_path_gadget_tight}
   If $X$ is canonical, $G[X]$ is connected, and $|X \cap P| \leq 35$, then $|X \cap P| = 35$ and there is an unique $\ell \in [5]$ such that $v_\ell \notin X$ and we have that $\statemap(X \cap P) = \state^\ell$.   
 \end{lem}

\begin{proof}
  \cref{thm:cvc_modtw_path_gadget_lb} implies that $|X \cap P| = 35$, $|X \cap \{u_1, u_2, w_{+,1}, \ldots, w_{+,4}\}| = 2$, that $X$ contains exactly one endpoint of each edge in $M$ and that there is an unique $\ell \in [5]$ such that $W_\ell \not\subseteq X$. To cover all edges between $W_\ell$ and $\bar{v}_\ell$, we must have that $\bar{v}_\ell \in X$ and $v_\ell \notin X$, since $\{\bar{v}_\ell, v_\ell\} \in M$. Furthermore, we must have $X \cap \{v_1, \ldots, v_5\} = \{v_1, \ldots, v_5\} \setminus \{v_\ell\}$, because otherwise $X$ does not cover the clique induced by $v_1, \ldots, v_5$. Hence, the uniqueness of $v_\ell$ follows.
  
  Recall that $A_f = \{a_{1,f}, a_{2,f}, b_{0,f}, b_{1,f}\}$ induces a $C_4$ and $|X \cap A_f| = 2$ because $A_f$ contains two edges of $M$, hence we have that $X \cap A_f \in \{\{a_{1,f}, b_{1,f}\}, \{a_{2,f}, b_{0,f}\}\}$ for all $f \in [4]$.
  
  We claim that $\statemap_{(X \cap P) \setminus \{u_1, u_2\}}(a_{1,f}) = \state^\ell_4$ for all $f \in [4]$. Observe that $\state^\ell_2 = \zero \Leftrightarrow \state^\ell_4 \neq \zero$ and $\state^\ell_1 = \state^\ell_2 \Leftrightarrow \state^\ell_4 \neq \one_0$. Hence, by construction $W_\ell$ is adjacent to $b_{[s^\ell_4 \neq \zero],f}$ and $c_{[\state^\ell_4 \neq \one_0], f}$, so $b_{[s^\ell_4 \neq \zero],f}, c_{[\state^\ell_4 \neq \one_0], f} \in X$ to cover the edges incident to $W_\ell$. So, we see that $a_{1,f} \in X \Leftrightarrow b_{1,f} \in X \Leftrightarrow \state^\ell_4 \neq \zero$ as desired. Concerning the root-connectivity of $a_{1,f}$ in $(X \cap P) \setminus \{u_1, u_2\}$, we know that the adjacent vertices $a_{2,f}$ and $b_{0,f}$ are not in $X$ when $a_{1,f}$ is in $X$, due to $A_f$ inducing a $C_4$, hence $a_{1,f}$ can only be root-connected via $c_{1,f}$. Finally, we see that $c_{1,f} \in X \Leftrightarrow \state^\ell_4 \neq \one_0$. This proves the claim.
  
  The claim implies that $\statemap_{X \cap P}(a_{1,4}) = \state^\ell_4$ as desired. We proceed by computing $\statemap_{(X \cap P) \setminus \{u_1, u_2\}}(\bar{a}_{1,i})$ for $i \in {1,2}$. Due to the degree-1-neighbor $\bar{a}_{2,i}$, we see that $\bar{a}_{1,i} \in X$ because $X$ is a connected vertex cover. The vertex $\bar{a}_{1,i}$ can only be root-connected via $\bar{c}_{1,i}$ and because $\bar{c}_{1,i}$ is an endpoint of a matching edge, we see that $\bar{c}_{1,i} \in X$ if and only if $\bar{c}_{1,i}$ is adjacent to $W_\ell$. For $i = 1$, we have that 
  \begin{equation*}
    \statemap_{(X \cap P) \setminus \{u_1, u_2\}}(\bar{a}_{1,1}) = \one_1 \Leftrightarrow \bar{c}_{1,1} \in X \Leftrightarrow \state^\ell_1 \neq \one_0 \Leftrightarrow \ell \in \{1,3,5\}.
  \end{equation*}
  For $i = 2$, we have that 
  \begin{equation*}
    \statemap_{(X \cap P) \setminus \{u_1, u_2\}}(\bar{a}_{1,2}) = \one_1 \Leftrightarrow \bar{c}_{1,2} \in X \Leftrightarrow \state^\ell_1 \neq \one_1 \Leftrightarrow \state^\ell_3 = \one_1.
  \end{equation*}
  In particular, we have shown that $\statemap_{X \cap P}(\bar{a}_{1,2}) = \state^\ell_3$ as desired.
  
  It remains to show that $\statemap_{X \cap P}(u_1) = \state^\ell_1$ and $\statemap_{X \cap P}(u_2) = \state^\ell_2$. Due to $|X \cap \{u_1, u_2, w_{+,1}, \ldots, w_{+,4}\}| = 2$ and $X$ being canonical, we see that
  \begin{equation*}
    X \cap \{u_1, u_2\} = \begin{cases}
                           \emptyset, & \ell = 1, \\
                           \{u_1\}, & \ell \in \{2,3\}, \\
                           \{u_1, u_2\}, & \ell \in \{4,5\}.
                          \end{cases}
  \end{equation*}
  Hence, we only have to determine the root-connectivity of $u_1$ and possibly $u_2$ in $X \cap P$ for $\ell > 1$. They can only obtain root-connectivity via $a_{1,1}$, $a_{1,2}$, $a_{1,3}$, or $\bar{a}_{1,1}$. By the previous calculations, at least one of these is root-connected in $(X \cap P) \setminus \{u_1, u_2\}$ if and only if $\state^\ell_3 = \one_0$ or $\state^\ell_4 = \one_1$, which happens precisely when $\ell \in \{3,5\}$ as desired (as $\ell = 1$ is excluded). \qed
\end{proof}

\begin{lem}\label{thm:cvc_modtw_state_exists} 
  For every $\ell \in [5]$, there exists a canonical vertex cover $X_P^\ell$ of $P$ such that $|X_P^\ell| = 35$, $X_P^\ell \cap \{v_1, \ldots, v_5\} = \{v_1, \ldots, v_5\} \setminus \{v_\ell\}$, and $\statemap(X_P^\ell) = \state^\ell$. If $X$ is a vertex cover of $G$ with $\rvertex \in X$, $X \cap P = X_P^\ell$, and $\statemap_X(\{u_1$, $u_2$, $\bar{a}_{1,2}$, $a_{1,4}\}) \subseteq \{\zero, \one_1\}$, then every vertex of $X_P^\ell$ is root-connected in $X$.
\end{lem}

\begin{proof}
 We claim that 
 \begin{equation*}
  X_P^\ell = \left(\bigcup_{k \in [5] \setminus \{\ell\}} W_k \cup \{v_k\}\right) \cup \{\bar{a}_{1,1}, \bar{a}_{1,2}\} \cup \{a_{2 - [\state_2^\ell = 0], f} \sep f \in [4]\} \cup U_\ell \cup N(W_\ell),
 \end{equation*}
 where $U_1 = \emptyset$, $U_2 = U_3 = \{u_1\}$, $U_4 = U_5 = \{u_1, u_2\}$, is the desired vertex cover. Clearly, $X_P^\ell$ is canonical. By construction of $P$, we compute that
 \begin{equation*}
 	N(W_\ell) = \{\bar{v}_\ell, \bar{c}_{[s_1^\ell \neq \one_0], 1}, \bar{c}_{[s_1^\ell \neq \one_1], 2}\} \cup \{b_{[\state_2^\ell = 0],f}, c_{[\state_1^\ell = \state_2^\ell],f} \sep f \in [4]\} \cup W_{+,\ell},
 \end{equation*} 
 where $W_{+,1} = \{w_{+,1}, w_{+,2}\}, W_{+,2} = \{w_{+,3}\}, W_{+,3} = \{w_{+,4}\}, W_{+,4} = W_{+,5} = \emptyset$. Note that $|U_\ell| + |W_{+,\ell}| = 2$ and hence $|X_P^\ell| = 35$ for all $\ell \in [5]$.
 
 We proceed by verifying that $X_P^\ell$ is a vertex cover of $P$. The only non-trivial edges to consider are $\{a_{1,f}, c_{1,f}\}$, $f \in [4]$, and the edges between $\{u_1, u_2\}$ and $\{a_{1,f} \sep f \in [3]\}$. If $a_{1,f} \notin X_P^\ell$, then $\state_2^\ell \neq \zero$ which also implies that $\state_1^\ell = \state_2^\ell$ and hence $c_{1,f} \in X_P^\ell$, so the edge $\{a_{1,f}, c_{1,f}\}$, $f \in [4]$, is covered in all cases. If $1 \leq \ell \leq 3$, then $\state_2^\ell = \zero$, so $a_{1,f} \in X_P^\ell$ for all $f \in [4]$. If $4 \leq \ell \leq 5$, then $u_1, u_2 \in X$, so in either case the edges between $\{u_1, u_2\}$ and $\{a_{1,f} \sep f \in [3]\}$ are covered.
 
 Moving on to the second part, assume that $X$ is a vertex cover of $G$ with $\rvertex \in X$, $X \cap P = X_P^\ell$, and $\statemap_X(\{u_1, u_2, \bar{a}_{1,2}, a_{1,4}\}) \subseteq \{\zero, \one_1\}$. We only have to consider the vertices in $X_P^\ell \setminus N(\rvertex) \subseteq \{a_{1,f} \sep f \in [4]\} \cup \{\bar{a}_{1,1}, \bar{a}_{1,2}\}$. The statement immediately follows if $u_1$ or $u_2$ is root-connected in $X$, because they are adjacent to all vertices in $\{a_{1,f} \sep f \in [3]\} \cup \{\bar{a}_{1,1}\}$ and $a_{1,4}$ and $\bar{a}_{1,2}$ are handled by assumption. It remains to consider the case $u_1, u_2 \notin X$ which corresponds to $\ell = 1$, so we see that $a_{1,f}, c_{1,f} \in X$ for all $f \in [4]$ and $\bar{c}_{1,1} \in X$. Then, $a_{1,f}$ is root-connected via $c_{1,f}$ and $\bar{a}_{1,1}$ is root-connected via $\bar{c}_{1,1}$. \qed
\end{proof}

%

In the complete construction, we create long paths by repeatedly concatenating the path gadgets $P$. To study the \emph{state transitions} between two consecutive path gadgets, suppose that we have two copies $P^1$ and $P^2$ of $P$ such that the vertices $a_{1,4}$ and $\bar{a}_{1,2}$ in $P^1$ are joined to the vertices $u_1$ and $u_2$ in $P^2$. We denote the vertices of $P^1$ with a superscript $1$ and the vertices of $P^2$ with a superscript $2$, e.g., $a^1_{1,4}$ refers to the vertex $a_{1,4}$ of $P^1$. Again, suppose that $P^1$ and $P^2$ are embedded as induced subgraphs in a larger graph $G$ with a root vertex $\rvertex$ and that only the vertices $u_1, u_2^1, a_{1,4}^2, \bar{a}_{1,2}^2$ and the clique vertices $v^1_\ell, v^2_\ell$, $\ell \in [5]$, have neighbors outside of $P^1 + P^2 + \rvertex$. Let $X$ be a connected vertex cover of $G$ with $\rvertex \in X$. 

\begin{lem}\label{thm:cvc_modtw_path_transition}
  Suppose that $X$ is canonical with respect to $\{u_1^1, u_2^1\}$ and $\{u_2^1, u_2^2\}$, that $G[X]$ is connected and that $|X \cap P^1| \leq 35$ and $|X \cap P^2| \leq 35$, then $\statemap(X \cap P^1) = \state^{\ell_1}$ and $\statemap(X \cap P^2) = \state^{\ell_2}$ with $\ell_1 \leq \ell_2$. 
  
  Additionally, for each $\ell \in [5]$, the set $X^\ell = X^\ell_{P^1} \cup X^\ell_{P^2}$ is a vertex cover of $P^1 + P^2$ with $\statemap_{X^\ell}(\{u_1^1, u_2^1, a_{1,4}^2, \bar{a}_{1,2}^2\}) \subseteq \{\zero, \one_1\}$.
\end{lem}

\begin{proof}
 By \cref{thm:cvc_modtw_path_gadget_tight}, we see that there are $\ell_1, \ell_2 \in [5]$ such that $\statemap(X \cap P^1) = \state^{\ell_1}$ and $\statemap(X \cap P^2) = \state^{\ell_2}$. It remains to show that $\ell_1 \leq \ell_2$.
 
 Define $U^1 = \{a_{1,4}^1, \bar{a}_{1,2}^1\}$ and $U^2 = \{u_1^2, u_2^2\}$ and $U = U^1 \cup U^2$. By the assumption on how $P^1 + P^2 + \rvertex$ can be connected to the rest of the graph $G$, one can see that any path from $U$ to $\rvertex$ passes through some vertex in $(V(P_1) \cup V(P_2)) \cap N(\rvertex)$. Hence, we can determine whether the vertices of $X \cap U$ are root-connected in $X$ by just considering the graph $P^1 + P^2 + \rvertex$. 
 
 \newcommand{\statepair}{\bar{\state}}
 Consider the state pairs $\statepair^1 = (\statemap_{X \cap P^1}(\bar{a}_{1,2}^1), \statemap_{X \cap P^1}(a_{1,4}^2)) = (\state^{\ell_1}_3, \state^{\ell_2}_4)$ and $\statepair^2 = (\statemap_{X \cap P^2}(u_1^2), \statemap_{X \cap P^2}(u_2^2)) = (\state^{\ell_2}_1, \state^{\ell_2}_2)$. We claim that whenever $\ell_1 > \ell_2$ there is some edge in $G[U]$ that is not covered by $X$ or there is a vertex in $X \cap U$ that is not root-connected in $X$. There is an uncovered edge in $G[U]$ if and only if both $\statepair^1$ and $\statepair^2$ each contain at least one $\zero$. This shows that $(\ell_1, \ell_2) \notin \{4,5\} \times [3]$. Some vertex in $X \cap U$ is not root-connected in $X$ if and only if either $\statepair^1$ or $\statepair^2$ contains a $\one_0$ and the other one only contains two $\zero$s or if both contain no $\one_1$ at all. This shows that $(\ell_1, \ell_2) \notin \{(5,4),(3,2),(3,1),(2,1)\}$ and concludes the proof of the first part.
 
 For the second part, notice that $\statemap(X_{P^1}^\ell) = \statemap(X_{P^2}^\ell) = \state^\ell$ by \cref{thm:cvc_modtw_path_gadget_tight} and using the same approach as in the last paragraph, we see that for $\ell = \ell_1 = \ell_2$ all edges in $G[U]$ are covered and all vertices in $X^\ell$ are root-connected in $X^\ell$. \qed
\end{proof}

\cref{thm:cvc_modtw_path_transition} is the reason for the chosen numbering of the elements of $\states$. We say that a \emph{cheat occurs} if $\ell_1 < \ell_2$. Creating arbitrarily long paths of the path gadgets $P$, \cref{thm:cvc_modtw_path_transition} tells us that at most $|\states| - 1 = 4 = \Oh(1)$ cheats may occur on such a path.

\subsubsection{Complete Construction}

\newcommand{\nregions}{{4\ngrps\grpsize + 1}}
\newcommand{\ncolumns}{{\nclss(\nregions)}}
\newcommand{\sequence}{\mathbf{h}}

\begin{figure}
  \centering
  \begin{tikzpicture}
	\begin{pgfonlayer}{nodelayer}
		\node [style=tiny rectangle] (0) at (4, 8) {};
		\node [style=tiny rectangle] (1) at (4, 4) {};
		\node [style=tiny rectangle] (2) at (4, -0.75) {};
		\node [style=none] (3) at (4, 2.25) {$\vdots$};
		\node [style=filled] (4) at (0, 8) {};
		\node [style=filled] (5) at (0, 4) {};
		\node [style=filled] (6) at (0, -0.75) {};
		\node [style=filled] (7) at (1, 0.25) {};
		\node [style=filled] (8) at (1, 5) {};
		\node [style=filled] (9) at (1, 9) {};
		\node [style=filled] (10) at (10, 5) {};
		\node [style=filled] (11) at (12.25, 5) {};
		\node [style=none] (18) at (10.25, 5.75) {$z^{i,\ell}$};
		\node [style=none] (19) at (12.25, 5.75) {$\bar{z}^{i, \ell}$};
		\node [style=none] (20) at (0, -1.5) {$x^{i, \ell}_{(5,5,\ldots)}$};
		\node [style=none] (21) at (1.25, 1) {$\bar{x}^{i, \ell}_{(5,5,\ldots)}$};
		\node [style=none] (22) at (0, 3.25) {$x^{i, \ell}_{(2,1,\ldots)}$};
		\node [style=none] (23) at (1.25, 5.75) {$\bar{x}^{i, \ell}_{(2,1,\ldots)}$};
		\node [style=none] (24) at (-0.25, 7.25) {$x^{i, \ell}_{(1,1,\ldots)}$};
		\node [style=none] (25) at (1, 9.75) {$\bar{x}^{i, \ell}_{(1,1,\ldots)}$};
		\node [style=none] (26) at (4, 8.75) {$y^{i, \ell}_{(1,1,\ldots)}$};
		\node [style=none] (27) at (4, 4.75) {$y^{i, \ell}_{(2,1,\ldots)}$};
		\node [style=none] (28) at (4, 0) {$y^{i, \ell}_{(5,5,\ldots)}$};
		\node [style=tiny rectangle] (29) at (-8, 9) {};
		\node [style=tiny rectangle] (30) at (-8, 8) {};
		\node [style=tiny rectangle] (31) at (-8, 7) {};
		\node [style=tiny rectangle] (32) at (-8, 6) {};
		\node [style=tiny rectangle] (33) at (-8, 5) {};
		\node [style=none] (34) at (-10, 10) {};
		\node [style=none] (35) at (-7, 10) {};
		\node [style=none] (36) at (-10, 4) {};
		\node [style=none] (37) at (-7, 4) {};
		\node [style=none] (38) at (-9, 9) {$v_1^{i,1,\ell}$};
		\node [style=none] (39) at (-9, 8) {$v_2^{i,1,\ell}$};
		\node [style=none] (40) at (-9, 7) {$v_3^{i,1,\ell}$};
		\node [style=none] (41) at (-9, 6) {$v_4^{i,1,\ell}$};
		\node [style=none] (42) at (-9, 5) {$v_5^{i,1,\ell}$};
		\node [style=tiny rectangle] (43) at (-8, 2) {};
		\node [style=tiny rectangle] (44) at (-8, 1) {};
		\node [style=tiny rectangle] (45) at (-8, 0) {};
		\node [style=tiny rectangle] (46) at (-8, -1) {};
		\node [style=tiny rectangle] (47) at (-8, -2) {};
		\node [style=none] (48) at (-10, 3) {};
		\node [style=none] (49) at (-7, 3) {};
		\node [style=none] (50) at (-10, -3) {};
		\node [style=none] (51) at (-7, -3) {};
		\node [style=none] (52) at (-9, 2) {$v_1^{i,2,\ell}$};
		\node [style=none] (53) at (-9, 1) {$v_2^{i,2,\ell}$};
		\node [style=none] (54) at (-9, 0) {$v_3^{i,2,\ell}$};
		\node [style=none] (55) at (-9, -1) {$v_4^{i,2,\ell}$};
		\node [style=none] (56) at (-9, -2) {$v_5^{i,2,\ell}$};
		\node [style=none] (57) at (-8.5, -4) {$\vdots$};
		\node [style=none] (58) at (0, 2.25) {$\vdots$};
		\node [style=filled] (59) at (8, -2.75) {};
		\node [style=filled] (60) at (10, -2.75) {};
		\node [style=none] (61) at (8.5, -2.25) {$o^\ell$};
		\node [style=none] (62) at (10.5, -2.25) {$\bar{o}^\ell$};
		\node [style=none] (63) at (7.5, -2) {};
		\node [style=none] (64) at (7.25, -2.25) {};
	\end{pgfonlayer}
	\begin{pgfonlayer}{edgelayer}
		\draw [style=thin] (9) to (4);
		\draw [style=thin] (8) to (5);
		\draw [style=thin] (7) to (6);
		\draw [style=thin] (6) to (2);
		\draw [style=thin] (5) to (1);
		\draw [style=thin] (4) to (0);
		\draw [style=thin] (11) to (10);
		\draw [style=thin, in=0, out=150, looseness=0.50] (10) to (0);
		\draw [style=thin, in=0, out=-180, looseness=0.50] (10) to (1);
		\draw [style=thin, in=0, out=-90, looseness=0.50] (10) to (2);
		\draw [style=dashed cyan thick] (37.center) to (35.center);
		\draw [style=dashed cyan thick] (35.center) to (34.center);
		\draw [style=dashed cyan thick] (34.center) to (36.center);
		\draw [style=dashed cyan thick] (36.center) to (37.center);
		\draw [style=dashed cyan thick] (51.center) to (49.center);
		\draw [style=dashed cyan thick] (49.center) to (48.center);
		\draw [style=dashed cyan thick] (48.center) to (50.center);
		\draw [style=dashed cyan thick] (50.center) to (51.center);
		\draw [style=thin, in=180, out=0] (29) to (4);
		\draw [style=thin, in=0, out=-180, looseness=0.75] (4) to (43);
		\draw [style=thin, in=0, out=-180, looseness=0.50] (5) to (30);
		\draw [style=thin, in=0, out=-180] (5) to (43);
		\draw [style=thin, in=0, out=-180, looseness=0.50] (6) to (33);
		\draw [style=thin] (6) to (47);
		\draw [style=thin, in=90, out=-30, looseness=0.50] (1) to (59);
		\draw [style=thin] (59) to (60);
		\draw [style=thin dashed] (63.center) to (59);
		\draw [style=thin dashed] (64.center) to (59);
	\end{pgfonlayer}
\end{tikzpicture}
  \caption{The decoding gadget for group $i \in [\ngrps]$ and column $\ell \in [\ncolumns]$. The clause gadget for column $\ell$ consists of $o^\ell$ and $\bar{o}^\ell$ and represents clause $C_{\ell'}$, where $\ell' = (\ell - 1) \mod \nclss$. In this figure the truth assignment for group $i$ corresponding to $(2,1,\ldots) \in [5]^\grpsize$ satisfies clause $C_{\ell'}$.}  
  \label{fig:cvc_modtw_decoding}
\end{figure}
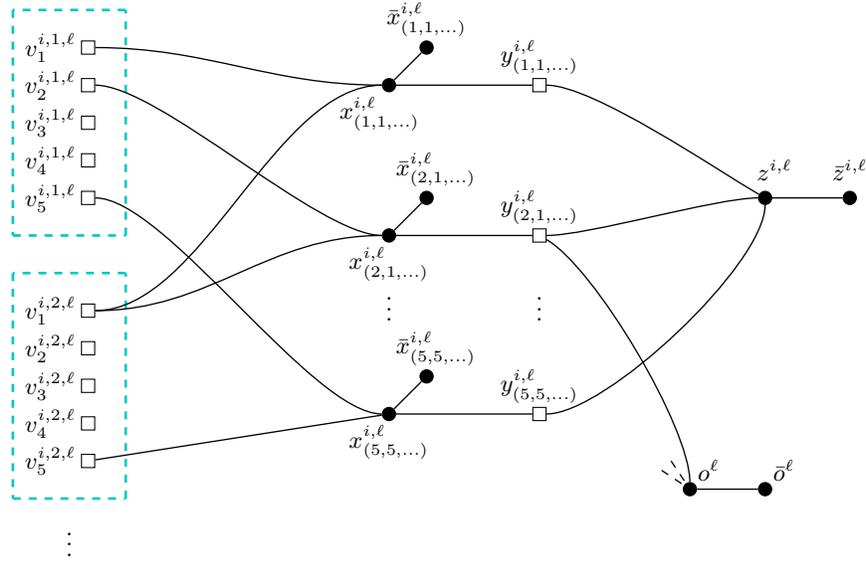

\paragraph{Setup.} 
Assume that \CVC can be solved in time $\Oh^*((5 - \eps)^{\tcpw(G)})$ for some $\eps > 0$. Given a \SAT-instance $\formula$ with $\nvars$ variables and $\nclss$ clauses, we construct an equivalent \CVC instance with twinclass-pathwidth approximately $\nvars \log_5(2)$ so that the existence of such an algorithm for \CVC would imply that \CNFSETH is false. 

We pick an integer $\vgrpsize$ only depending on $\eps$; the precise choice of $\vgrpsize$ will be discussed at a later point. The variables of $\formula$ are partitioned into groups of size at most $\vgrpsize$, resulting in $\ngrps = \lceil \nvars / \vgrpsize \rceil$ groups. Furthermore, we pick the smallest integer $\grpsize$ that satisfies $5^\grpsize \geq 2^\vgrpsize$. We now begin with the construction of the \CVC instance $(G = G(\formula, \vgrpsize), \budget)$.

We create the root vertex $\rvertex$ and attach a leaf $\rvertex'$ which forces $\rvertex$ into any connected vertex cover.
For every group $i \in [\ngrps]$, we create $\grpsize$ long path-like gadgets $P^{i, j}$, $j \in [\grpsize]$, where each $P^{i, j}$ consists of $\ncolumns$ copies $P^{i, j, \ell}$, $\ell \in [\ncolumns]$, of the path gadget $P$ and consecutive copies are connected by a join. More precisely, the vertices in some $P^{i, j, \ell}$ inherit their names from $P$ and the superscript of $P^{i, j, \ell}$ and for every $i \in [\ngrps]$, $j \in [\grpsize]$, $\ell \in [\ncolumns - 1]$, the output vertices $a_{1,4}^{i, j, \ell}$ and $\bar{a}_{1,2}^{i, j, \ell}$ are joined to the input vertices $u_1^{i, j, \ell + 1}$ and $u_2^{i, j, \ell + 1}$ of the next path gadget. The ends of each path $P^{i, j}$, namely the vertices $u_1^{i, j, 1}$, $u_2^{i, j, 1}$, $a_{1,4}^{i, j, \ncolumns}$, $\bar{a}_{1,2}^{i, j, \ncolumns}$ are made adjacent to the root $\rvertex$. 

For every group $i \in [\ngrps]$ and column $\ell \in [\ncolumns]$, we create a \emph{decoding gadget} $D^{i,\ell}$ in the same style as Cygan et al.~\cite{CyganNPPRW11arxiv} for \CVC parameterized by pathwidth. Every variable group $i$ has at most $2^\vgrpsize$ possible truth assignments and by choice of $\grpsize$ we have that $5^\grpsize \geq 2^\vgrpsize$, so we can find an injective mapping $\embedding \colon \{0,1\}^\vgrpsize \rightarrow [5]^\grpsize$ which assigns to each truth assignment $\tassign \in \{0,1\}^\vgrpsize$ a sequence $\embedding(\tassign) \in [5]^\grpsize$. For each sequence $\sequence = (h_1, \ldots, h_\grpsize) \in [5]^\grpsize$, we create vertices $x^{i, \ell}_\sequence$, $\bar{x}^{i, \ell}_\sequence$, $y^{i, \ell}_\sequence$ and edges $\{x^{i, \ell}_\sequence, \bar{x}^{i, \ell}_\sequence\}$, $\{x^{i, \ell}_\sequence, y^{i, \ell}_\sequence\}$, $\{y^{i, \ell}_\sequence, \rvertex\}$. Furthermore, we add the edge $\{x^{i, \ell}_\sequence, v^{i, j, \ell}_{h_j}\}$ for all $\sequence = (h_1, \ldots, h_\grpsize) \in [5]^\grpsize$ and $j \in [\grpsize]$. Finally, we create two adjacent vertices $z^{i, \ell}$ and $\bar{z}^{i, \ell}$ and edges $\{z^{i, \ell}, y^{i, \ell}_\sequence\}$ for all $\sequence \in [5]^\grpsize$. For every group $i \in [\ngrps]$ and column $\ell \in [\ncolumns]$, we bundle the the path gadgets $P^{i,j,\ell}$, $j \in [\grpsize]$, and the decoding gadget $D^{i, \ell}$ into the \emph{block} $B^{i, \ell}$.

Lastly, we construct the \emph{clause gadgets}.
We number the clauses of $\formula$ by $C_0, \ldots, C_{\nclss - 1}$. For every column $\ell \in [\ncolumns]$, we create an adjacent pair of vertices $o^{\ell}$ and $\bar{o}^{\ell}$. Let $\ell' \in [0, \nclss - 1]$ be the remainder of $(\ell - 1)$ modulo $\nclss$. For every $i \in [\ngrps]$, $\sequence \in \embedding(\{0,1\}^\vgrpsize)$, we add the edge $\{o^{\ell}, y^{i, \ell}_\sequence\}$ whenever $\embedding^{-1}(\sequence)$ is a truth assignment for variable group $i$ that satisfies clause $C_{\ell'}$. See \cref{fig:cvc_modtw_decoding} for a depiction of the decoding and clause gadgets and \cref{fig:cvc_modtw_schematic} for a high-level view of the whole construction.

\begin{figure}
  \centering
  \scalebox{.75}{\input{pictures/cvc_modtw_schematic.tikz}}
  \caption{The matrix structure of the constructed graph. Every $\nclss$ columns form a region.}  
  \label{fig:cvc_modtw_schematic}
\end{figure}

\begin{lem}\label{thm:cvc_modtw_sat_to_sol}
 If $\formula$ is satisfiable, then there exists a connected vertex cover $X$ of $G = G(\formula, \vgrpsize)$ of size $|X| \leq (35\ngrps \grpsize + (5^\grpsize + 2)\ngrps + 1)\ncolumns + 1 = \budget$.
\end{lem}

\begin{proof}
 Let $\tassign$ be a satisfying truth assignment of $\formula$ and let $\tassign^i$ denote the restriction of $\tassign$ to the $i$-th variable group for every $i \in [\ngrps]$ and let $\embedding(\tassign^i) = \sequence^i = (h^i_1, \ldots, h^i_\grpsize)$ be the corresponding sequence. 
 The connected vertex cover is given by 
 \begin{equation*}
   X = \{\rvertex\} \cup \bigcup_{\ell \in [\ncolumns]} \left(\{o^\ell\} \cup \bigcup_{i \in [\ngrps]} \left(\{y^{i, \ell}_{\sequence^i}, z^{i,\ell}\} \cup \bigcup_{\sequence \in [5]^\grpsize} \{x^{i,\ell}_\sequence\} \cup \bigcup_{j \in [\grpsize]} X_{P^{i,j,\ell}}^{h^i_j} \right) \right),
 \end{equation*}
 where $X_{P^{i,j,\ell}}^{h^i_j}$ refers to the sets given by \cref{thm:cvc_modtw_state_exists}. 
 
 Clearly, $|X| = \budget$, so it remains to prove that $X$ is a connected vertex cover. By \cref{thm:cvc_modtw_state_exists} and the second part of \cref{thm:cvc_modtw_path_transition} all edges induced by the path gadgets are covered by $X$ and all vertices on the path gadgets that belong to $X$ are root-connected, except for possibly the vertices at the ends, i.e. $\bigcup_{i \in [\ngrps]} \bigcup_{j \in [\grpsize]} \{u^{i,j,1}_1, u^{i,j,1}_2, a^{i,j,\ncolumns}_{1,4}, \bar{a}^{i,j,\ncolumns}_{1,2}\}$, but these are contained in the neighborhood of $\rvertex$ by construction.
 
 Fix $i \in [\ngrps]$, $\ell \in [\ncolumns]$, and consider the corresponding decoding gadget. Since $z^{i, \ell} \in X$ and $x^{i,\ell}_\sequence \in X$ for all $\sequence \in [5]^\grpsize$, all edges induced by the decoding gadget and all edges between the decoding gadget and the path gadgets are covered by $X$. Furthermore, since $o^{\ell} \in X$, all edges inside the clause gadget and all edges between the clause gadget and the decoding gadgets are covered by $X$. Hence, $X$ has to be a vertex cover of $G$.
 
 It remains to prove that the vertices in the decoding and clause gadgets that belong to $X$ are also root-connected. Again, fix $i \in [\ngrps]$, $\ell \in [\ncolumns]$, and $\sequence = (h_1, \ldots, h_\grpsize) \in [5]^\grpsize \setminus \{\sequence^i\}$. Since $\sequence \neq \sequence^i$, there is some $j \in [\grpsize]$ such that $v^{i,j,\ell}_{h_j} \in X$ by \cref{thm:cvc_modtw_state_exists} which connects $x_\sequence^{i, \ell}$ to the root $\rvertex$. The vertices $x_{\sequence^i}^{i, \ell}$ and $z^{i, \ell}$ are root-connected via $y^{i, \ell}_{\sequence^i} \in X$.
 
 We conclude by showing that $o^\ell$ is root-connected for all $\ell \in [\ncolumns]$. Since $\tassign$ is a satisfying truth assignment of $\formula$, there is some variable group $i \in [\ngrps]$ such that $\tassign^i$ already satisfies clause $C_{\ell'}$, where $\ell'$ is the remainder of $\ell - 1$ modulo $\nclss$. By construction of $G$ and $X$, the vertex $y^{i, \ell}_{\sequence^i} \in X$ is adjacent to $o^\ell$, since $\embedding(\tassign^i) = \sequence^i$, and connects $o^\ell$ to the root $\rvertex$. This shows that all vertices of $X$ are root-connected, so $G[X]$ has to be connected. \qed
\end{proof}

\begin{lem}\label{thm:cvc_modtw_sol_to_sat}
  If there exists a connected vertex cover $X$ of $G = G(\formula, \vgrpsize)$ of size $|X| \leq (35\ngrps \grpsize + (5^\grpsize + 2)\ngrps + 1)\ncolumns + 1 = \budget$, then $\formula$ is satisfiable.
\end{lem}

\begin{proof}
  We assume without loss of generality that $X$ is canonical with respect to each twinclass $\{u_1^{i, j, \ell}, u_2^{i, j, \ell}\}$, $i \in [\ngrps]$, $j \in [\grpsize]$, $\ell \in [\ncolumns]$. 
  
  We begin by arguing that $X$ has to satisfy $|X| = \budget$. First, we must have that $\rvertex \in X$, because $\rvertex$ has a neighbor of degree 1. By \cref{thm:cvc_modtw_path_gadget_lb}, we have that $|X \cap P^{i,j,\ell}| \geq 35$ for all $i \in [\ngrps]$, $j \in [\grpsize]$, $\ell \in [\ncolumns]$. In every decoding gadget, i.e.\ one for every $i \in [\ngrps]$ and $\ell \in [\ncolumns]$, the set $\{z^{i,\ell}\} \cup \bigcup_{\sequence \in [5]^\grpsize} x_\sequence^{i, \ell}$ has to be contained in $X$, since every vertex in this set has a neighbor of degree 1. Furthermore, to connect $z^{i,j}$ to $\rvertex$, at least one of the vertices $y^{i,\ell}_\sequence$, $\sequence \in [5]^\grpsize$, has to be contained in $X$. Hence, $X$ must contain at least $5^\grpsize + 2$ vertices per decoding gadget. Lastly, $o^\ell \in X$ for all $\ell \in [\ncolumns]$, since $o^\ell$ has a neighbor of degree 1. Since we have only considered disjoint vertex sets, this shows that $|X| = \budget$ and all of the previous inequalities have to be tight, in particular for every $i \in [\ngrps]$ and $\ell \in [\ncolumns]$, there is a unique $\sequence \in [5]^\grpsize$ such that $y^{i,\ell}_\sequence \in X$. 
  
  By \cref{thm:cvc_modtw_path_gadget_tight}, we know that $X$ assumes one of the five possible states on each $P^{i,j,\ell}$. Fix some $P^{i,j} = \bigcup_{\ell \in [\ncolumns]} P^{i,j,\ell}$ and note that due to \cref{thm:cvc_modtw_path_transition} the state can change at most four times along $P^{i,j}$. Such a state change is called a \emph{cheat}. Let $\gamma \in [0, 4 \ngrps \grpsize]$ and define the $\gamma$-th \emph{region} $R^\gamma = \bigcup_{i \in [\ngrps]} \bigcup_{j \in [\grpsize]} \bigcup_{\ell = \gamma \nclss + 1}^{(\gamma + 1) \nclss} P^{i,j,\ell}$. Since there are $\nregions$ regions and $\ngrps \grpsize$ many paths, there is at least one region $R^\gamma$ such that no cheat occurs in $R^\gamma$. We consider region $R^\gamma$ for the rest of the proof and read off a satisfying truth assignment from this region.
  
  For $i \in [\ngrps]$, let $\sequence^{i} = (h^i_1, \ldots, h^i_\grpsize) \in [5]^\grpsize$ such that $v^{i,j,\gamma \nclss + 1}_{h^i_j} \notin X$ for all $j \in [\grpsize]$; this is well-defined by \cref{thm:cvc_modtw_path_gadget_tight}. Since $R^\gamma$ does not contain any cheats, the definition of $\sequence^i$ is independent of which column $\ell \in [\gamma \nclss + 1, (\gamma + 1)\nclss]$ we consider. For every $i \in [\ngrps]$ and $\ell \in [\gamma \nclss + 1, (\gamma + 1)\nclss]$, we claim that $y^{i, \ell}_\sequence \in X$ if and only if $\sequence = \sequence^i$. We have already established that for every $i$ and $\ell$, there is exactly one $\sequence$ such that $y^{i, \ell}_\sequence \in X$. Consider the vertex $x^{i, \ell}_{\sequence^i} \in X$, its neighbors in $G$ are $v^{i, 1, \ell}_{h^i_1}, v^{i, 2, \ell}_{h^i_2}, \ldots, v^{i, \grpsize, \ell}_{h^i_\grpsize}$, $\bar{x}^{i,\ell}_{\sequence^i}$, and $y^{i,\ell}_{\sequence^i}$. By construction of $\sequence^i$ and the tight allocation of the budget, we have $(N(x^{i, \ell}_{\sequence^i}) \setminus \{y^{i, \ell}_{\sequence^i}\}) \cap X = \emptyset$. Therefore, $X$ has to include $y^{i, \ell}_{\sequence^i}$ to connect $x^{i, \ell}_{\sequence^i}$ to the root $\rvertex$. This shows the claim.
  
  For $i \in [\ngrps]$, we define the truth assignment $\tassign^i$ for group $i$ by taking an arbitrary truth assignment if $\sequence^i \notin \embedding(\{0,1\}^\vgrpsize)$ and setting $\tassign^i = \embedding^{-1}(\sequence^i)$ otherwise. By setting $\tassign = \bigcup_{i \in [\ngrps]} \tassign^i$ we obtain a truth assignment for all variables and we claim that $\tassign$ satisfies $\formula$. Consider some clause $C_{\ell'}$, $\ell' \in [0, \nclss - 1]$, and let $\ell = \gamma \nclss + \ell' + 1$. We have already argued that $o^\ell \in X$ and to connect $o^\ell$ to the root $\rvertex$, there has to be some $y^{i, \ell}_{\sequence} \in N(o^\ell) \cap X$. By the previous claim, $\sequence = \sequence^i$ for some $i \in [\ngrps]$ and therefore $\tassign^i$, and also $\tassign$, satisfy clause $C_{\ell'}$ due to the construction of $G$. Because the choice of $C_{\ell'}$ was arbitrary, $\tassign$ has to be a satisfying assignment of $\formula$. \qed
\end{proof}

\begin{lem}\label{thm:cvc_modtw_bound}
  The constructed graph $G = G(\formula, \vgrpsize)$ has $\tcpw(G) \leq \ngrps \grpsize + 3 \cdot 5^\grpsize + \Oh(1)$ and a path decomposition of $G^q = G / \tcpartition(G)$ of this width can be constructed in polynomial time.
\end{lem}

\begin{proof}
  By construction, all sets $\{u_1^{i,j,\ell}, u_2^{i,j,\ell}\}$, $i \in [\ngrps]$, $j \in [\grpsize]$, $\ell \in [\ncolumns]$, are twinclasses. Let $G'$ be the graph obtained by contracting each of these twinclasses, denoting the resulting vertex by $u^{i,j,\ell}$, then $G^q$ is a subgraph of $G'$. We will show that $\tcpw(G) = \pw(G^q) \leq \pw(G') \leq \ngrps \grpsize + 3 \cdot 5^\grpsize + \Oh(1)$ by giving an appropriate strategy for the mixed-search-game on $G'$ and applying \cref{thm:mixed_search}.
  
  \begin{algorithm}
    Place searchers on $\rvertex$ and $\rvertex'$\;
    Place searchers on $u^{i,j,1}$ for all $i \in [\ngrps]$, $j \in [\grpsize]$\;
    \For{$\ell \in [\ncolumns]$}
    {
      Place searchers on $o^\ell$ and $\bar{o}^\ell$\;
      \For{$i \in [\ngrps]$}
      {
        Place searchers on all vertices of the decoding gadget $D^{i, \ell}$\; 
        \For{$j \in [\grpsize]$}
        {
          Place searchers on all vertices of $P^{i,j,\ell} - \{u_1^{i,j,\ell}, u_2^{i,j,\ell}\}$\;
          Remove searcher from $u^{i,j,\ell}$ and place it on $u^{i,j,\ell + 1}$\;
          Remove searchers on $P^{i,j,\ell} - \{u_1^{i,j,\ell}, u_2^{i,j,\ell}\}$\;
        }
        Remove searchers on $D^{i, \ell}$\;
      }
      Remove searchers on $o^\ell$ and $\bar{o}^\ell$\;
    }
    \caption{Mixed-search-strategy for $G'$}
    \label{algo:cvc_modtw_search_game}
  \end{algorithm}
  \vspace*{-.5cm}
  The mixed-search-strategy for $G'$ described in \cref{algo:cvc_modtw_search_game} proceeds column by column and group by group in each column. The maximum number of placed searchers occurs on line $8$ and is $2 + \ngrps \grpsize + 2 + (3 \cdot 5^\grpsize + 2) + 61$. \qed
\end{proof}

\begin{thm}
  No algorithm can solve \CVC, given a path decomposition of $G^q = G / \tcpartition(G)$ of width $k$, in time $\Oh^*((5 - \eps)^k)$ for some $\eps > 0$, unless \CNFSETH fails.
\end{thm}

\begin{proof}
  Suppose there is an algorithm $\algo$ that solves \CVC in time $\Oh^*((5 - \eps)^k)$ for some $\eps > 0$ given a path decomposition of $G^q = G / \tcpartition(G)$ of width $k$. Given $\vgrpsize$, we define $\delta_1 < 1$ such that $(5 - \eps)^{\log_5(2)} = 2^{\delta_1}$ and $\delta_2$ such that $(5 - \eps)^{1 / \vgrpsize} = 2^{\delta_2}$. By picking $\vgrpsize$ large enough, we can ensure that $\delta = \delta_1 + \delta_2 < 1$. We show how to solve \SAT  using $\algo$ in time $\Oh^*(2^{\delta \nvars})$, where $\nvars$ is the number of variables, thus contradicting \CNFSETH.
  
  Given a \SAT instance $\formula$, construct $G = G(\formula, \vgrpsize)$ and the path decomposition from \cref{thm:cvc_modtw_bound} in polynomial time, as we have $\vgrpsize = \Oh(1)$ and hence $\grpsize = \Oh(1)$. We run $\algo$ on $G$ and return its answer. This is correct by \cref{thm:cvc_modtw_sat_to_sol} and \cref{thm:cvc_modtw_sol_to_sat}. Due to \cref{thm:cvc_modtw_bound}, the running time is
  \begin{alignat*}{6}
    \phantom{\leq} \,\, & \Oh^*\left( (5 - \eps)^{\ngrps \grpsize + 3 \cdot 5^\grpsize + \Oh(1)} \right)
    & \,\,\leq\,\, & \Oh^*\left( (5 - \eps)^{\ngrps \grpsize} \right) 
    & \,\,\leq\,\, & \Oh^*\left( (5 - \eps)^{\lceil \frac{\nvars}{\vgrpsize} \rceil \grpsize} \right) \\
    \leq \,\, & \Oh^*\left( (5 - \eps)^{\frac{\nvars}{\vgrpsize} \grpsize} \right) 
    & \,\,\leq\,\, & \Oh^*\left( (5 - \eps)^{\frac{\nvars}{\vgrpsize} \lceil \log_5(2^\vgrpsize) \rceil} \right) 
    & \,\,\leq\,\, & \Oh^*\left( (5 - \eps)^{\frac{\nvars}{\vgrpsize} \log_5(2^\vgrpsize)} (5 - \eps)^{\frac{\nvars}{\vgrpsize}} \right)\\
    \leq \,\, & \Oh^*\left( 2^{\delta_1 \vgrpsize \frac{\nvars}{\vgrpsize}} 2^{\delta_2 \nvars} \right)
    & \,\,\leq\,\, & \Oh^*\left( 2^{(\delta_1 + \delta_2) \nvars} \right)
    & \,\,\leq\,\, & \Oh^*\left( 2^{\delta \nvars} \right),
    \end{alignat*}
  hence completing the proof.  \qed
\end{proof}

\subsection{Feedback Vertex Set}
\label{sec:modtw_fvs_lb}

This subsection is devoted to proving that \FVS parameterized by twinclass-pathwidth cannot be solved in time $\Oh^*((5-\eps)^{\tcpw(G)})$ for some $\eps > 0$ unless the \SETH fails. The main challenge is the design of the path gadget. The decoding gadgets are adapted from the lower bound constructions for \OCT by Hegerfeld and Kratsch~\cite{HegerfeldK22} which rely on \emph{arrows} that are adapted from Lokshtanov et al.~\cite{LokshtanovMS18}. We remark that our construction will rely on false twinclasses and not true twinclasses, because in the algorithm for \FVS it can already be seen that true twinclasses only admit four distinct states instead of the desired five.

\renewcommand{\nregions}{{4\ngrps\grpsize + 1}}
\renewcommand{\ncolumns}{{\nclss(\nregions)}}
\renewcommand{\sequence}{\mathbf{h}}
\newcommand{\arrow}{A}

\subsubsection*{Triangle edges.}
Given two vertices $u$ and $v$, by \emph{adding a triangle edge between $u$ and $v$} we mean that we add a new vertex $w_{\{u,v\}}$ and the edges $\{u, v\}$, $\{u, w_{\{u,v\}}\}$, $\{w_{\{u,v\}}, v\}$, so that the three vertices $u$, $v$, $w_{\{u,v\}}$ induce a triangle. The vertex $w_{\{u,v\}}$ will not receive any further neighbors in the construction. Any feedback vertex set $X$ has to intersect $\{u, v, w_{\{u,v\}}\}$ and since $w_{\{u,v\}}$ has only degree 2, we can always assume that $w_{\{u,v\}} \notin X$. In this way, a triangle edge naturally implements a logical or between $u$ and $v$. 

\subsubsection*{Arrows.} 
Given two vertices $u$ and $v$, by \emph{adding an arrow from $u$ to $v$} we mean that we add three vertices $x_{uv}$, $y_{uv}$, $z_{uv}$ and the edges $\{u, x_{uv}\}$, $\{u, y_{uv}\}$, $\{x_{uv}, y_{uv}\}$, $\{y_{uv}, z_{uv}\}$, $\{y_{uv}, v\}$, $\{z_{uv}, v\}$, i.e., we are essentially adding two consecutive triangle edges between $u$ and $v$. The resulting graph is denoted by $\arrow(u,v)$ and $u$ is the \emph{tail} and $v$ the \emph{head} of the arrow. None of the vertices in $V(\arrow(u,v)) \setminus \{u,v\}$ will receive any further neighbors in the construction. The construction of an arrow is symmetric, but the direction will be relevant for constructing a cycle packing that witnesses a lower bound on the size of a feedback vertex set. 

We use arrows to propagate deletions throughout the graph. Let $X$ be a feedback vertex set. If $u \notin X$, then we can resolve both triangles simultaneously by putting $y_{uv}$ into $X$. If $u \in X$, then the first triangle is already resolved and we can safely put $v$ into $X$, hence propagating the deletion from $u$ to $v$. The former solution is called the \emph{passive} solution of the arrow and the latter is the \emph{active} solution. Using simple exchange arguments, we see that it is sufficient to only consider feedback vertex sets that on each arrow either use the passive solution or the active solution.

\subsubsection*{Setup.} 
Assume that \FVS can be solved in time $\Oh^*((5 - \eps)^{\tcpw(G)})$ for some $\eps > 0$. Given a $\clss$-\SAT-instance $\formula$ with $\nvars$ variables and $\nclss$ clauses, we construct an equivalent \FVS instance with twinclass-pathwidth approximately $\nvars \log_5(2)$ so that the existence of such an algorithm for \FVS would imply that \SETH is false. 

We pick an integer $\vgrpsize$ only depending on $\eps$; the precise choice of $\vgrpsize$ will be discussed at a later point. The variables of $\formula$ are partitioned into groups of size at most $\vgrpsize$, resulting in $\ngrps = \lceil \nvars / \vgrpsize \rceil$ groups. Furthermore, we pick the smallest integer $\grpsize$ that satisfies $5^\grpsize \geq 2^\vgrpsize$. We now begin with the construction of the FVS instance $(G = G(\formula, \vgrpsize), \budget)$.

\subsubsection*{Root.} We create a distinguished vertex $\rvertex$ called the \emph{root} which will be connected to several vertices throughout the construction. Given a vertex subset $Y \subseteq V(G)$ with $\rvertex \in Y$, we say that a vertex $v \in Y$ is \emph{root-connected} in $Y$ if there is a $v,\rvertex$-path in $G[Y]$. We will just say \emph{root-connected} if $Y$ is clear from the context.  The construction and choice of budget will ensure that the root vertex $\rvertex$ cannot be deleted by the desired feedback vertex sets. 

\begin{figure}
  \centering
  \scalebox{0.9}{\begin{tikzpicture}
	\begin{pgfonlayer}{nodelayer}
		\node [style=filled] (58) at (17, 0) {};
		\node [style=tiny rectangle] (59) at (21, 5) {};
		\node [style=filled] (60) at (17, 2) {};
		\node [style=filled] (61) at (21, 3) {};
		\node [style=filled] (62) at (21, 1) {};
		\node [style=filled] (63) at (21, -1) {};
		\node [style=filled] (64) at (21, -3) {};
		\node [style=filled] (65) at (21, 7) {};
		\node [style=none] (66) at (24, 7) {$v_3$};
		\node [style=none] (67) at (25.45, 5) {$v_1, v_2, v_4, v_5$};
		\node [style=none] (68) at (24, 3) {$v_4$};
		\node [style=none] (69) at (24.5, 1) {$v_4, v_5$};
		\node [style=none] (70) at (25.45, -1) {$v_2, v_3, v_4, v_5$};
		\node [style=none] (71) at (25.45, -3) {$v_1, v_2, v_3, v_5$};
		\node [style=none] (72) at (23.5, 7) {};
		\node [style=none] (73) at (23.5, 5) {};
		\node [style=none] (74) at (23.5, 3) {};
		\node [style=none] (75) at (23.5, 1) {};
		\node [style=none] (76) at (23.5, -1) {};
		\node [style=none] (77) at (23.5, -3) {};
		\node [style=none] (78) at (17, 2.5) {$u_1$};
		\node [style=none] (79) at (17, -0.75) {$u_2$};
		\node [style=none] (80) at (21.5, 4.5) {$c_1$};
		\node [style=none] (81) at (21.5, 6.5) {$c_0$};
		\node [style=none] (82) at (21.5, -3.5) {$a_4$};
		\node [style=none] (83) at (21.5, -1.5) {$a_3$};
		\node [style=none] (84) at (21.5, 0.5) {$a_2$};
		\node [style=none] (85) at (21.5, 2.5) {$a_1$};
		\node [style=filled] (86) at (34, 7) {};
		\node [style=filled] (87) at (34, 4.5) {};
		\node [style=filled] (88) at (34, 2) {};
		\node [style=filled] (89) at (34, -0.5) {};
		\node [style=filled] (90) at (34, -3) {};
		\node [style=none] (91) at (34, 6.25) {$v_1$};
		\node [style=none] (92) at (33, 8) {};
		\node [style=none] (93) at (35, 8) {};
		\node [style=none] (94) at (33, -4) {};
		\node [style=none] (95) at (35, -4) {};
		\node [style=none] (96) at (34, 3.75) {$v_2$};
		\node [style=none] (97) at (34, 1.25) {$v_3$};
		\node [style=none] (98) at (34, -1.25) {$v_4$};
		\node [style=none] (99) at (34, -3.75) {$v_5$};
		\node [style=none] (100) at (34, -4.625) {clique using};
		\node [style=none] (101) at (32, 7) {};
		\node [style=none] (102) at (32, 4.5) {};
		\node [style=none] (103) at (32, 2) {};
		\node [style=none] (104) at (32, -0.5) {};
		\node [style=none] (105) at (32, -3) {};
		\node [style=none] (106) at (31, 7) {$a_4, c_1$};
		\node [style=none] (107) at (30.5, 4.5) {$a_3, a_4, c_1$};
		\node [style=none] (108) at (30.5, 2) {$a_3, a_4, c_0$};
		\node [style=none] (109) at (30, -0.5) {$a_1, a_2, a_3, c_1$};
		\node [style=none] (110) at (30, -3) {$a_2, a_3, a_4, c_1$};
		\node [style=filled] (111) at (39, 7.5) {};
		\node [style=filled] (112) at (39, 6.5) {};
		\node [style=filled] (114) at (39, 4) {};
		\node [style=filled] (115) at (39, 2.5) {};
		\node [style=filled] (116) at (39, 1.5) {};
		\node [style=filled] (117) at (39, 0) {};
		\node [style=filled] (118) at (39, -1) {};
		\node [style=filled] (119) at (39, -2.5) {};
		\node [style=filled] (120) at (39, -3.5) {};
		\node [style=tiny rectangle] (121) at (39, 5) {};
		\node [style=none] (122) at (39, 8) {$b_{1,1}$};
		\node [style=none] (123) at (39.75, 6) {$b_{1,2}$};
		\node [style=none] (124) at (41, 7.5) {};
		\node [style=none] (125) at (41, 6.5) {};
		\node [style=none] (126) at (41, 5) {};
		\node [style=none] (127) at (41, 4) {};
		\node [style=none] (128) at (41, 2.5) {};
		\node [style=none] (129) at (41, 1.5) {};
		\node [style=none] (130) at (41, 0) {};
		\node [style=none] (131) at (41, -1) {};
		\node [style=none] (132) at (34.5, 6.75) {};
		\node [style=none] (133) at (34.5, 6.5) {};
		\node [style=none] (134) at (34.5, 6.25) {};
		\node [style=none] (135) at (34.5, 6) {};
		\node [style=none] (136) at (34.5, 2.5) {};
		\node [style=none] (137) at (34.5, 2.25) {};
		\node [style=none] (138) at (34.5, 1.75) {};
		\node [style=none] (139) at (34.5, 1.5) {};
		\node [style=none] (140) at (34.5, -0.75) {};
		\node [style=none] (141) at (34.5, -0.25) {};
		\node [style=none] (142) at (34.5, 0) {};
		\node [style=none] (143) at (34.5, 0.25) {};
		\node [style=none] (144) at (34.5, -2.75) {};
		\node [style=none] (145) at (34.5, -2.5) {};
		\node [style=none] (146) at (34.5, -2.25) {};
		\node [style=none] (147) at (34.5, -2) {};
		\node [style=none] (148) at (39, 5.5) {$b_{2,1}$};
		\node [style=none] (149) at (39.75, 3.5) {$b_{2,2}$};
		\node [style=none] (150) at (39, 3) {$b_{3,1}$};
		\node [style=none] (151) at (39.75, 1) {$b_{3,2}$};
		\node [style=none] (152) at (39.25, 0.5) {$b_{4,1}$};
		\node [style=none] (153) at (39.75, -1.5) {$b_{4,2}$};
		\node [style=none] (154) at (39.75, -3) {$b_{5,1}$};
		\node [style=none] (155) at (39.75, -4) {$b_{5,2}$};
		\node [style=none] (156) at (38, 8.5) {};
		\node [style=none] (157) at (40.5, 8.5) {};
		\node [style=none] (158) at (40.5, -4.5) {};
		\node [style=none] (159) at (38, -4.5) {};
		\node [style=none] (160) at (39.25, -5.125) {complete 5-partite};
		\node [style=tiny rectangle] (161) at (17, -4.5) {};
		\node [style=none] (162) at (18, -4.5) {$=$};
		\node [style=none] (163) at (22, -4.5) {adjacent to the root $\hat{r}$};
		\node [style=none] (164) at (34, -5.25) {triangle edges};
		\node [style=none] (165) at (39.25, -5.75) {using triangle edges};
		\node [style=none] (166) at (18, 8) {Path Gadget:};
		\node [style=none] (167) at (16.5, -5.5) {};
		\node [style=none] (168) at (17.5, -5.5) {};
		\node [style=none] (169) at (18, -5.5) {$=$};
		\node [style=none] (170) at (22, -5.5) {triangle edge};
	\end{pgfonlayer}
	\begin{pgfonlayer}{edgelayer}
		\draw (62) to (63);
		\draw (60) to (61);
		\draw (61) to (58);
		\draw (58) to (62);
		\draw (62) to (60);
		\draw (60) to (63);
		\draw (63) to (58);
		\draw (59) to (61);
		\draw [style=dark green] (65) to (59);
		\draw [style=dark green] (63) to (64);
		\draw [style=dark green, bend left=15] (64) to (62);
		\draw [style=dark green, bend left] (64) to (61);
		\draw [style=dark green] (65) to (72.center);
		\draw [style=dark green] (59) to (73.center);
		\draw [style=dark green] (61) to (74.center);
		\draw [style=dark green] (62) to (75.center);
		\draw [style=dark green] (63) to (76.center);
		\draw [style=dark green] (64) to (77.center);
		\draw [style=dashed cyan thick] (93.center) to (95.center);
		\draw [style=dashed cyan thick] (95.center) to (94.center);
		\draw [style=dashed cyan thick] (94.center) to (92.center);
		\draw [style=dashed cyan thick] (92.center) to (93.center);
		\draw [style=dark green] (105.center) to (90);
		\draw [style=dark green] (104.center) to (89);
		\draw [style=dark green] (103.center) to (88);
		\draw [style=dark green] (102.center) to (87);
		\draw [style=dark green] (101.center) to (86);
		\draw (111) to (112);
		\draw (111) to (124.center);
		\draw (124.center) to (112);
		\draw (112) to (125.center);
		\draw (125.center) to (111);
		\draw (121) to (126.center);
		\draw (126.center) to (114);
		\draw (114) to (127.center);
		\draw (127.center) to (121);
		\draw (115) to (128.center);
		\draw (128.center) to (116);
		\draw (116) to (129.center);
		\draw (129.center) to (115);
		\draw (117) to (130.center);
		\draw (131.center) to (117);
		\draw [style=dark green] (87) to (111);
		\draw [style=dark green] (87) to (112);
		\draw [style=dark green] (87) to (115);
		\draw [style=dark green] (87) to (116);
		\draw [style=dark green] (87) to (117);
		\draw [style=dark green] (87) to (118);
		\draw [style=dark green] (87) to (119);
		\draw [style=dark green] (87) to (120);
		\draw [style=dark green] (86) to (132.center);
		\draw [style=dark green] (86) to (133.center);
		\draw [style=dark green] (86) to (134.center);
		\draw [style=dark green] (86) to (135.center);
		\draw [style=dark green] (88) to (136.center);
		\draw [style=dark green] (88) to (137.center);
		\draw [style=dark green] (88) to (138.center);
		\draw [style=dark green] (88) to (139.center);
		\draw [style=dark green] (89) to (141.center);
		\draw [style=dark green] (89) to (142.center);
		\draw [style=dark green] (89) to (143.center);
		\draw [style=dark green] (89) to (140.center);
		\draw [style=dark green] (147.center) to (90);
		\draw [style=dark green] (146.center) to (90);
		\draw [style=dark green] (90) to (145.center);
		\draw [style=dark green] (90) to (144.center);
		\draw [style=dashed cyan thick] (159.center) to (156.center);
		\draw [style=dashed cyan thick] (156.center) to (157.center);
		\draw [style=dashed cyan thick] (157.center) to (158.center);
		\draw [style=dashed cyan thick] (158.center) to (159.center);
		\draw [style=dark green] (167.center) to (168.center);
	\end{pgfonlayer}
\end{tikzpicture}}
  \caption{The superscripts in vertex names are omitted and the edges between the auxiliary vertices, connectivity vertices and clique vertices are not drawn directly for visual clarity. All vertices that are depicted with a rectangle are adjacent to the root vertex $\rvertex$. The thick green edges denote triangle edges. The vertices inside the dashed rectangle induce a 5-clique or a complete 5-partite graph using triangle edges. The edges from the output vertices to the next pair of input vertices are hinted at.}
  \label{fig:fvs_modtw_path}
  \vspace*{-0.5cm}
\end{figure}
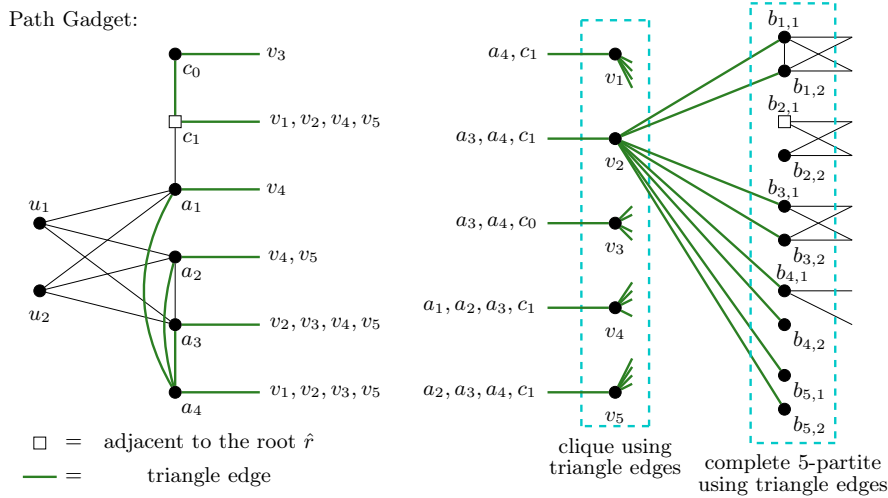

\subsubsection*{Path gadgets.}
For every $i \in [\ngrps]$, $j \in [\grpsize]$, $\ell \in [\ncolumns]$, we create a path gadget $P^{i, j, \ell}$ that consists of two \emph{input} vertices $u^{i, j, \ell}_1$, $u^{i, j, \ell}_2$ forming a false twinclass; four \emph{auxiliary} vertices $a^{i, j, \ell}_1$, $\ldots$, $a^{i, j, \ell}_4$; two \emph{connectivity} vertices $c^{i, j, \ell}_0$, $c^{i, j, \ell}_1$; five \emph{clique} vertices $v^{i, j, \ell}_1$, \ldots, $v^{i, j, \ell}_5$; and ten \emph{output} vertices in pairs of two $b^{i, j, \ell}_{1,1}$, $b^{i, j, \ell}_{1,2}$, $b^{i, j, \ell}_{2,1}$, $b^{i, j, \ell}_{2,2}$, \ldots, $b^{i, j, \ell}_{5,2}$. We add a join between the input vertices $u^{i, j, \ell}_1$, $u^{i, j, \ell}_2$ and the first three auxiliary vertices $a^{i, j, \ell}_1$, $a^{i, j, \ell}_2$, $a^{i, j, \ell}_3$, furthermore we add the edges $\{a^{i, j, \ell}_2, a^{i, j, \ell}_3\}$, $\{a^{i, j, \ell}_1, c^{i, j, \ell}_1\}$, and $\{b^{i, j, \ell}_{1,1}, b^{i, j, \ell}_{1,2}\}$. The vertices $c^{i, j, \ell}_1$ and $b^{i, j, \ell}_{2,1}$ are made adjacent to the root $\rvertex$. We add triangle edges between $a^{i, j, \ell}_4$ and the other auxiliary vertices $a^{i, j, \ell}_1$, $a^{i, j, \ell}_2$, $a^{i, j, \ell}_3$ and we add a triangle edge between $c^{i, j, \ell}_0$ and $c^{i, j, \ell}_1$. We add a triangle edge between every pair of distinct clique vertices $v^{i, j, \ell}_\varphi$, $\varphi \in [5]$, and every pair of output vertices $b^{i, j, \ell}_{\varphi, \gamma}$ and $b^{i, j, \ell}_{\varphi', \gamma'}$ with $\varphi \neq \varphi' \in [5]$ and $\gamma, \gamma' \in \{1,2\}$. For all $\varphi \in [5]$, we add a triangle edge between $v^{i, j, \ell}_\varphi$ and every $b^{i, j, \ell}_{\psi, \gamma}$ for $\psi \in [5] \setminus \{\varphi\}$ and $\gamma \in \{1,2\}$. 
We finish the construction of $P^{i, j, \ell}$ by describing how to connect the clique vertices $v^{i, j, \ell}_\varphi$, $\varphi \in [5]$, to the left side of $P^{i, j, \ell}$. For each $\varphi \in [5]$, we add triangle edges between $v^{i, j, \ell}_\varphi$ and one or several \emph{target} vertices on the left side of $P^{i, j, \ell}$. The target vertices, depending on $\varphi \in [5]$, are
\begin{itemize}
  \item for $\varphi = 1$: $a^{i, j, \ell}_4$ and $c^{i, j, \ell}_1$;
  \item for $\varphi = 2$: $a^{i, j, \ell}_3$, $a^{i, j, \ell}_4$, and $c^{i, j, \ell}_1$;
  \item for $\varphi = 3$: $a^{i, j, \ell}_3$, $a^{i, j, \ell}_4$, and $c^{i, j, \ell}_0$;
  \item for $\varphi = 4$: $a^{i, j, \ell}_1$, $a^{i, j, \ell}_2$, $a^{i, j, \ell}_3$, and $c^{i, j, \ell}_1$;
  \item for $\varphi = 5$: $a^{i, j, \ell}_2$, $a^{i, j, \ell}_3$, $a^{i, j, \ell}_4$, and $c^{i, j, \ell}_1$.
\end{itemize}
Finally, for $\ell \in [\ncolumns - 1]$, we connect $P^{i, j, \ell}$ to $P^{i, j, \ell + 1}$ by adding a join between the output pair $b^{i, j, \ell}_{\varphi, 1}$, $b^{i, j, \ell}_{\varphi, 2}$ and the next input vertices $u^{i, j, \ell + 1}_1$, $u^{i, j, \ell + 1}_2$ for every $\varphi \in \{1, 2, 3\}$ and we join the vertex $b^{i, j, \ell}_{4, 1}$ to $u^{i, j, \ell + 1}_1$ and $u^{i, j, \ell + 1}_2$. This concludes the description of the path gadgets, cf.~\cref{fig:fvs_modtw_path}.

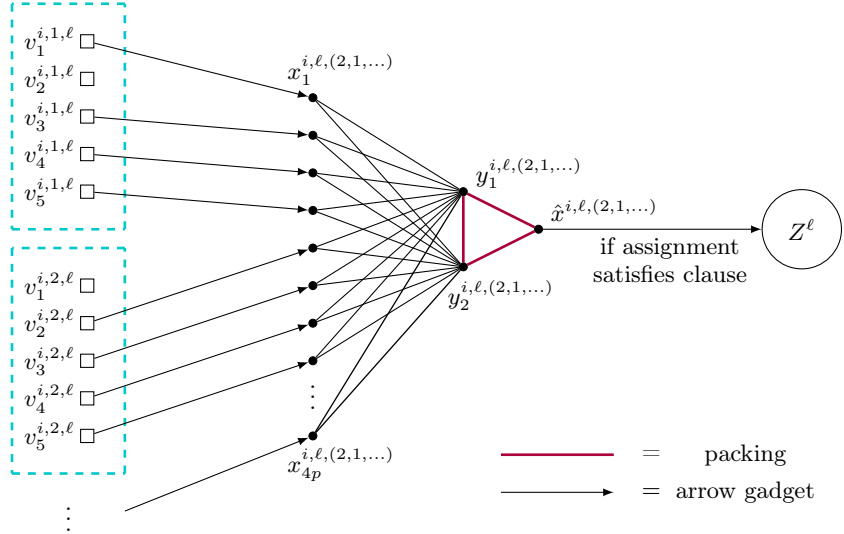
\begin{figure}
  \centering
  \begin{tikzpicture}
	\begin{pgfonlayer}{nodelayer}
		\node [style=nano] (0) at (-2, 3) {};
		\node [style=nano] (1) at (-2, 2) {};
		\node [style=nano] (2) at (-2, 1) {};
		\node [style=nano] (3) at (-2, 0) {};
		\node [style=nano] (4) at (-2, -1) {};
		\node [style=nano] (5) at (-2, -2) {};
		\node [style=nano] (6) at (-2, -3) {};
		\node [style=nano] (7) at (-2, -4) {};
		\node [style=none] (8) at (-2, -4.75) {$\vdots$};
		\node [style=nano] (9) at (2, 0.5) {};
		\node [style=nano] (10) at (2, -1.5) {};
		\node [style=nano] (11) at (4, -0.5) {};
		\node [style=tiny rectangle] (12) at (-8, 4.5) {};
		\node [style=tiny rectangle] (13) at (-8, 3.5) {};
		\node [style=tiny rectangle] (14) at (-8, 2.5) {};
		\node [style=tiny rectangle] (15) at (-8, 1.5) {};
		\node [style=tiny rectangle] (16) at (-8, 0.5) {};
		\node [style=none] (17) at (-10, 5.5) {};
		\node [style=none] (18) at (-7, 5.5) {};
		\node [style=none] (19) at (-10, -0.5) {};
		\node [style=none] (20) at (-7, -0.5) {};
		\node [style=none] (21) at (-9, 4.5) {$v_1^{i,1,\ell}$};
		\node [style=none] (22) at (-9, 3.5) {$v_2^{i,1,\ell}$};
		\node [style=none] (23) at (-9, 2.5) {$v_3^{i,1,\ell}$};
		\node [style=none] (24) at (-9, 1.5) {$v_4^{i,1,\ell}$};
		\node [style=none] (25) at (-9, 0.5) {$v_5^{i,1,\ell}$};
		\node [style=tiny rectangle] (26) at (-8, -2) {};
		\node [style=tiny rectangle] (27) at (-8, -3) {};
		\node [style=tiny rectangle] (28) at (-8, -4) {};
		\node [style=tiny rectangle] (29) at (-8, -5) {};
		\node [style=tiny rectangle] (30) at (-8, -6) {};
		\node [style=none] (31) at (-10, -1) {};
		\node [style=none] (32) at (-7, -1) {};
		\node [style=none] (33) at (-10, -7) {};
		\node [style=none] (34) at (-7, -7) {};
		\node [style=none] (35) at (-9, -2) {$v_1^{i,2,\ell}$};
		\node [style=none] (36) at (-9, -3) {$v_2^{i,2,\ell}$};
		\node [style=none] (37) at (-9, -4) {$v_3^{i,2,\ell}$};
		\node [style=none] (38) at (-9, -5) {$v_4^{i,2,\ell}$};
		\node [style=none] (39) at (-9, -6) {$v_5^{i,2,\ell}$};
		\node [style=none] (40) at (-8.5, -8) {$\vdots$};
		\node [style=none] (46) at (5.75, 0) {$\hat{x}^{i,\ell,(2,1,\ldots)}$};
		\node [style=none] (47) at (3.75, 1) {$y^{i,\ell,(2,1,\ldots)}_1$};
		\node [style=none] (48) at (3, -2.25) {$y^{i,\ell,(2,1,\ldots)}_2$};
		\node [style=big empty] (49) at (11, -0.5) {$Z^\ell$};
		\node [style=none] (51) at (7.5, -1) {if assignment};
		\node [style=none] (52) at (7.5, -1.675) {satisfies clause};
		\node [style=none] (53) at (3, -7.5) {};
		\node [style=none] (54) at (6, -7.5) {};
		\node [style=none] (55) at (7, -7.5) {$=$};
		\node [style=none] (56) at (9.5, -7.5) {arrow gadget};
		\node [style=none] (57) at (3, -6.5) {};
		\node [style=none] (58) at (6, -6.5) {};
		\node [style=none] (59) at (7, -6.5) {=};
		\node [style=none] (60) at (9.5, -6.5) {packing};
		\node [style=nano] (61) at (-2, -6) {};
		\node [style=none] (62) at (-1.25, 3.75) {$x^{i,\ell,(2,1,\ldots)}_1$};
		\node [style=none] (63) at (-1.25, -6.75) {$x^{i,\ell,(2,1,\ldots)}_{4p}$};
		\node [style=none] (64) at (-7, -8) {};
	\end{pgfonlayer}
	\begin{pgfonlayer}{edgelayer}
		\draw (0) to (9);
		\draw (9) to (1);
		\draw (2) to (9);
		\draw (9) to (3);
		\draw (4) to (9);
		\draw (9) to (5);
		\draw (9) to (6);
		\draw (9) to (7);
		\draw (0) to (10);
		\draw (1) to (10);
		\draw (10) to (2);
		\draw (3) to (10);
		\draw (10) to (4);
		\draw (5) to (10);
		\draw (10) to (6);
		\draw (10) to (7);
		\draw [style=packing] (11) to (9);
		\draw [style=packing] (9) to (10);
		\draw [style=packing] (10) to (11);
		\draw [style=dashed cyan thick] (20.center) to (18.center);
		\draw [style=dashed cyan thick] (18.center) to (17.center);
		\draw [style=dashed cyan thick] (17.center) to (19.center);
		\draw [style=dashed cyan thick] (19.center) to (20.center);
		\draw [style=dashed cyan thick] (34.center) to (32.center);
		\draw [style=dashed cyan thick] (32.center) to (31.center);
		\draw [style=dashed cyan thick] (31.center) to (33.center);
		\draw [style=dashed cyan thick] (33.center) to (34.center);
		\draw [style=double arrow] (12) to (0);
		\draw [style=double arrow] (14) to (1);
		\draw [style=double arrow] (15) to (2);
		\draw [style=double arrow] (16) to (3);
		\draw [style=double arrow] (27) to (4);
		\draw [style=double arrow] (28) to (5);
		\draw [style=double arrow] (29) to (6);
		\draw [style=double arrow] (30) to (7);
		\draw [style=double arrow] (11) to (49);
		\draw [style=double arrow] (53.center) to (54.center);
		\draw [style=packing] (57.center) to (58.center);
		\draw [style=thin] (9) to (61);
		\draw [style=thin] (61) to (10);
		\draw [style=double arrow] (64.center) to (61);
	\end{pgfonlayer}
\end{tikzpicture}
  \caption{A depiction of the decoding $D^{i,\ell,\sequence}$ and clause gadget $Z^\ell$ with $\sequence = (2,1,\ldots)$. The red triangle is part of the packing $\packing$. The arrows point in the direction of the deletion propagation.}
  \label{fig:fvs_modtw_decoding}
\end{figure}

\subsubsection*{Decoding gadgets.}
For every group $i \in [\ngrps]$, column $\ell \in [\ncolumns]$, and state sequence $\sequence = (h_1, \ldots, h_\grpsize) \in [5]^\grpsize$, we create a decoding gadget $D^{i, \ell, \sequence}$ consisting of $4\grpsize$ vertices $x^{i, \ell, \sequence}_1$, $\ldots$, $x^{i, \ell, \sequence}_{4\grpsize}$; a distinguished vertex $\hat{x}^{i, \ell, \sequence}$; and two vertices $y^{i, \ell, \sequence}_1$ and $y^{i, \ell, \sequence}_2$. We add the edges $\{y^{i, \ell, \sequence}_1, y^{i, \ell, \sequence}_2\}$, $\{y^{i, \ell, \sequence}_1, \hat{x}^{i, \ell, \sequence}\}$, $\{y^{i, \ell, \sequence}_2, \hat{x}^{i, \ell, \sequence}\}$ and for every $\gamma \in [4\grpsize]$, the edges $\{y^{i, \ell, \sequence}_1, x^{i, \ell, \sequence}_\gamma\}$ and $\{y^{i, \ell, \sequence}_2, x^{i, \ell, \sequence}_\gamma\}$, hence $\{y^{i, \ell, \sequence}_1, y^{i, \ell, \sequence}_2, x^{i, \ell, \sequence}_\gamma\}$ induces a triangle for every $\gamma \in [4\grpsize]$. The path gadgets $P^{i, j, \ell}$ with $j \in [\grpsize]$ are connected to $D^{i, \ell, \sequence}$ as follows. For every clique vertex $v^{i, j, \ell}_\varphi$ with $\varphi \in [5] \setminus \{h_j\}$, we pick a private vertex $x^{i, \ell, \sequence}_\gamma$, $\gamma \in [4\grpsize]$, and add an arrow from $v^{i, j, \ell}_\varphi$ to $x^{i, \ell, \sequence}_\gamma$. Since there are precisely $4\grpsize$ such $v^{i, j, \ell}_\varphi$ for fixed $i$, $\ell$, and $\sequence$, this construction works out. For every $i \in [\ngrps]$, $\ell \in [\ncolumns]$, the \emph{block} $B^{i,\ell}$ consists of the path gadgets $P^{i,j,\ell}$, $j \in [\grpsize]$, and the decoding gadgets $D^{i,\ell,\sequence}$, $\sequence \in [5]^\grpsize$. See \cref{fig:fvs_modtw_decoding} for a depiction of the decoding gadget.

\subsubsection*{Mapping truth assignments to state sequences.}
Every variable group $i \in [\ngrps]$ has at most $2^\vgrpsize$ possible truth assignments. By choice of $\grpsize$, we have that $5^\grpsize \geq 2^\vgrpsize$, hence we can fix an injective mapping $\embedding \colon \{0,1\}^\vgrpsize \rightarrow [5]^\grpsize$ that maps truth assignments $\tassign \in \{0,1\}^\vgrpsize$ to state sequences $\sequence \in [5]^\grpsize$.

\subsubsection*{Clause cycles.}
We number the clauses of $\formula$ by $C_0, \ldots, C_{\nclss - 1}$. For every column $\ell \in [\ncolumns]$, we create a cycle $Z^\ell$ consisting of $\clss 5^\grpsize$ vertices $z^\ell_\gamma$, $\gamma \in [\clss 5^\grpsize]$. Let $\ell'$ be the remainder of $\ell - 1$ modulo $\nclss$. For every group $i \in [\ngrps]$ and state sequence $\sequence \in [5]^\grpsize$, we add an arrow from $\hat{x}^{i, \ell, \sequence}$ to a private $z^\ell_\gamma$ if $\sequence \in \embedding(\{0,1\}^\vgrpsize)$ and $\embedding^{-1}(\sequence)$ is a truth assignment for variable group $i$ that satisfies clause $C_{\ell'}$. Since $\formula$ is a $\clss$-\SAT instance, every clause intersects at most $\clss$ variable groups. Every variable group has at most $2^\vgrpsize \leq 5^\grpsize$ possible truth assignments, hence $\clss 5^\grpsize$ is a sufficient number of vertices for this construction to work out. See \cref{fig:fvs_modtw_schematic} for a depiction of the high-level structure.

\begin{figure}
  \centering
  \scalebox{0.75}{\input{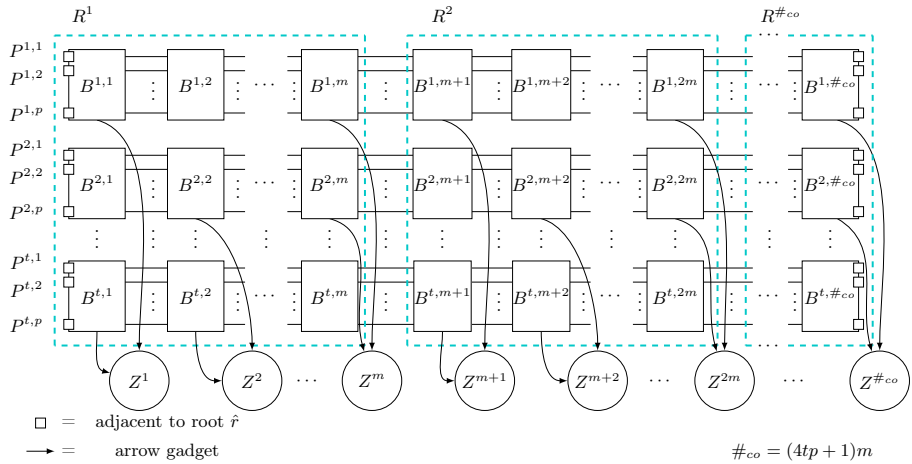}}
  \caption{The matrix structure of the constructed graph. Every $\nclss$ columns form a region.}
  \label{fig:fvs_modtw_schematic}
  \vspace*{-1cm}
\end{figure}

\subsubsection*{Packing.}
We construct a vertex-disjoint packing $\packing$ that will witness a lower bound on the size of any feedback vertex set in the constructed graph $G$. The packing $\packing$ consists of the following subgraphs:
\begin{itemize}
  \item the triangle edge between $c^{i, j, \ell}_0$ and $c^{i, j, \ell}_1$ for all $i \in [\ngrps]$, $j \in [\grpsize]$, $\ell \in [\ncolumns]$,
  \item the graph induced by the clique vertices $v^{i, j, \ell}_\varphi$, $\varphi \in [5]$, and the triangle edges between them for all $i \in [\ngrps]$, $j \in [\grpsize]$, $\ell \in [\ncolumns]$,
  \item the graph induced by the output vertices $b^{i, j, \ell}_{\varphi, \gamma}$, $\varphi \in [5]$, $\gamma \in \{1,2\}$, and the triangle edges between them for all $i \in [\ngrps]$, $j \in [\grpsize]$, $\ell \in [\ncolumns]$,
  \item the graph induced by the input vertices $u^{i, j, \ell}_1$, $u^{i, j, \ell}_2$ and the auxiliary vertices $a^{i, j, \ell}_1$, $\ldots$, $a^{i, j, \ell}_4$ and the triangle edges between them for all $i \in [\ngrps]$, $j \in [\grpsize]$, $\ell \in [\ncolumns]$,
  \item the triangle induced by $\hat{x}^{i, \ell, \sequence}, y^{i, \ell, \sequence}_1, y^{i, \ell, \sequence}_2$ for all $i \in [\ngrps]$, $\ell \in [\ncolumns]$, $\sequence \in [5]^\grpsize$,
  \item the second triangle in every arrow $\arrow(u,v)$, i.e., the triangle containing the head $v$ if the arrow was constructed from $u$ to $v$.
\end{itemize}
Observe that in the construction of $G$ at most the tail of an arrow is incident with any of the other subgraphs in $\packing$, hence the subgraphs in $\packing$ are indeed vertex-disjoint.
Let $n_\arrow$ be the number of arrows in $G$, we define 
\begin{equation*}
  \cost_\packing = (1 + 4 + 8 + 3)\ngrps \grpsize \ncolumns + \ngrps \ncolumns 5^\grpsize + n_\arrow 
\end{equation*}

\begin{lem} \label{thm:fvs_modtw_packing}
 Let $X$ be a feedback vertex set of $G$, then $|X| \geq \cost_\packing$.
\end{lem}

\begin{proof}
  We first apply the standard exchange arguments for triangle edges and arrows to $X$, obtaining a feedback vertex set $X'$ of $G$ with $|X'| \leq |X|$ that never contains the degree-2 vertex in a triangle edge and always uses the passive or active solution on any arrow.
  
  For every triangle in $\packing$, the feedback vertex set $X'$ must clearly contain at least one vertex of that triangle. Fix $i \in [\ngrps]$, $j \in [\grpsize]$, $\ell \in [\ncolumns]$ for the rest of the proof. Consider the graph induced by the clique vertices $v^{i, j, \ell}_\varphi$, $\varphi \in [5]$, and suppose that there are $\varphi \neq \psi \in [5]$ such that $v^{i, j, \ell}_\varphi, v^{i, j, \ell}_\psi \notin X'$, then the triangle edge between these two vertices is not resolved by assumption on $X'$. Hence, $X'$ contains at least four of the vertices $v^{i, j, \ell}_\varphi$, $\varphi \in [5]$. Similarly, consider the graph induced by the output vertices $b^{i, j, \ell}_{\varphi, \gamma}$, $\varphi \in [5]$, $\gamma \in \{1,2\}$, and suppose that there are $\varphi \neq \psi \in [5]$, $\gamma, \gamma' \in \{1,2\}$ such that $b^{i, j, \ell}_{\varphi, \gamma}, b^{i, j, \ell}_{\psi, \gamma'} \notin X'$, then the triangle edge between these two vertices is not resolved by assumption on $X'$. Hence, $X'$ contains at least eight of these vertices, in particular four out of five pairs $b^{i, j, \ell}_{\varphi, 1}$, $b^{i, j, \ell}_{\varphi, 2}$, $\varphi \in [5]$, must be completely contained in $X'$.
  
  It remains to show that $X'$ contains at least three vertices in the subgraph induced by the input vertices $u^{i, j, \ell}_1$, $u^{i, j, \ell}_2$ and the auxiliary vertices $a^{i, j, \ell}_1$, $\ldots$, $a^{i, j, \ell}_4$. First, observe that $X'$ has to contain all of the first three auxiliary vertices $a^{i, j, \ell}_1$, $a^{i, j, \ell}_2$, $a^{i, j, \ell}_3$ or the last auxiliary vertex $a^{i, j, \ell}_4$, otherwise there is an unresolved triangle edge incident to the last auxiliary vertex $a^{i, j, \ell}_4$. We distinguish three cases based on $\alpha = |X' \cap \{u^{i, j, \ell}_1, u^{i, j, \ell}_2\}|$. If $\alpha = 2$, we are done by the first observation. If $\alpha = 1$, there is a triangle induced by $a^{i, j, \ell}_2$, $a^{i, j, \ell}_3$, and the remaining input vertex which needs to be resolved. Hence, $a^{i, j, \ell}_2 \in X'$ or $a^{i, j, \ell}_3 \in X'$ and due to the first observation $X'$ has to contain at least one further vertex. Finally, if $\alpha = 0$, note that the graph induced by the input vertices and the first three auxiliary vertices contains a $K_{2,3}$, so $X'$ has to contain at least two of the first three auxiliary vertices and due to the first observation $X'$ has to contain at least one further vertex, hence we are done. \qed
\end{proof}

\begin{lem}\label{thm:fvs_modtw_sat_to_sol}
 If $\formula$ is satisfiable, then there is a feedback vertex set $X$ of $G$ with $|X| \leq \cost_\packing$.
\end{lem}

\begin{proof}
  Let $\tassign$ be a satisfying truth assignment of $\formula$ and let $\tassign^i$ be the induced truth assignment for variable group $i \in [\ngrps]$. Each truth assignment $\tassign^i$ corresponds to a state sequence $\embedding(\tassign^i) = \sequence^i = (h^i_1, \ldots, h^i_\grpsize)$ which we will use to construct the feedback vertex set $X$. On every path gadget $P^{i, j, \ell}$, $i \in [\ngrps]$, $j \in [\grpsize]$, $\ell \in [\ncolumns]$, we consider five different types of solutions $X^{i,j,\ell}_\varphi$, $\varphi \in [5]$, which we will define now:
  \begin{itemize}
    \item $X^{i,j,\ell}_1 = \{v^{i,j,\ell}_\varphi, b^{i,j,\ell}_{\varphi, 1}, b^{i,j,\ell}_{\varphi, 2} \sep \varphi \in [5] \setminus\{1\}\} \cup \{c^{i,j,\ell}_1\} \cup \{a^{i,j,\ell}_4\} \cup \{u^{i,j,\ell}_1, u^{i,j,\ell}_2\}$
    \item $X^{i,j,\ell}_2 = \{v^{i,j,\ell}_\varphi, b^{i,j,\ell}_{\varphi, 1}, b^{i,j,\ell}_{\varphi, 2} \sep \varphi \in [5] \setminus\{2\}\} \cup \{c^{i,j,\ell}_1\} \cup \{a^{i,j,\ell}_3, a^{i,j,\ell}_4\} \cup \{u^{i,j,\ell}_1\}$
    \item $X^{i,j,\ell}_3 = \{v^{i,j,\ell}_\varphi, b^{i,j,\ell}_{\varphi, 1}, b^{i,j,\ell}_{\varphi, 2} \sep \varphi \in [5] \setminus\{3\}\} \cup \{c^{i,j,\ell}_0\} \cup \{a^{i,j,\ell}_3, a^{i,j,\ell}_4\} \cup \{u^{i,j,\ell}_1\}$
    \item $X^{i,j,\ell}_4 = \{v^{i,j,\ell}_\varphi, b^{i,j,\ell}_{\varphi, 1}, b^{i,j,\ell}_{\varphi, 2} \sep \varphi \in [5] \setminus\{4\}\} \cup \{c^{i,j,\ell}_1\} \cup \{a^{i,j,\ell}_1, a^{i,j,\ell}_2, a^{i,j,\ell}_3\} \cup \emptyset$
    \item $X^{i,j,\ell}_5 = \{v^{i,j,\ell}_\varphi, b^{i,j,\ell}_{\varphi, 1}, b^{i,j,\ell}_{\varphi, 2} \sep \varphi \in [5] \setminus\{5\}\} \cup \{c^{i,j,\ell}_1\} \cup \{a^{i,j,\ell}_2, a^{i,j,\ell}_3, a^{i,j,\ell}_4\} \cup \emptyset$
  \end{itemize}
  The feedback vertex set on the path gadgets $P^{i,j,\ell}$ is given by 
 \begin{equation*}
   X_P = \bigcup_{i \in [\ngrps]} \bigcup_{j \in [\grpsize]} \bigcup_{\ell \in [\ncolumns]} X^{i, j, \ell}_{h^i_j}.
 \end{equation*}
 On the decoding gadgets $D^{i, \ell, \sequence}$, we define 
 \begin{equation*}
   X_D = \bigcup_{i \in [\ngrps]} \bigcup_{\ell \in [\ncolumns]} \left(\{\hat{x}^{i, \ell, \sequence^i}\} \cup \{y^{i, \ell, \sequence}_1 \sep \sequence \in [5]^\grpsize \setminus \{\sequence^i\}\} \right).
 \end{equation*}
 We obtain the desired feedback vertex set $X$ by starting with $X_P \cup X_D$ and propagating the deletions throughout $G$ using the arrows, i.e., if the tail $u$ of an arrow $\arrow(u,v)$ is in $X$, then we choose the active solution on this arrow and otherwise we choose the passive solution. Since $|X^{i, j, \ell}_\varphi| = 16$ for all $i \in [\ngrps]$, $j \in [\grpsize]$, $\ell \in [\ncolumns]$, $\varphi \in [5]$, we compute that $|X_P| = 16 \ngrps \grpsize \ncolumns$ and for $X_D$, we see that $|X_D| = \ngrps \ncolumns 5^\grpsize$ and hence $|X| = \cost_\packing$ as desired, since we perform one additional deletion per arrow.
 
 It remains to show that $X$ is a feedback vertex set of $G$, i.e., that $G - X$ is a forest. First, notice that the passive solution of an arrow $\arrow(u,v)$ disconnects $u$ from $v$ inside $\arrow(u,v)$ and that the remainder of $\arrow(u,v) - \{u,v\}$ cannot partake in any cycles. The active solution of an arrow $\arrow(u,v)$ deletes $u$ and $v$, so that the three remaining vertices of the arrow form a single connected component. Since the path gadgets $P^{i, j, \ell}$ are connected to the decoding gadgets $D^{i, \ell, \sequence}$ only via arrows and also the decoding gadgets are only connected to the clause cycles $Z^\ell$ via arrows, $X$ disconnects these three types of gadgets from each other and we can handle each type separately.
 
 We begin with the decoding gadgets $D^{i, \ell, \sequence}$, $i \in [\ngrps]$, $\ell \in [\ncolumns]$, $\sequence \in [5]^\grpsize$. Every $D^{i, \ell, \sequence}$ is in its own connected component in $G - X$, since one can only enter or leave $D^{i, \ell, \sequence}$ via an arrow. Every cycle in $D^{i, \ell, \sequence}$ intersects $y^{i, \ell, \sequence}_1$ which is in $X$ if $\sequence \neq \sequence^i$. Hence, it remains to consider the case $\sequence = \sequence^i$. In this case, $X$ contains $\hat{x}^{i, \ell, \sequence^i}$ by definition of $X_D$ and we claim that $x^{i, \ell, \sequence^i}_\gamma \in X$ for all $\gamma \in [4\grpsize]$ due to propagation via arrows. By construction of $G$, every $x^{i, \ell, \sequence^i}_\gamma$, $\gamma \in [4\grpsize]$, is the head of an arrow $\arrow(v^{i, j, \ell}_\varphi, x^{i, \ell, \sequence^i}_\gamma)$ for some $j \in [\grpsize]$ and $\varphi \in [5] \setminus \{h^i_p\}$, but every such $v^{i, j, \ell}_\varphi$ is in $X$ by definition of $X_P$. Hence, these deletions are propagated to the $x^{i, \ell, \sequence^i}_\gamma$, $\gamma \in [4\grpsize]$ and the only remaining vertices of $D^{i, \ell, \sequence^i}$ are $y^{i, \ell, \sequence^i}_1$ and $y^{i, \ell, \sequence^i}_2$ which clearly induce an acyclic graph.
 
 We continue with the clause cycles $Z^\ell$, $\ell \in [\ncolumns]$. Again, each clause cycle $Z^\ell$ is in its own connected component in $G - X$ and $Z^\ell$ consists of a single large cycle with vertices $z^\ell_\gamma$, $\gamma \in [\clss 5^\grpsize]$. We claim that $X$ propagates a deletion to at least one of these $z^\ell_\gamma$. Let $\ell'$ be the remainder of $\ell - 1$ modulo $\nclss$. Since $\tassign$ satisfies $\formula$ and in particular clause $C
 _{\ell'}$, there is some variable group $i \in [\ngrps]$ such that already $\tassign^i$ satisfies clause $C_{\ell'}$. By construction of $G$, there is an arrow $\arrow(\hat{x}^{i, \ell, \sequence^i}, z^\ell_\gamma)$ for some $\gamma \in [\clss 5^\grpsize]$ because $\embedding(\tassign^i) = \sequence^i$. By definition of $X_D$, we have that $\hat{x}^{i, \ell, \sequence^i} \in X$ and a deletion is indeed propagated to $z^\ell_\gamma$, thus resolving the clause cycle $Z^\ell$.
 
 It remains to show that there is no cycle in $G - X$ intersecting a path gadget $P^{i,j,\ell}$, $i \in [\ngrps]$, $j \in [\grpsize]$, $\ell \in [\ncolumns]$. All path gadgets are connected to each other via the root vertex $\rvertex$ and furthermore consecutive path gadgets $P^{i,j,\ell}$ and $P^{i,j,\ell + 1}$ are connected via the joins between them. We first show that there is no cycle in $G - X$ that is completely contained in a single path gadget $P^{i, j, \ell}$. It is easy to see that each $X \cap P^{i,j,\ell} = X^{i,j,\ell}_{h^i_j}$ contains at least one vertex per triangle edge in $P^{i, j, \ell}$. Any further cycle that could remain in $P^{i,j,\ell}$ can only involve the vertices $u^{i,j,\ell}_1$, $u^{i,j,\ell}_2$, $a^{i,j,\ell}_1$, $a^{i,j,\ell}_2$, and $a^{i,j,\ell}_3$. These vertices induce a $K_{2,3}$ plus the edge $\{a^{i,j,\ell}_2, a^{i,j,\ell}_3\}$ in $G$. In each $X^{i,j,\ell}_\varphi$, $\varphi \in [5]$, one side of the biclique $K_{2,3}$ is contained completely with the exception of at most one vertex and $a^{i,j,\ell}_2$ and $a^{i,j,\ell}_3$ only remain together if the other side is contained completely. Hence, no cycle remains there either. 
 
 Observe that $P^{i,j,\ell}$ is separated from any $P^{i,j,\ell'}$ with $\ell' \notin \{\ell - 1, \ell, \ell + 1\}$ in $G - (X \cup \{\rvertex\})$, because $X$ contains at least one endpoint of each triangle edge between the clique vertices $v^{i, j, \ell}_\varphi$, $\varphi \in [5]$, and the output vertices $b^{i, j, \ell}_{\varphi, \gamma}$, $\varphi \in [5]$, $\gamma \in \{1,2\}$. Hence, any cycle in $G - (X \cup \{\rvertex\})$ would have to involve two consecutive path gadgets. Furthermore, $\{u^{i,j,\ell + 1}_1, u^{i,j,\ell + 1}_2\}$ is a separator of size two between $P^{i,j,\ell}$ and $P^{i,j,\ell + 1}$ in $G - (X \cup \{\rvertex\})$, so any cycle involving both path gadgets has to contain $u^{i,j,\ell + 1}_1$ and $u^{i,j,\ell + 1}_2$. Therefore, we only have to consider the partial solutions $X^{i,j,\ell}_4 \cup X^{i,j,\ell+1}_4$ and $X^{i,j,\ell}_5 \cup X^{i,j,\ell+1}_5$ as otherwise at least one of $u^{i,j,\ell + 1}_1$ and $u^{i,j,\ell + 1}_2$ will be deleted. In both cases, the connected component of $G-X$ containing $u^{i,j,\ell + 1}_1$ and $u^{i,j,\ell + 1}_2$ induces a path on three vertices plus some pendant edges from the triangle edges. Hence, there is no cycle in $G - (X \cup \{\rvertex\})$.
 
 We are left with showing that $G - X$ contains no cycle containing the root vertex $\rvertex$. We do so by arguing that each vertex in $G - X$ has at most one path to $\rvertex$ in $G - X$. The neighbors of $\rvertex$ are the vertices $b^{i,j,\ell}_{2,1}$ and $c^{i,j,\ell}_1$ for all $i \in [\ngrps]$, $j \in [\grpsize]$, $\ell \in [\ncolumns]$. It is sufficient to show that there is no path between any of these neighbors in $G - (X \cup \{\rvertex\})$. By the same argument as in the previous paragraph, we only have to consider consecutive path gadgets $P^{i,j,\ell}$ and $P^{i,j,\ell + 1}$. By resolving the triangle edges between the clique vertices $v^{i,j,\ell}_\varphi$, $\varphi \in [5]$, and the output vertices $b^{i,j,\ell}_{\varphi, \gamma}$, $\varphi \in [5]$, $\gamma \in \{1,2\}$, all paths in $G - \rvertex$ between $b^{i,j,\ell}_{2,1}$ and $c^{i,j,\ell}_1$ are intersected by $X$. Similarly for paths in $G - \rvertex$ between $c^{i,j,\ell}_1$ and one of the vertices $c^{i,j,\ell + 1}_1$ or $b^{i,j,\ell + 1}_{2,1}$ and paths between $b^{i,j,\ell}_{2,1}$ and $b^{i,j,\ell + 1}_{2,1}$.
 
 It remains to consider paths in $G - (X \cup \{\rvertex\})$ between $b^{i,j,\ell}_{2,1}$ and $c^{i,j,\ell+1}_1$. We distinguish based on the chosen partial solution $X^{i,j,\ell}_\varphi \cup X^{i,j,\ell + 1}_\varphi$, $\varphi \in [5]$. For $\varphi \neq 3$, we see that $c^{i,j,\ell}_1 \in X$. For $\varphi = 3$, we see that $b^{i,j,\ell}_{2,1} \in X$. Hence, no such path can exist and $X$ has to be a feedback vertex set. \qed
\end{proof}

 We say that a vertex subset $X \subseteq V(G)$ is \emph{canonical} with respect to the twinclass $\{u^{i,j,\ell}_1, u^{i,j,\ell}_2\}$ if $u^{i,j,\ell}_2 \in X$ implies $u^{i,j,\ell}_1 \in X$. Since $\{u^{i,j,\ell}_1, u^{i,j,\ell}_2\}$ is a twinclass, we can always assume that we are working with a canonical subset.
 
 Given a vertex subset $X \subseteq V(G) \setminus \{\rvertex\}$ that is canonical with respect to each twinclass $\{u^{i,j,\ell}_1, u^{i,j,\ell}_2\}$, we define $\statemap_X \colon [\ngrps] \times [\grpsize] \times [\ncolumns] \rightarrow \{\two, \one_0, \one_1, \zero_0, \zero_1\}$ by
 \begin{equation*}
   \statemap_X(i,j,\ell) = 
   \begin{cases}
     \two,    & \text{if $|X \cap \{u^{i,j,\ell}_1, u^{i,j,\ell}_2\}| = 2$}, \\
     \one_0,  & \text{if $X \cap \{u^{i,j,\ell}_1, u^{i,j,\ell}_2\} = \{u^{i,j,\ell}_1\}$ and} \\
              & \text{   $u^{i,j,\ell}_2$ is not root-connected in $(P^{i,j,\ell} + \rvertex) - X$}, \\
     \one_1,  & \text{if $X \cap \{u^{i,j,\ell}_1, u^{i,j,\ell}_2\} = \{u^{i,j,\ell}_1\}$ and} \\
              & \text{   $u^{i,j,\ell}_2$ is root-connected in $(P^{i,j,\ell} + \rvertex) - X$}, \\
     \zero_0, & \text{if $X \cap \{u^{i,j,\ell}_1, u^{i,j,\ell}_2\} = \emptyset$ and} \\
              & \text{   $u^{i,j,\ell}_1$ and $u^{i,j,\ell}_2$ are not connected in $(P^{i,j,\ell} + \rvertex) - X$}, \\
     \zero_1, & \text{if $X \cap \{u^{i,j,\ell}_1, u^{i,j,\ell}_2\} = \emptyset$ and} \\
              & \text{   $u^{i,j,\ell}_1$ and $u^{i,j,\ell}_2$ are connected in $(P^{i,j,\ell} + \rvertex) - X$}.
   \end{cases}
 \end{equation*}
 Due to the assumption that $X$ is canonical, we see that $\statemap_X$ is well-defined. We remark that the meaning of the subscript is slightly different when one or no vertex of the twinclass is in $X$. We also introduce the notation $\state^1 = \two$, $\state^2 = \one_0$, $\state^3 = \one_1$, $\state^4 = \zero_0$, and $\state^5 = \zero_1$.

\begin{lem}\label{thm:fvs_modtw_sol_to_sat}
  If there is a feedback vertex set $X$ of $G$ of size $|X| \leq \cost_\packing$, then $\formula$ is satisfiable.
\end{lem}

\begin{proof}
  Due to \cref{thm:fvs_modtw_packing}, we immediately see that $|X| = \cost_\packing$ and $X \cap V(H)$ has to be a minimum feedback vertex set of $H$ for any $H \in \packing$. So, $X$ contains precisely one vertex of each triangle in $\packing$ and satisfies the \emph{packing equations} for all $i \in [\ngrps]$, $j \in [\grpsize]$, $\ell \in [\ncolumns]$:
  \begin{itemize}
  	\item $|X \cap \{v^{i,j,\ell}_\varphi \sep \varphi \in [5]\}| = 4$,
  	\item $|X \cap \{b^{i,j,\ell}_{\varphi, 1}, b^{i,j,\ell}_{\varphi,2} \sep \varphi \in [5]\}| = 8$,
  	\item $|X \cap \{u^{i,j,\ell}_1, u^{i,j,\ell}_2, a^{i,j,\ell}_1, a^{i,j,\ell}_2, a^{i,j,\ell}_3, a^{i,j,\ell}_4\}| = 3$.
	\end{itemize}
  In particular, this also implies that $X$ cannot contain the root vertex $\rvertex$.	
	
	 Furthermore, due to the standard exchange arguments for triangle edges and arrows, we can assume for any triangle edge between $u$ and $v$ that $X$ contains $u$ or $v$ and for any arrow $\arrow(u,v)$ that $X$ uses the passive solution or the active solution on $\arrow(u,v)$. Finally, we can assume that $X$ is canonical with respect to each twinclass $\{u^{i,j,\ell}_1, u^{i,j,\ell}_2\}$, i.e., $u^{i,j,\ell}_2 \in X$ implies that $u^{i,j,\ell}_1 \in X$.
  
  We begin by studying the structure of $X \cap P^{i,j,\ell}$ for any $i \in [\ngrps]$, $j \in [\grpsize]$, $\ell \in [\ncolumns]$. For fixed $i,j,\ell$, there is a unique $\varphi \in [5]$ such that $v^{i,j,\ell}_\varphi \notin X$ due to the packing equations. Hence, we must have $X \cap \{b^{i,j,\ell}_{\psi, 1}, b^{i,j,\ell}_{\psi,2} \sep \psi \in [5]\} = \{b^{i,j,\ell}_{\psi, 1}, b^{i,j,\ell}_{\psi,2} \sep \psi \in [5] \setminus \{\varphi\}\}$ due to the packing equations and the triangle edges between $v^{i,j,\ell}_\varphi$ and the output vertices $\{b^{i,j,\ell}_{\psi, 1}, b^{i,j,\ell}_{\psi,2} \sep \psi \in [5] \setminus \{\varphi\}\}$.
  
  For the left side of a path gadget $P^{i,j,\ell}$, we claim that $v^{i,j,\ell}_\varphi \notin X$ implies that $\statemap_X(i,j,\ell) = \state^{\varphi'}$ with $\varphi' \geq \varphi$. For $\varphi = 1$ there is nothing to show. One can see that $(\varphi', \varphi) \notin ([3] \times \{4,5\}) \cup (\{1\} \times \{2,3\})$ by considering the size of $X \cap \{u^{i,j,\ell}_1, u^{i,j,\ell}_2, a^{i,j,\ell}_1, a^{i,j,\ell}_2, a^{i,j,\ell}_3, a^{i,j,\ell}_4\}$ in those cases: Due to the triangle edges between the clique vertices $v^{i,j,\ell}_\psi$, $\psi \in [5]$ and auxiliary vertices $a^{i,j,\ell}_\gamma$, $\gamma \in [4]$, we see that $X$ contains at least two auxiliary vertices if $\varphi \geq 2$ and at least three if $\varphi \geq 4$. Using the packing equations, we see that this implies $|X \cap \{u^{i,j,\ell}_1, u^{i,j,\ell}_2\}| \leq 1$ if $\varphi \geq 2$ and $X \cap \{u^{i,j,\ell}_1, u^{i,j,\ell}_2\} = \emptyset$ if $\varphi \geq 4$, but the listed cases contradict this. It remains to handle the two cases $(\varphi', \varphi) = (2,3)$ and $(\varphi', \varphi) = (4,5)$. In the first case, the triangle edges between the vertex $v^{i,j,\ell}_3$ and the vertices $a^{i,j,\ell}_3$, $a^{i,j,\ell}_4$, $c^{i,j,\ell}_0$ together with the packing equations imply that $u^{i,j,\ell}_2$, $a^{i,j,\ell}_1$, $c^{i,j,\ell}_1 \notin X$, but then $\statemap_X(i,j,\ell) = \one_1 = \state^3 \neq \state^{\varphi'}$ because $u^{i,j,\ell}_2$, $a^{i,j,\ell}_1$, $c^{i,j,\ell}_1$, $\rvertex$ is a path in $(P^{i,j,\ell} + \rvertex) - X$. In the second case, the triangle edges between $v^{i,j,\ell}_5$ and the auxiliary vertices $a^{i,j,\ell}_2$, $a^{i,j,\ell}_3$, $a^{i,j,\ell}_4$ together with the packing equations imply that $u^{i,j,\ell}_1$, $u^{i,j,\ell}_2$, $a^{i,j,\ell}_1 \notin X$ and hence $\statemap_X(i,j,\ell) = \zero_1 = \state^5 \neq \state^{\varphi'}$. This proves the claim.
  
  Next, we claim that for any $i \in [\ngrps]$, $j \in [\grpsize]$, and $\ell_1, \ell_2 \in [\ncolumns]$ with $\ell_1 < \ell_2$, that the unique $\varphi_1 \in [5]$ and $\varphi_2 \in [5]$ such that $v^{i,j,\ell_1}_{\varphi_1} \notin X$ and $v^{i,j,\ell_2}_{\varphi_2} \notin X$ satisfy $\varphi_1 \geq \varphi_2$. We can assume without loss of generality that $\ell_2 = \ell_1 + 1$. By the previous arguments, we know that $b^{i,j,\ell_1}_{\varphi_1, 1}, b^{i,j,\ell_1}_{\varphi_1,2} \notin X$ and $\statemap_X(i,j,\ell_2) = \state^{\varphi'}$ with $\varphi' \geq \varphi_2$, so we are done if we can show that $\varphi_1 \geq \varphi'$. We do so by arguing that $G - X$ contains a cycle in all other cases, thus contradicting that $X$ is a feedback vertex set. If $\varphi_1 < \varphi'$ and $(\varphi_1, \varphi') \notin \{(2,3),(4,5)\}$, then $G[\{b^{i,j,\ell_1}_{\varphi_1,1}, b^{i,j,\ell_1}_{\varphi_1,2}, u^{i,j,\ell_1 + 1}_1, u^{i,j,\ell_1 + 1}_2\} \setminus X]$ simply contains a cycle. If $(\varphi_1, \varphi') = (2,3)$, then there is a cycle passing through the root $\rvertex$ in $G - X$ visiting $\rvertex$, $b^{i,j,\ell_1}_{2,1}$, $u^{i,j,\ell_1 + 1}_2$, and then uses the path to $\rvertex$ inside $(P^{i,j,\ell_1 + 1} + \rvertex) - X$ which exists due to $\statemap_X(i,j,\ell_1 + 1) = \state^{\varphi'} = \state^3 = \one_1$. If $(\varphi_1, \varphi') = (4,5)$, then there is a cycle in $G - X$ visiting $u^{i,j,\ell_1 + 1}_1$, $b^{i,j,\ell_1}_{4,1}$, $u^{i,j,\ell_1 + 1}_2$, and then uses the path between $u^{i,j,\ell_1 + 1}_2$ and $u^{i,j,\ell_1 + 1}_1$ in $(P^{i,j,\ell_1 + 1} + \rvertex) - X$ which exists due to $\statemap_X(i,j,\ell_1 + 1) = \state^{\varphi'} = \state^5 = \zero_1$. This shows the claim.
  
  We say that $X$ \emph{cheats} from $P^{i,j,\ell}$ to $P^{i,j,\ell + 1}$ if $v^{i,j,\ell}_{\varphi_1}, v^{i,j,\ell + 1}_{\varphi_2} \notin X$ with $\varphi_1 > \varphi_2$. By the previous claim, there can be at most four cheats for fixed $i$ and $j$. For $\gamma \in [\nregions]$, we define the $\gamma$-th \emph{column region} $R^\gamma = [(\gamma - 1)\nclss + 1, \gamma \nclss]$. Since there are $\ngrps \grpsize$ paths, there is a column region $R^\gamma$ that contains no cheats by the pigeonhole principle, i.e., for all $i \in [\ngrps]$, $j \in [\grpsize]$, $\ell_1, \ell_2 \in R^\gamma$, $\varphi \in [5]$, we have $v^{i,j,\ell_1}_\varphi \notin X$ if and only if $v^{i,j,\ell_2}_\varphi \notin X$. Fix this $\gamma$ for the remainder of the proof. 
  
  We obtain sequences $\sequence^i = (h^i_1, \ldots, h^i_\grpsize) \in [5]^\grpsize$, $i \in [\ngrps]$, by defining $h^i_j \in [5]$ as the unique number satisfying $v^{i,j,\gamma \nclss}_{h^i_j} \notin X$. Since $R^\gamma$ contains no cheats, note that we would obtain the same sequences if we use any column $\ell \in R^\gamma \setminus \{\gamma \nclss\}$ instead of column $\gamma \nclss$ in the definition. We obtain a truth assignment $\tassign^i$ for variable group $i$ by setting $\tassign^i = \embedding^{-1}(\sequence^i)$ if $\sequence^i \in \embedding(\{0,1\}^\vgrpsize)$ and otherwise picking an arbitrary truth assignment. 
  
  We claim that $\tassign = \tassign^1 \cup \cdots \cup \tassign^\ngrps$ is a satisfying assignment of $\formula$. To prove this claim, we begin by showing for all $i \in [\ngrps]$, $\ell \in R^\gamma$, $\sequence \in [5]^\grpsize$, that $\hat{x}^{i, \ell, \sequence} \in X$ implies $\sequence = \sequence^i$. Suppose that $\sequence = (h_1, \ldots, h_\grpsize) \neq \sequence^i$, then there is some $j \in [\grpsize]$ with $h_j \neq h^i_j$. There is an arrow from $v^{i,j,\ell}_{h^i_j} \notin X$ to some $x^{i, \ell, \sequence}_\gamma$, $\gamma \in [4\grpsize]$, but $X$ uses the passive solution on this arrow and hence $x^{i, \ell, \sequence}_\gamma \notin X$ as well, otherwise the packing equation for the second triangle in the arrow would be violated. To resolve the triangle in $D^{i, \ell, \sequence}$ induced by $\{x^{i, \ell, \sequence}_\gamma, y^{i, \ell, \sequence}_1, y^{i, \ell, \sequence}_2\}$, we must have $y^{i, \ell, \sequence}_1 \in X$ or $y^{i, \ell, \sequence}_2 \in X$. Hence, we must have $\hat{x}^{i, \ell, \sequence} \notin X$ in either case, as otherwise the packing equation for the triangle induced by $\{\hat{x}^{i, \ell, \sequence}, y^{i, \ell, \sequence}_1, y^{i, \ell, \sequence}_2\}$ would be violated. This proves the subclaim.
  
  Consider clause $C_{\ell'}$, $\ell' \in [0, \nclss - 1]$, we will argue now that $\tassign$ satisfies clause $C_{\ell'}$. The clause cycle $Z^\ell$ with $\ell = (\gamma - 1)\nclss + \ell' + 1 \in R^\gamma$ corresponds to clause $C_{\ell'}$ and since $X$ is a feedback vertex set, there exists some $z^\ell_\eta \in X \cap Z^\ell$, $\eta \in [\clss 5^\grpsize]$. By construction of $G$, there is at most one arrow incident to $z^\ell_\eta$. If there is no incident arrow, then $z^\ell_\eta$ is not contained in any of the subgraphs in the packing $\packing$ and hence $z^\ell_\eta \in X$ contradicts $|X| = \cost_\packing$. So, there is exactly one arrow incident to $z^\ell_\eta$ and by construction of $G$, this arrow comes from some $\hat{x}^{i, \ell, \sequence}$. We must have $\hat{x}^{i, \ell, \sequence} \in X$ as well, because $X$ uses the active solution on this arrow. The previous claim implies that $\sequence = \sequence^i$. Finally, such an arrow only exists, by construction, if $\embedding^{-1}(\sequence) = \embedding^{-1}(\sequence^i) = \tassign^i$ satisfies clause $C_{\ell'}$, so $\tassign$ must satisfy $C_{\ell'}$ as well. In this step we use that the definition of $\sequence^i$ is independent of the considered column in region $R^\gamma$. Since the choice of $C_{\ell'}$ was arbitrary, this shows that $\formula$ is satisfiable. \qed
\end{proof}

\begin{lem}\label{thm:fvs_modtw_bound}
  The graph $G = G(\formula, \vgrpsize)$ has $\tcpw(G) \leq \ngrps \grpsize + (4 \grpsize + 3 + \clss) 5^\grpsize + \Oh(1)$ and a path decomposition of $G^q = G / \tcpartition(G)$ of this width can be constructed in polynomial time. 
\end{lem}

\begin{proof}
  By construction, all sets $\{u_1^{i,j,\ell}, u_2^{i,j,\ell}\}$, $i \in [\ngrps]$, $j \in [\grpsize]$, $\ell \in [\ncolumns]$, are twinclasses. Let $G'$ be the graph obtained by contracting each of these twinclasses, denoting the resulting vertex by $u^{i,j,\ell}$, then $G^q$ is a subgraph of $G'$. We will show that $\tcpw(G) = \pw(G^q) \leq \pw(G') \leq \ngrps \grpsize + (4 \grpsize + 3 + \clss) 5^\grpsize + \Oh(1)$ by giving an appropriate strategy for the mixed-search-game on $G'$ and applying \cref{thm:mixed_search}.

  \begin{algorithm}
    Handling of arrows: whenever a searcher is placed on the tail $u$ of an arrow $\arrow(u,v)$, we place searchers on all vertices of $\arrow(u,v)$ and immediately afterwards remove the searchers from $\arrow(u,v) - \{u,v\}$ again\;
    Place searcher on $\rvertex$\;
    Place searchers on $u^{i,j,1}$ for all $i \in [\ngrps]$, $j \in [\grpsize]$\;
    \For{$\ell \in [\ncolumns]$}
    {
      Place searchers on all vertices of the clause cycle $Z^\ell$\;
      \For{$i \in [\ngrps]$}
      {
        \For{$\sequence \in [5]^\grpsize$}
        {
          Place searchers on all vertices of the decoding gadget $D^{i, \ell, \sequence}$\;
        }
        \For{$j \in [\grpsize]$}
        {
          Place searchers on all vertices of $P^{i,j,\ell} - \{u_1^{i,j,\ell}, u_2^{i,j,\ell}\}$\;
          Remove searcher from $u^{i,j,\ell}$ and place it on $u^{i,j,\ell + 1}$\;
          Remove searchers on $P^{i,j,\ell} - \{u_1^{i,j,\ell}, u_2^{i,j,\ell}\}$\;
        }
        \For{$\sequence \in [5]^\grpsize$}
        {
          Remove searchers on $D^{i, \ell, \sequence}$\;
        }
      }
      Remove searchers on $Z^\ell$\;
    }
    \caption{Mixed-search-strategy for $G'$}
    \label{algo:fvs_modtw_search_game}
  \end{algorithm}
  The mixed-search-strategy for $G'$ is described in \cref{algo:fvs_modtw_search_game} and the central idea is to proceed column by column and group by group in each column. The maximum number of placed searchers occurs on line $10$ and is divided into one searcher for $\rvertex$; one searcher for each $(i,j) \in [\ngrps] \times [\grpsize]$; $\clss 5^\grpsize$ searchers for the current $Z^\ell$; $(4\grpsize + 3) 5^\grpsize$ searchers for all $D^{i, \ell, \sequence}$ with the current $i$ and $\ell$; $\Oh(1)$ searchers for the current $P^{i,j,\ell}$; and $\Oh(1)$ searchers to handle an arrow $\arrow(u,v)$. Note that arrows can be handled sequentially, i.e., there will be at any point in the search-strategy at most one arrow $\arrow(u,v)$ with searchers on $\arrow(u,v) - \{u,v\}$. Furthermore, note that whenever we place a searcher on the tail $u$ of an arrow $\arrow(u,v)$, we have already placed a searcher on the head $v$ of the arrow. \qed
\end{proof}

\begin{thm}
 There is no algorithm that solves \FVS, given a path decomposition of $G^q = G/\tcpartition(G)$ of width $k$, in time $\Oh^*((5 - \eps)^k)$ for some $\eps > 0$, unless \SETH fails.
\end{thm}

\begin{proof}
  Assume that there exists an algorithm $\algo$ that solves \FVS in time $\Oh^*((5 - \eps)^k)$ for some $\eps > 0$ given a path decomposition of $G^q = G / \tcpartition(G)$ of width $k$. Given $\vgrpsize$, we define $\delta_1 < 1$ such that $(5 - \eps)^{\log_5(2)} = 2^{\delta_1}$ and $\delta_2$ such that $(5 - \eps)^{1 / \vgrpsize} = 2^{\delta_2}$. By picking $\vgrpsize$ large enough, we can ensure that $\delta = \delta_1 + \delta_2 < 1$. We will show how to solve $\clss$-\SAT  using $\algo$ in time $\Oh^*(2^{\delta \nvars})$, where $\nvars$ is the number of variables, for all $\clss$, thus contradicting \SETH.
  
  Given a $\clss$-\SAT instance $\formula$, we construct $G = G(\formula, \vgrpsize)$ and the path decomposition from \cref{thm:fvs_modtw_bound} in polynomial time, note that we have $\clss = \Oh(1)$, $\vgrpsize = \Oh(1)$ and hence $\grpsize = \Oh(1)$. We then run $\algo$ on $G$ and return its answer. This is correct by \cref{thm:fvs_modtw_sat_to_sol} and \cref{thm:fvs_modtw_sol_to_sat}. Due to \cref{thm:fvs_modtw_bound}, we have that $\tcpw(G) \leq \ngrps \grpsize + f(\clss, \grpsize)$ for some function $f(\clss, \grpsize)$ and hence we can bound the running time by
  \begin{alignat*}{6}
    \phantom{\leq} \,\, & \Oh^*\left( (5 - \eps)^{\ngrps \grpsize + f(\clss, \grpsize)} \right)
    & \,\,\leq\,\, & \Oh^*\left( (5 - \eps)^{\ngrps \grpsize} \right) 
    & \,\,\leq\,\, & \Oh^*\left( (5 - \eps)^{\lceil \frac{\nvars}{\vgrpsize} \rceil \grpsize} \right) \\
    \leq \,\, & \Oh^*\left( (5 - \eps)^{\frac{\nvars}{\vgrpsize} \grpsize} \right) 
    & \,\,\leq\,\, & \Oh^*\left( (5 - \eps)^{\frac{\nvars}{\vgrpsize} \lceil \log_5(2^\vgrpsize) \rceil} \right) 
    & \,\,\leq\,\, & \Oh^*\left( (5 - \eps)^{\frac{\nvars}{\vgrpsize} \log_5(2^\vgrpsize)} (5 - \eps)^{\frac{\nvars}{\vgrpsize}} \right)\\
    \leq \,\, & \Oh^*\left( 2^{\delta_1 \vgrpsize \frac{\nvars}{\vgrpsize}} 2^{\delta_2 \nvars} \right)
    & \,\,\leq\,\, & \Oh^*\left( 2^{(\delta_1 + \delta_2) \nvars} \right)
    & \,\,\leq\,\, & \Oh^*\left( 2^{\delta \nvars} \right),
  \end{alignat*}
  hence completing the proof.  \qed
\end{proof}

\bibliographystyle{splncs04}
\bibliography{main}

\newpage

\appendix

\section{Independent Set Parameterized by Modular-Treewidth}
\label{sec:modtw_vc_algo}

Let $G = (V,E)$ be a graph with a cost function $\cfct \colon V \rightarrow \NN \setminus \{0\}$. We show how to compute for every $\module \in \modtree(G)$ an independent set $X_\module \subseteq \module$ of $G[\module]$ of maximum cost in time $\Oh^*(2^{\modtw(G)})$ given an optimal tree decomposition of every prime node in the modular decomposition of $G$. 

\begin{lem}\label{thm:is_modular_structure}
	If $X$ is an independent set of $G$, then for every module $\module \in \modpartition(G)$ either $X \cap \module = \emptyset$ or $X \cap \module$ is a non-empty independent set of $G[\module]$. Furthermore, $X^q := \modprojection_V(X)$ is an independent set of $G^q := G^q_V = G / \modpartition(G)$.
\end{lem}

\begin{proof}
	If $G[X \cap \module]$ contains an edge, then so does $G[X]$, hence the first part is trivially true.	If $G^q[X^q]$ contains an edge $\{\qvertex, v^q_{\module'}\}$, then $\module$ and $\module'$ are adjacent modules and $X \cap \module \neq \emptyset$ and $X \cap \module' \neq \emptyset$, so $G[X]$ cannot be an independent set. \qed
\end{proof}

Proceeding bottom-up along the modular decomposition tree of $G$, we make use of \cref{thm:is_modular_structure} to compute $X_\module$ for all $\module \in \modtree(G)$. As the base case, we consider singleton modules, i.e., $\module = \{v\}$ for some $v \in V$. Clearly, $X_{\{v\}} = \{v\}$ is an independent set of maximum cost of $G[\{v\}]$ in this case. Otherwise, inductively assume that we have computed an independent set $X_\module$ of maximum cost of $G[\module]$ for all $\module \in \modpartition(G)$ and we want to compute an independent set $X_V$ of maximum cost of $G$.

\subsubsection*{Parallel and series nodes.} If $G^q$ is a parallel or series node in the modular decomposition tree, i.e., $G^q$ is an independent set or clique respectively, then we give a special algorithm to compute $X_V$ that does not use a tree decomposition. If $G^q$ is a parallel node, then we simply set $X_V = \bigcup_{\module \in \modpartition(G)} X_\module$. If $G^q$ is a series node, then any independent set may intersect at most one module $\module \in \modpartition(G)$, else the set would immediately induce an edge. Thus, we set in this case $X_V = \arg \max_{X_\module} \cfct(X_\module)$, where the maximum ranges over all $X_\module$ with $\module \in \modpartition(G)$.

\subsubsection*{Prime nodes.} If $G^q = (V^q, E^q)$ is a prime node, then we are given a tree decomposition $(\TT^q, (\bag^q_t)_{t \in V(\TT^q)})$ of $G^q$ of width at most $k$, which we can assume to be \emph{very nice} by \cref{thm:very_nice_tree_decomposition}. We perform dynamic programming along this tree decomposition. By \cref{thm:is_modular_structure}, it is natural that every module in the currently considered bag has two possible states; it can be empty (state $\zero$), or non-empty (state $\one$) and we take an independent set of maximum cost inside. Given that we have already computed the maximum independent sets $X_\module$ for each $\module \in \modpartition(G)$, we define the partial solutions of the dynamic programming as follows. 

For each node $t \in V(\TT^q)$ of the tree decomposition, we define $\dpsols_t$ as the family consisting of all $X \subseteq V_t = \modprojection_V^{-1}(V^q_t)$ such that the following properties hold for all $\module \in \modpartition(G)$:
\begin{itemize}
 \item $X \cap \module \in \{\emptyset, X_\module\}$, 
 \item if $X \cap \module \neq \emptyset$, then $X \cap \module' = \emptyset$ for all $\{\qvertex, v^q_{\module'}\} \in E(G^q_t)$.
\end{itemize}
Given a \emph{$t$-signature} $f \colon \bag^q_t \rightarrow \states := \{\zero, \one\}$, we define the subfamily $\dpsols_t[f] \subseteq \dpsols_t$ consisting of all $X \in \dpsols_t$ such that the following properties hold for all $\qvertex \in \bag^q_t$:
\begin{itemize}
 \item $f(\qvertex) = \zero$ implies that $X \cap \module = \emptyset$,
 \item $f(\qvertex) = \one$ implies that $X \cap \module = X_\module$.
\end{itemize}
For each $t \in V(\TT^q)$ and $t$-signature $f \colon \bag^q_t \rightarrow \states$, we compute $\dppoly_t[f] := \max_{X \in \dpsols_t[f]} \cfct(X)$ by dynamic programming along the tree decomposition using the following recurrences depending on the bag type of node $t$.

\subsubsection*{Leaf bag.}
The base case, where $\bag_t = \bag^q_t = \emptyset$ and $t$ is a leaf node of the tree decomposition, i.e., $t$ has no children. Here, we simply have $\dpsols_t = \emptyset$ and hence $\dppoly_t[\emptyset] = 0$.

\subsubsection*{Introduce vertex bag.}
We have that $\bag^q_t = \bag^q_s \cup \{\qvertex\}$ and $\qvertex \notin \bag^q_s$, where $s$ is the only child node of $t$. We extend every $s$-signature by one of the two possible states for $\qvertex$ and update the cost if necessary. Note that no edges incident to $\qvertex$ are introduced yet. Hence, the recurrence is given by
\begin{align*}
 \dppoly_t[f[\qvertex \mapsto \zero]] & = \dppoly_s[f], \\
 \dppoly_t[f[\qvertex \mapsto \one]]  & = \dppoly_s[f] + \cfct(X_\module),
\end{align*}
where $f$ is an $s$-signature.

\subsubsection*{Introduce edge bag.}
Let the introduced edge be denoted $\{\qvertex, v^q_{\module'}\}$. We have that $\{\qvertex, v^q_{\module'}\} \subseteq \bag^q_t = \bag^q_s$, where $s$ is the only child node of $t$. The recurrence only needs to filter all partial solutions $X$ that intersect both $\module$ and $\module'$, since these cannot be independent sets. Hence, the recurrence is given by
\begin{equation*}
 \dppoly_t[f] = [f(\qvertex) = \zero \vee f(v^q_{\module'}) = \zero]\dppoly_s[f],
\end{equation*}
where $f$ is a $t$-signature.

\subsubsection*{Forget vertex bag.}
We have that $\bag^q_t = \bag^q_s \setminus \{\qvertex\}$ and $\qvertex \in \bag^q_s$, where $s$ is the only child node of $t$. We simply try both states for the forgotten module $\module$ and take the maximum, so the recurrence is given by
\begin{equation*}
 \dppoly_t[f] = \max(\dppoly_s[f[\qvertex \mapsto \zero]], \dppoly_s[f[\qvertex \mapsto \one]]),
\end{equation*}
where $f$ is a $t$-signature.

\subsubsection*{Join bag.}
We have that $\bag^q_t = \bag^q_{s_1} = \bag^q_{s_2}$, where $s_1$ and $s_2$ are the two children of $t$. For each $t$-signature $f$, we can simply combine a best partial solution compatible with $f$ at $s_1$ with one at $s_2$, but we do have to account for overcounting in the cost. We have that $V^q_{s_1} \cap V^q_{s_2} = \bag^q_t$, so these partial solutions can only overlap in the current bag. Hence, the recurrence is given by
\begin{equation*}
 \dppoly_t[f] = \dppoly_{s_1}[f] + \dppoly_{s_2}[f] - \cfct(\modprojection_V^{-1}(f^{-1}(\one))),
\end{equation*}
where $f$ is a $t$-signature.

\subsubsection*{Lexicographic maximum independent set.}
When using this algorithm as a subroutine, we want to find an independent set $X$ that lexicographically maximizes $(\tilde{\cfct}(X), \tilde{\wfct}(X))$, where $\tilde{\cfct}\colon V \rightarrow [1, N_c]$ and $\tilde{\wfct}\colon V \rightarrow [1, N_w]$ are some given cost and weight function with maximum value $N_c$ and $N_w$ respectively. Setting $\cfct(v) = (|V|+1) N_w \tilde{\cfct}(v) + \tilde{\wfct}(v)$ for all $v \in V$, we can simulate this setting with a single cost function $\cfct$ and recover $\tilde{\wfct}(X) = \cfct(X) \mod (|V|+1)N_w$ and $\tilde{\cfct}(X) = (\cfct(X) - \tilde{\wfct}(X)) / ((|V|+1)N_w)$. Alternatively, we may augment the dynamic programming to remember which arguments in the recurrences lead to the maximum to construct the independent set $X$ and simply compute the values $\tilde{\cfct}(X)$ and $\tilde{\wfct}(X)$ directly.

\begin{thm}
	\label{thm:is_mtw_algo}
	Let $G = (V, E)$ be a graph, $\tilde{\cfct}\colon V \rightarrow [1, N_c]$ be a cost function, and $\tilde{\wfct}\colon V \rightarrow [1, N_w]$ be a weight function.
	If $N_c, N_w \leq |V|^{\Oh(1)}$, then there exists an algorithm that given a tree decomposition of width $k$ for every prime quotient graph in the modular decomposition tree of $G$, computes an independent set $X$ of $G$ lexicographically maximizing $(\tilde{\cfct}(X), \tilde{\wfct}(X))$ in time $\Oh^*(2^k)$.
\end{thm}

\begin{proof}
	We first transform $\tilde{\cfct}$ and $\tilde{\wfct}$ into a single cost function $\cfct$ as described and then run the algorithm described in this section. Note that $\cfct$ is also polynomially bounded by $|V|$. The modular decomposition tree of $G$ contains at most $2|V|$ nodes. The base case, parallel nodes, and series nodes are handled in polynomial time. For every prime node, we perform the dynamic programming along the given tree decomposition in time $\Oh^*(2^k)$. Hence, the theorem follows. \qed
\end{proof}

\newpage
\section{Problem Definitions}
\label{sec:problems}

%

\textbf{Connected Vertex Cover}
\begin{itemize}
 \item[]\hspace*{-0.5cm} \textbf{Input:} An undirected graph $G = (V, E)$, a cost function $\cfct\colon V \rightarrow \NN \setminus \{0\}$ and an integer $\budget$.
 \item[]\hspace*{-0.5cm} \textbf{Question:} Is there a set $X \subseteq V$, $\cfct(X) \leq \budget$, such that $G - X$ contains no edges and $G[X]$ is connected?
\end{itemize}

\noindent\textbf{Connected Dominating Set}
\begin{itemize}
 \item[]\hspace*{-0.5cm} \textbf{Input:} An undirected graph $G = (V, E)$, a cost function $\cfct\colon V \rightarrow \NN \setminus \{0\}$ and an integer $\budget$.
 \item[]\hspace*{-0.5cm} \textbf{Question:} Is there a set $X \subseteq V$, $\cfct(X) \leq \budget$, such that $N[X] = V$ and $G[X]$ is connected?
\end{itemize}

\noindent\textbf{(Node) Steiner Tree}
\begin{itemize}
 \item[]\hspace*{-0.5cm} \textbf{Input:} An undirected graph $G = (V, E)$, a set of terminals $\terminals \subseteq V$, a cost function $\cfct\colon V \rightarrow \NN \setminus \{0\}$ and an integer $\budget$.
 \item[]\hspace*{-0.5cm} \textbf{Question:} Is there a set $X \subseteq V$, $\cfct(X) \leq \budget$, such that $\terminals \subseteq X$ and $G[X]$ is connected?
\end{itemize}

\noindent\textbf{Feedback Vertex Set}
\begin{itemize}
 \item[]\hspace*{-0.5cm} \textbf{Input:} An undirected graph $G = (V, E)$, a cost function $\cfct\colon V \rightarrow \NN \setminus \{0\}$ and an integer $\budget$.
 \item[]\hspace*{-0.5cm} \textbf{Question:} Is there a set $X \subseteq V$, $\cfct(X) \leq \budget$, such that $G - X$ contains no cycles?
\end{itemize}

\noindent\textbf{Vertex Cover}
\begin{itemize}
 \item[]\hspace*{-0.5cm} \textbf{Input:} An undirected graph $G = (V, E)$, a cost function $\cfct\colon V \rightarrow \NN \setminus \{0\}$ and an integer $\budget$.
 \item[]\hspace*{-0.5cm} \textbf{Question:} Is there a set $X \subseteq V$, $\cfct(X) \leq \budget$, such that $G - X$ contains no edges?
\end{itemize}

\noindent\textbf{Dominating Set}
\begin{itemize}
 \item[]\hspace*{-0.5cm} \textbf{Input:} An undirected graph $G = (V, E)$, a cost function $\cfct\colon V \rightarrow \NN \setminus \{0\}$ and an integer $\budget$.
 \item[]\hspace*{-0.5cm} \textbf{Question:} Is there a set $X \subseteq V$, $\cfct(X) \leq \budget$, such that $N[X] = V$?
\end{itemize}

\noindent\textbf{Satisfiability}
\begin{itemize}
 \item[]\hspace*{-0.5cm} \textbf{Input:} A boolean formula $\formula$ in conjunctive normal form.
 \item[]\hspace*{-0.5cm} \textbf{Question:} Is there a satisfying assignment $\tassign$ for $\formula$?
\end{itemize}

\noindent\textbf{$\clss$-Satisfiability}
\begin{itemize}
 \item[]\hspace*{-0.5cm} \textbf{Input:} A boolean formula $\formula$ in conjunctive normal form with clauses of size at most $\clss$.
 \item[]\hspace*{-0.5cm} \textbf{Question:} Is there a satisfying assignment $\tassign$ for $\formula$?
\end{itemize}

\end{document}